\pdfoutput=1
\documentclass[11pt,a4paper]{article}
\usepackage[utf8]{inputenc}
\usepackage[left=2.5cm,right=2.5cm,top=2.8cm,bottom=2.8cm]{geometry}
\usepackage{graphicx}
\usepackage{amsmath,amssymb,amsthm,amscd,mathtools}
\usepackage{cite}
\usepackage{xcolor}
\usepackage{booktabs}
\usepackage{multirow}
\usepackage{graphbox}
\usepackage{jheppub}
\usepackage[all]{xy}
\usepackage{colortbl}
\usepackage{bbold}
\usepackage{url}
\usepackage{empheq}
\usepackage{algpseudocode}
\usepackage{fancyvrb}
\usepackage{here}

\usepackage{lmodern}
\usepackage[T1]{fontenc}
\usepackage{microtype}
\usepackage{bbm}

\renewcommand{\mathbf}{\mathbold}

\definecolor{green1}{HTML}{244819}
\definecolor{cyan1}{HTML}{37cdaa}
\definecolor{blue1}{HTML}{5d7ac4}
\definecolor{red1}{HTML}{921818}
\definecolor{purple1}{HTML}{53047A}
\definecolor{orange1}{HTML}{e07229}
\definecolor{yellow1}{HTML}{edcb52}
\definecolor{gr}{gray}{0.5}
\definecolor{gr1}{gray}{0.7}
\newcommand{\gr}[1]{{\color{gr}#1}}
\newcommand{\mtop}[1]{{\color{red1}#1}}
\newcommand{\mbot}[1]{{\color{purple1}#1}}
\newcommand{\yel}[1]{{\color{yellow1}#1}}

\usepackage[font=small,labelfont=bf,width=0.85\textwidth]{caption}

\providecommand{\hypersetup}[1]{}
\providecommand{\hyperref}[2][]{#2}
\providecommand{\texorpdfstring}[2]{#1}
\providecommand{\pdfbookmark}[3][]{}

\hypersetup{plainpages=false}
\hypersetup{pdfpagemode=UseOutlines}
\hypersetup{bookmarksnumbered=true}
\hypersetup{bookmarksopen=true}
\hypersetup{pdfstartview=FitH}
\hypersetup{citecolor=[rgb]{0 .4 0}}
\hypersetup{urlcolor=[rgb]{.4 0 0}}
\hypersetup{linkcolor=[rgb]{0 0 .4}}
\hypersetup{citebordercolor={.6 .9 .6}}
\hypersetup{urlbordercolor={.7 .8 1}}
\hypersetup{linkbordercolor={1 .7 .7}}
\hypersetup{pdfborder={0 0 .5}}

\makeatletter
\newlength{\apb@width}
\newcommand{\autoparbox}[2][c]{\settowidth{\apb@width}{#2}\parbox[#1]{\apb@width}{#2}}
\newcommand{\includegraphicsbox}[2][]{\autoparbox{\includegraphics[#1]{#2}}}
\makeatother

\newcommand{\namedref}[2]{\hyperref[#2]{#1~\ref*{#2}}}
\newcommand{\secref}[1]{\namedref{Section}{#1}}

\newcommand{\appref}[1]{\namedref{Appendix}{#1}}

\renewcommand{\algref}[1]{\namedref{Algorithm}{#1}}

\newcommand{\figref}[1]{\namedref{Figure}{#1}}

\newcommand{\thmref}[1]{\namedref{Theorem}{#1}}

\newcommand{\lemref}[1]{\namedref{Lemma}{#1}}

\newcommand{\propref}[1]{\namedref{Proposition}{#1}}

\newcommand{\exref}[1]{\namedref{Example}{#1}}

\newcommand{\defref}[1]{\namedref{Definition}{#1}}

\newcommand{\remref}[1]{\namedref{Remark}{#1}}

\makeatletter
\def\mr@ignsp#1 {\ifx\:#1\@empty\else #1\expandafter\mr@ignsp\fi}%
\newcommand{\multiref}[1]{\begingroup
\xdef\mr@no@sparg{\expandafter\mr@ignsp#1 \: }%
\def\mr@comma{}%
\@for\mr@refs:=\mr@no@sparg\do{\mr@comma\def\mr@comma{,\,}\ref{\mr@refs}}%
\endgroup}
\makeatother
\renewcommand{\eqref}[1]{(\multiref{#1})}

\makeatletter\newcommand*{\etal}{%
    \@ifnextchar{.}%
        {et\penalty50\ al}%
        {et\penalty50\ al.\@\xspace}%
}\makeatother
\usepackage{xspace}

\makeatletter\newcommand*{\etc}{%
    \@ifnextchar{.}%
        {etc}%
        {etc.\@\xspace}%
}\makeatother

\newcommand{\nn}{\nonumber}
\newcommand{\nln}{\nonumber\\}

\newcommand{\eq}[1]{\begin{align} #1 \end{align}}

\newcommand{\brk}[1]{(#1)}
\newcommand{\lrbrk}[1]{\left(#1\right)}
\newcommand{\bigbrk}[1]{\bigl(#1\bigr)}
\newcommand{\Bigbrk}[1]{\Bigl(#1\Bigr)}

\newcommand{\sbrk}[1]{[#1]}
\newcommand{\lrsbrk}[1]{\left[#1\right]}
\newcommand{\bigsbrk}[1]{\bigl[#1\bigr]}

\newcommand{\Bigsbrk}[1]{\Bigl[#1\Bigr]}

\newcommand{\brc}[1]{\{#1\}}
\newcommand{\lrbrc}[1]{\left\{#1\right\}}
\newcommand{\bigbrc}[1]{\bigl\{#1\bigr\}}
\newcommand{\Bigbrc}[1]{\Bigl\{#1\Bigr\}}
\newcommand{\vev}[1]{\langle #1\rangle}

\newcommand{\lrnorm}[1]{\left\lVert#1\right\rVert}

\newcommand{\supbrk}[1]{^{\brk{#1}}}

\newcommand{\ep}{\epsilon}
\newcommand{\de}{\delta}

\newcommand{\dd}{\mathrm{d}}

\newcommand{\mId}{\mathbb{1}}
\newcommand{\Integers}{\mathbb{Z}}

\newcommand{\Complex}{\mathbb{C}}
\newcommand{\Affine}{\mathbb{A}}
\newcommand{\Torus}{\mathbb{G}}

\newcommand{\defas}{:=}

\newcommand{\MI}{I}
\newcommand{\Mi}{J}
\newcommand{\subsm}[1]{_{\scriptscriptstyle #1}}
\newcommand{\mism}[1]{I\subsm{#1}}

\newcommand{\Std}{{\tt{Std}}}
\newcommand{\Ext}{{\tt{Ext}}}
\newcommand{\Mons}{{\tt{Mons}}}
\newcommand{\RStd}{{\tt{RStd}}}
\newcommand{\fun}[2]{\code{#1[}\,#2\,\code{]}}
\newcommand{\Fun}[2]{%
    \code{#1\resizebox{!}{.35cm}{[}}%
    \,#2\,%
    \code{\resizebox{!}{.35cm}{]}}%
}
\newcommand{\spec}[1]{\fun{Spec}{#1}}
\newcommand{\Spec}[1]{\Fun{Spec}{#1}}
\newcommand{\diag}[1]{\fun{Diag}{#1}}
\newcommand{\Diag}[1]{\Fun{Diag}{#1}}
\newcommand{\rowReduce}[1]{\fun{RowReduce}{#1}}
\newcommand{\RowReduce}[1]{\Fun{RowReduce}{#1}}
\newcommand{\nullSpace}[1]{\fun{NullSpace}{#1}}
\newcommand{\jordan}[1]{\fun{JordanDecomposition}{#1}}

\newcommand{\tr}{^T}
\usepackage{pifont}
\newcommand{\mzero}{\gr{\cdot}}
\newcommand{\chline}{\arrayrulecolor{gr}\hline}
\arrayrulecolor{gr}
\newcommand{\mZero}{\mathbb{0}}

\newcommand{\CC}{\mathbb{C}}
\newcommand{\ii}{\sqrt{-1}}

\newcommand{\NN}{\mathbb{N}}
\def\pd#1{\partial_{#1}}

\newcommand{\ZZ}{\mathbb{Z}}

\newcommand{\C}{\mathbb{C}}

\newcommand{\cD}{\mathcal{D}}

\newcommand{\cF}{\mathcal{F}}

\newcommand{\cI}{\mathcal{I}}
\newcommand{\cJ}{\mathcal{J}}
\newcommand{\cK}{\mathcal{K}}
\newcommand{\cL}{\mathcal{L}}
\newcommand{\cM}{\mathcal{M}}
\newcommand{\cN}{\mathcal{N}}
\newcommand{\cO}{\mathcal{O}}

\newcommand{\cR}{\mathcal{R}}

\makeatletter
\newcommand*\bigcdot{\mathpalette\bigcdot@{1}}
\newcommand*\smtimes{\mathpalette\smtimes@{.7}}
\newcommand*\bigcdot@[2]{\mathbin{\vcenter{\hbox{\scalebox{#2}{$\m@th#1\bullet$}}}}}
\newcommand*\smtimes@[2]{\mathbin{\vcenter{\hbox{\scalebox{#2}{$\m@th#1\times$}}}}}
\makeatother

\newcommand{\soft}[1]{\textsc{#1}}

\newcommand{\code}[1]{\texttt{#1}}

\def\asir#1{}
\def\ntcomment#1{}
\def\DD{{\cal D}}
\def\Hom{\mathrm{Hom}}
\def\MM{{\cal M}}
\def\mod{\mathrm{mod}}

\def\OO{{\cal O}}
\def\PP{{\cal P}}
\def\ord{\mathrm{ord}\,}
\def\qed{\hfill $\Box$}
\def\CD{{\cal D}}   
\def\CR{{\cal R}}   
\def\CM{{\cal M}}   
\def\CN{{\cal N}}   
\def\ydim{m}
\def\ypdim{m-m'}
\def\ypcodim{m'}

\def\zvalue{a^\prime}

\def\done#1{ }

\theoremstyle{plain}
\newtheorem{theorem}{Theorem}[section]
\newtheorem{example}[theorem]{Example}
\newtheorem{lemma}[theorem]{Lemma}
\newtheorem{proposition}[theorem]{Proposition}
\newtheorem{remark}[theorem]{Remark}
\newtheorem{definition}[theorem]{Definition}
\newtheorem{algorithm}[theorem]{Algorithm}

\title{Restrictions of Pfaffian Systems for Feynman Integrals}

\author[a,b]{Vsevolod Chestnov,}
\author[d,e]{Saiei J. Matsubara-Heo,}
\author[b]{Henrik J. Munch,}
\author[d]{and Nobuki Takayama}

\newcommand{\padova}{
    Dipartimento di Fisica e Astronomia, Universit\`a degli Studi di Padova
    e INFN, Sezione di Padova,
    via Marzolo 8, I-35131 Padova, Italy.
}

\newcommand{\kobe}{Department of Mathematics, Kobe University,
1-1, Rokkodai, Nada-ku, Kobe 657-8501, Japan.}

\newcommand{\kumamoto}{Faculty of Advanced Science and Technology, Kumamoto University, 2-39-1 Kurokami Chuo-ku Kumamoto
860-8555 Japan.}

\newcommand{\bologna}{
    Dipartimento di Fisica e Astronomia, Universit\`a di Bologna
    e INFN, Sezione di Bologna,
    via Irnerio 46, I-40126 Bologna, Italy.
}

\affiliation[a]{\bologna}
\affiliation[b]{\padova}
\affiliation[d]{\kobe}
\affiliation[e]{\kumamoto}

\emailAdd{vsevolod.chestnov@unibo.it}
\emailAdd{saiei@educ.kumamoto-u.ac.jp}
\emailAdd{henrikjessen.munch@studenti.unipd.it}
\emailAdd{takayama@math.kobe-u.ac.jp}

\abstract{
This work studies limits of Pfaffian systems, a class of first-order PDEs appearing in the Feynman integral calculus.
Such limits appear naturally in the context of scattering amplitudes when there is a separation of scale in a given set of kinematic variables.
We model these limits, which are often singular, via \emph{restrictions} of $\mathcal{D}$-modules.
We thereby develop two different restriction algorithms: one based on gauge transformations, and another relying on the Macaulay matrix.
These algorithms output Pfaffian systems containing fewer variables and of smaller rank.
We show that it is also possible to retain logarithmic corrections in the limiting variable.
The algorithms are showcased in many examples involving Feynman integrals and hypergeometric functions coming from GKZ systems.
This work serves as a continuation of \cite{Chestnov:2022alh}.
}

\begin{document}
\maketitle

\newpage

\section{Introduction}

The evaluation of Feynman integrals in perturbative quantum field theory is a central, albeit challenging, problem.
A widely used approach is to obtain a set of differential equations (DEQs) whose solution yields the Feynman integrals in question.
In this setting, one considers a vector $\MI$ of unknown functions, dubbed the \emph{master integrals}, which depend on a collection of kinematic variables $z = (z_1, \ldots, z_m)$, such as the masses and momenta involved in a scattering experiment.
Using the shorthand notation $\pd{i} = \frac{\pd{}}{\pd{} z_i}$, the DEQs obeyed by the master integrals take the form
\eq{
    \pd{i} \MI = P_i \cdot \MI 
    \quad \text{for} \quad
    i=1, \ldots, m \, ,
    \label{eqn:pfaffian_system_intro}
}
where each $P_i$ is an $r \times r$ matrix with entries being rational in $z$.
This is called the \emph{Pfaffian system} for $\MI$.
Using \eqref{eqn:pfaffian_system_intro}, the unknown vector $\MI$ can be analytically expressed in terms of iterated integrals \cite{bams/1183539443,Duhr:2014woa}, or computed by numerical integration \cite{Liu:2022chg,Hidding:2020ytt, Armadillo:2022ugh} for fixed values of $z$.

In practical calculations pertaining to high-energy physics, there is often a
scale separation between the $z$ variables. Namely, one variable, say $z_1$,
is \emph{small} compared to the rest (or large, in which case $1/z_1$ is
small). A severely incomplete list of examples includes the large top mass
expansion \cite{Davies:2021kex}, the small transverse momentum expansion
\cite{Alasfar:2021ppe}, the multi-Regge limit \cite{Caron-Huot:2020vlo}, 
the soft limit in gravity \cite{Parra-Martinez:2020dzs}, 
various kinematic regions in effective field theories such as the soft-collinear effective theory \cite{Becher:2014oda}, 
and the computation of boundary conditions using expansion by regions \cite{Dulat:2014mda} in the threshold \cite{DiVita:2014pza} or massless \cite{Mastrolia:2017pfy} limits.
There are also interesting limits such as threshold expansions \cite{Lee:2018ojn, Beneke:1997zp} and limits of the $x$-parameter in the Simplified Differential Equations approach \cite{Syrrakos:2023syr,Papadopoulos:2014lla}  which involve special kinematic configurations rather than scale-separated variables. We can model these
situations as a (typically singular!) limit $z_1 \to 0$ on a Pfaffian system.
Starting from \eqref{eqn:pfaffian_system_intro}, the aim of this article is
then to derive a \emph{simpler} Pfaffian system which holds in the limit (see
also \cite[Chapter 12]{haraoka2020linear} and \cite{Bytev:2022tav} for similar
approaches). By simpler, we mean that the new Pfaffian system has a smaller
rank than $r$, and it depends on one variable less, because $z_1$ decouples
from the system.

Our algorithms are rooted in the theory of \emph{restrictions} (in the sense
that $z_1$ is being ``restricted'' to $0$), which is historically tied to the theory of $\cD$-modules.
As a supplement to studying DEQs from the viewpoint of analysis, $\cD$-modules leverage the algebraic structure of differential operators.
When a particular $\cD$-module is \emph{holonomic}, meaning that it is associated to a system of DEQs with a finite number of solutions, one can take advantage of the remarkable work by M. Kashiwara \cite{kashiwara-1975} and subsequent series of works (see \cite{Hotta-Tanisaki-Takeuchi-2008} and references therein), which generalized the theory of regular connections by P. Deligne \cite{Deligne}.
This is now established as the theory of regular holonomic systems \cite{Hotta-Tanisaki-Takeuchi-2008}.
In the context of holonomic $\cD$-modules, we can construct bases of differential operators, obtain relations between basis elements, count the number of solutions to DEQs, and more.
This naturally links $\cD$-modules to Pfaffian systems for Feynman integrals.
Indeed, in this work, restrictions are studied from two complementary points of view: one at the level of Pfaffian systems, and the other at the level of $\cD$-modules.

We take our inspiration from a particularly well-studied holonomic $\cD$-module: the GKZ hypergeometric system \cite{GKZ-1989} - though, as we show, the algorithms presented here also apply beyond this case.
In the GKZ framework \cite{Nasrollahpoursamami:2016,delaCruz:2019skx,Klausen:2021yrt,Klausen:2019hrg,Klausen:2023gui,Tellander:2021xdz,Pal:2021llg,Pal:2023kgu,Ananthanarayan:2022ntm,Feng:2022kgh,Feng:2022ude,Feng:2019bdx,Zhang:2023fil,walther2022feynman,Dlapa:2023cvx,Agostini:2022cgv, Munch:2022ouq, Klemm:2019dbm,Bonisch:2020qmm}, one generalizes parametric representations of a Feynman integral to include extra variables, such that now $z = (z_1, \ldots, z_N)$ where $N \geq m$, with the inequality generally being strict.
On one hand, this means that there are more variables to manage.
On the other hand, additional structure arises which can be facilitated in the computation of Pfaffian systems.
In particular, one can immediately write a collection of (higher-order) differential operators which annihilate the generalized Feynman integral.
These operators generate the GKZ $\cD$-module.
Using the fast algorithm of \cite{Hibi-Nishiyama-Takayama-2017}, it is possible to construct a basis of operators (the so-called standard monomials) for the GKZ $\cD$-module, which translates into a basis of integrals for the Pfaffian system.
Encoding algebraic relations between the generators of the GKZ system into a special matrix, called the \emph{Macaulay matrix}, one can readily obtain the Pfaffian matrices associated to this basis too.
This strategy for computing Pfaffian systems was carried out in the manuscript \cite{Chestnov:2022alh}, and we regard this work as continuation of the latter.

After having made use of this GKZ technology, one seeks to dispense with the auxiliary $z$ variables to match with the proper Feynman integral.
In other words, one asks for a \emph{restriction} to the case $N=m$, such that the number of variables agree.
In this work, we present two strategies with different benefits that solve the restriction problem.

\begin{paragraph}{Pfaffian-level restriction}
In this strategy, we assume that a Pfaffian system \eqref{eqn:pfaffian_system_intro} is given.
For instance, the system could have been derived via integration-by-parts
\brk{IBP} following Laporta's algorithm \cite{Laporta:2001dd},
intersection theory~\cite{Mastrolia:2018uzb, Frellesvig:2019uqt, Frellesvig:2020qot,Caron-Huot:2021xqj, Caron-Huot:2021iev}, creative
telescoping \cite{vanhove-2021}, or by the Macaulay matrix method.

Suppose that we want to restrict $z_1 = 0$, and that the Pfaffian matrices $P_i$ are singular there.
The first step in our protocol is to bring the Pfaffian system into \emph{normal form}, meaning (roughly) that the Pfaffian matrix $P_1$ has a simple pole at $z_1 = 0$, and the remaining matrices $P_2, \ldots$ are finite there.
Normal form is achieved by Moser reduction, an efficient algorithm involving a series of gauge transformations that we review and give references for in \appref{subsec:Moser_Reduction}.

Next we must choose between two cases:
(i) do we seek solutions for the integrals $\MI$ that are holomorphic at $z_1=0$, or (ii) do we seek solutions of the form $I\supbrk{n}(z_2,\ldots)\log^n(z_1)$ that behave logarithmically at the origin of $z_1$?
In case (i), one can immediately write a gauge transformation - see \eqref{eq:gauge-transform} - which yields the restricted Pfaffian system in the variables $z_2, \ldots$
In case (ii), one repeats this gauge transformation procedure several times, once for each power of $\log(z_1)$.

This Pfaffian-level restriction protocol is efficient and applicable to any Pfaffian system.
While we have only tested it for restrictions to hyperplanes, namely restrictions of the form $z_1 = f(z_2,\ldots)$ where $f$ is a linear function of the remaining variables, we expect that it will be applicable to restrictions onto general hypersurfaces too.

While the whole strategy ultimately only requires gauge transformations applied to Pfaffian systems, we derive and motivate it by studying $\cD$-modules.
The derivation is partly conjectural, but we verify it in many examples.
\end{paragraph}
\begin{paragraph}{$\cD$-module-level restriction}
In this approach, we assume that a holonomic $\cD$-module is given.
One example is the GKZ $\cD$-module%
\footnote{Alternative $\cD$-modules related to Feynman integrals can also be envisaged, such as the annihilating operators coming from conformal symmetry \cite{Henn:2023tbo} and the Yangian bootstrap, see \cite{Loebbert:2022nfu} and references therein.}.

Traditional $\cD$-module restriction algorithms stem from the breakthrough work by T. Oaku \cite{Oaku-1997}.
They rely on Gr\"obner bases in non-commutative rings, which can become computationally expensive when many variables are involved.
In this work, instead, we extend the Macaulay matrix method from \cite{Chestnov:2022alh} to incorporate restriction of variables.

Our approach contains two steps.
In the first step, we guess a basis of differential operators for the restricted $\cD$-module.
This is done systematically via \algref{alg:rest_to_pt}.
In the second step, we construct a Macaulay matrix for this basis.
This matrix is polynomial in the $z$ variables, and so we can immediately fix desired values for the $z$'s without running into singularities.
The Macaulay matrix induces a linear system of equations, the solution to which is the restricted Pfaffian system (such a linear system can be swiftly solved by rational reconstruction over finite fields \cite{Peraro:2019svx,Klappert:2019emp}).
The restriction method based on the Macaulay matrix is summarized in \algref{alg:alg3}.

In comparison to the Pfaffian-level restriction protocol, the $\cD$-module level restriction does not require a Pfaffian system as input.
This is an appealing feature, because the original, un-restricted Pfaffian system can often be computationally expensive to obtain.
Furthermore, while the Pfaffian-level restriction method is partly conjectural, the $\cD$-module-level restriction protocol is proven in this text to yield the correct restriction module.

The $\cD$-module-level protocol works for restriction to general hypersurfaces, in addition to hyperplanes.
As far as we are aware, this is the first instance of such an algorithm.
We illustrate this feature in \secref{subsec:ex:F4} by computing the restriction of the Appell $F_4(x,y)$ function onto its singular locus $xy((x-y)^2 - 2(x+y)+1) = 0$.

The computer algebra software \soft{Risa/Asir} was used extensively in our $\cD$-module computations, especially the newly developed package {\tt mt\_mm}.
\soft{Risa/Asir} can be built from source by cloning the \code{git} repository
\href{https://github.com/openxm-org/OpenXM}{https://github.com/openxm-org/OpenXM}.
Source code for {\tt mt\_mm} can be found in the directory {\tt OpenXM/src/asir-contrib/packages/src/mt\_mm}.

\end{paragraph}

\begin{paragraph}{Symmetry relations}
In addition to integration-by-parts relations, some Feynman integrals are also linearly related to one another via \emph{symmetry relations}.
These symmetries stem from reparametrizations of the Feynman \emph{integrand} resulting in the same function after integration.
Although this case is not modeled as a singular limit on a Pfaffian system, we nevertheless observe that symmetry relations can be treated using the Pfaffian-level restriction protocol.
Our test case is the equal-mass limit of the Pfaffian system for the three-mass two-loop sunrise integral - see \secref{subsec: Unequal to equal mass sunrise}.
\end{paragraph}

\begin{paragraph}{Outline}

The remaining part of this text is structured as follows.

In \secref{sec:preliminaries}, we collect preliminary material on $\cD$-modules associated to Euler integrals (a class of integrals which includes Feynman integrals).
This section also defines the notions of Pfaffian systems and restrictions.
To further familiarize oneself with the language of $\cD$-modules, we strongly recommend reading the review in \appref{sec:nut}.

\secref{sec:restriction_of_a_pfaffian_system} introduces the central concept of a \emph{logarithmic integrable connection}, a particular $\cD$-module which represents the Pfaffian system we seek to restrict.
The story presented in this section leads us to the three equations \eqref{eq:gauge-transform}, \eqref{eqn:restriction_equation} and \eqref{eq:log_constraint} which together outline the Pfaffian-level restriction protocol.
Some steps in the construction are technical, so in \secref{sec:relation_to_Feynman_integrals} we summarize how to apply the method to Feynman integrals in practical terms.
In particular, we there write the holomorphic restriction \algref{alg:holo_rest} and the logarithmic restriction \algref{alg:log_rest}.

Moving on to \secref{sec:NTsec1_MMMR}, we begin our quest to formulate restriction algorithms at the level of $\cD$-modules. 
We begin by recalling the important notions of weight vectors and $b$-functions in \secref{sec:NTsec1}, as they pertain to restriction algorithms based on Gr\"obner bases.
This subsection culminates in \algref{alg:rest_to_pt}, which systematically searches for a differential operator basis for the restricted $\cD$-module.
Having laid out this groundwork, we are then ready in \secref{sec:NTsec1_MMMR} to formulate the $\cD$-module-level restriction algorithm based on the Macaulay matrix - see \algref{alg:alg3}.
Here we heavily rely on earlier results from the manuscript \cite{Chestnov:2022alh}.

\secref{sec:examples}, making up the largest body of this work, contains a plethora of examples studying restrictions on Feynman integrals and hypergeometric functions.
The goal of each example is always to derive a Pfaffian system which holds in some limit.

We give conclusions and speculate on future developments in \secref{sec:conclusion}.

A long appendix follows the main text.
We have already mentioned \appref{sec:nut} and \appref{subsec:Moser_Reduction} above.
\appref{appendix:improve_all} describes improvements of algorithms mentioned in the main text.
Finally, \appref{sec:proofs} is a collection of proofs which require technical details on connections and $\cD$-modules.
We hope that this layout makes it easier to read the paper.

\end{paragraph}

\section{Preliminaries}\label{sec:preliminaries}

In this section, we recall the framework of twisted cohomology groups for Feynman integrals and fix basic notation used throughout the paper.
We use terminology on $\cD$-modules, such as \emph{regular holonomic}, which is explained in the standard textbooks \cite{Hotta-Tanisaki-Takeuchi-2008} and \cite{Borel}.
See \cite[Appendix B]{Chestnov:2022alh} for a short introduction to holonomic $\cD$-modules (for a pedagogical presentation which is tailored to Feynman integrals, see also the recent article \cite{Henn:2023tbo}).
\appref{sec:nut} in the present manuscript gives an introduction to tensor products of modules, as they are used extensively throughout this text.

\subsection{Twisted cohomology groups as \texorpdfstring{$\cD$}{}-modules}
We study Euler integrals of the form
\eq{
    \label{f_Gamma(z)}
    f_\Gamma(z) = \int_\Gamma g(z;x)^{\beta_0} \, x_1^{-\beta_1} \cdots x_n^{-\beta_n} \, \frac{\dd x}{x}
    \quad , \quad
    \frac{\dd x}{x} \defas \frac{\dd x_1}{x_1} \wedge \cdots \wedge \frac{\dd x_n}{x_n} \, ,
}
which are integrated along a {\it twisted cycle} $\Gamma$, namely an integration contour without boundary along which the branch of the integrand is specified \cite[Chapter 3]{aomoto2011theory}.
The exponents $\beta = (\beta_0, \ldots, \beta_n) \in \Complex^{n+1}$ are complex parameters.
$g(z;x)$ is a Laurent polynomial in $x$ with monomial coefficients $z_i$:
\eq{
    \label{g(z;x)}
    g(z;x) = \sum_{i=1}^N z_i \, x^{\alpha_i} .
}
This is written in multivariate exponent notation, i.e.~given an integer vector $\alpha_i \in \Integers^n$ we set
\begin{align}
    x^{\alpha_i} \defas
    x_1^{\alpha_{i,1}}
    \cdots
    x_n^{\alpha_{i,n}},
    \label{eq:multivar-exp}
\end{align}
where $\alpha_{i,j}$ stands for the $j$-th component of the vector $\alpha_i$.
The monomial exponents $\alpha_i$ induce an $(n+1) \times N$ matrix
\eq{
    \label{A_mat}
    A = (a_1 \ldots a_N )
}
with columns defined by $a_i := (1, \alpha_i)$, under the assumption that  $\code{Span}\{a_1,\ldots,a_N\} = \ZZ^{n+1}$.
If the variables $z_i$ are all independent, the Euler integral \eqref{f_Gamma(z)} is a solution to a PDE system called the GKZ hypergeometric system \cite{GKZ-1989}.
Generalized versions of Feynman integrals are also solutions to GKZ systems \cite{Nasrollahpoursamami:2016,delaCruz:2019skx}.
Indeed, the Lee-Pomeransky representation \cite{Lee:2013hzt} takes the form \eqref{f_Gamma(z)}, but in this case the $z_i$ are subject to linear constraints determined by the topology of a Feynman diagram - some of them may be equal to unity, or equal to each other.
Therefore, we allow $z$ to vary on an affine subspace $Y = \Affine^m \subset \Affine^N$ of dimension $m \leq N$.
The coordinates $z = (z_1, \ldots, z_m)$ denote an affine parametrization of $Y$.
Physically, $Y$ is the space of kinematic variables, e.g.~the masses and momenta associated to a given scattering process.

Let $\Torus_m$ (resp. $\Affine$) be the complex torus (resp. the complex affine line)%
\footnote{
    The torus $\Torus_m$ (resp. the complex affine line $\Affine$) is equal to $\Complex^\ast := \Complex \setminus \{0\}$ (resp. $\Complex$) as a set.
}
equipped with the Zariski topology,
We define
\eq{
    X \defas
    \big\{(z,x)\in Y \times (\Torus_m)^n    \big| \, g(z;x) \neq 0 \big\}
    .
}
$\pi : X \ \to \ Y$ is then the natural projection from the space of kinematic and integration variables $(z,x)$ to the space of kinematic variables $z$.

Let $\cD_Y$ represent the ring of linear differential operators with polynomial coefficients in $z$.
It is generated by $z_1,\dots,z_m$ and $\pd{1},\dots,\pd{m}$, over $\C$, with generators subject to the commutation relations $[z_i,z_j]=[\pd{i},\pd{j}]=0$, $[\pd{i},z_j]=\delta_{ij}$.
Setting $\mathcal{O}(X)\defas \C[z_1,\dots,z_m,x_1^{\pm 1},\dots,x_n^{\pm 1},\frac{1}{g}]$, we define an action of $\cD_Y$ on
$h=h(z;x) \in \mathcal{O}(X)$ by
\begin{align}
    \frac{\pd{}}{\pd{} z_i}\bullet h &\defas
     \frac{\pd{} h}{\pd{} z_i} + \beta_0
     \left(
     \frac{1}{g(z;x)}
     \frac{\pd{} g(z;x)}{\pd{} z_i}
     \right)
     h.\label{eqn:twisted_action}
\end{align}

The collection of (relative) $k$-forms on $X$ is an $\cO(X)$-module
\eq{
    \Omega_{X/Y}^{k}:=\bigoplus_{K \subset \{1, \ldots, n\}, \, |K| = k}
    {\mathcal{O}(X)} \, \dd x^K \, .
}
This space is acted on by the covariant derivative in integration variables only:
\eq{
\label{covariant_derivative_GKZ}
    \nabla_x  \defas
    \dd_x + \beta_0 \frac{\dd_x g(z;x)}{g(z;x)} \wedge -
    \sum_{i=1}^n \beta_i \frac{\dd x_i}{x_i}  \wedge \, .
}
The $n$-th \emph{relative de Rham cohomology group} is then defined as
\eq{
\label{de_Rham_cohomology_group}
    M_A(\beta;Y) &\defas
     \Omega_{X/Y}^{n}
    \ \Big / \
    \mathrm{Im}\Bigbrk{
        \nabla_x : \, \Omega_{X/Y}^{n - 1}
        \longrightarrow
        \Omega_{X/Y}^{n}
    } \ .
}
The action of $\cD_Y$ on $\cO(X)$ defined by \eqref{eqn:twisted_action} induces an action on the quotient $M_A(\beta;Y)$, thereby endowing $M_A(\beta;Y)$ with the structure of a $\cD_Y$-module.
If $Y=\Affine^N$ then $M_A(\beta;Y)$ equals the GKZ system \cite{Schulze-Walther-2009},
but we do not assume this unless stated explicitly.

Define the \emph{holonomic rank} of a (left) $\cD_Y$-module to be the number of holomorphic solutions%
\footnote{The definition of solutions for a $\cD$-module is given in \appref{nteq:homSol}.
}
at a generic point in $Y$.
While this definition is \emph{analytic}, we also give an \emph{algebraic} definition in the next subsection.

We have the following result on $M_A(\beta;Y)$:
\begin{theorem}\label{thm:Gauss_Manin}
    $M_A(\beta;Y)$ is regular holonomic.
    If the parameter $\beta$ is generic, its holonomic rank is given by the absolute value of the Euler characteristic of a generic fiber, $\big|\chi(\pi^{-1}(z)) \big|$.
    The latter equals the number of roots of the \emph{likelihood equation \cite{huh2013maximum}}
    \eq{
        {\beta_0 \over g(z;x)}
         {\pd{} g(z;x) \over \pd{} x_i}-\frac{\beta_i}{x_i}=0
         \quad , \quad
         i=1,\dots,n \, ,
         \label{likelihood_equation}
    }
    when $\beta$ be a generic complex vector.
\end{theorem}
\noindent
See the proof in~\appref{app:Gauss_Manin}.

\begin{example}\label{ex:nbox}\rm

    Let
    \eq{
        I_{\nu_1,\ldots,\nu_5} =
        \int
        \frac{\dd^d k_1 \dd^d k_2}
        {
            \bigsbrk{ k_1^2 }^{\nu_1}
            \bigsbrk{ k_2^2 }^{\nu_2}
            \bigsbrk{ (p_1-k_1)^2 }^{\nu_3}
            \bigsbrk{ (p_{12} - k_{12})^2 }^{\nu_4}
            \bigsbrk{ (-p_{123} + k_{12})^2 }^{\nu_5}
        }
    }
    \begin{minipage}{12.5cm}
    denote the massless two-loop N-box integral family in $d = 4-2\ep$ dimensions with kinematics
    \eq{
        p_i^2 = 0
        \quad , \quad
        s = 2 p_1 \cdot p_2
        \quad , \quad
        t = 2 p_2 \cdot p_3 \, .
    }
    We used the notation $p_I = \sum_{i \in I} p_i$ in the propagators.
    \end{minipage}
    \begin{minipage}{3cm}
        \centering
        \includegraphicsbox{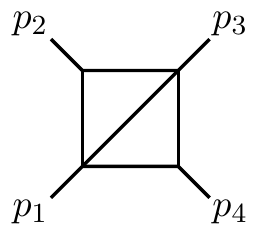}
    \end{minipage}

    Up to a prefactor, the Lee-Pomeransky representation of
    $I_{\nu_1,\ldots,\nu_5}$ can be brought into the form
    \eq{
        \int_{(0,\infty)^5}
        g(z;x)^{\ep-2} \,
        x_1^{\nu_1+\ep\delta} \cdots x_5^{\nu_5+\ep\delta}
        \frac{\dd x}{x}
        \label{eq:nbox_LP_integral}
    }
    with
    \eq{
        \nonumber
        g(z;x)
        &=
        x_1 x_2 \, +
        x_1 x_4 \, +
        x_1 x_5 \, +
        x_2 x_3 \, +
        x_2 x_4 \, \\ & +
        x_2 x_5 \, +
        x_3 x_4 \, +
        x_3 x_5 \, +
        x_1 x_2 x_4 \, +
        z x_2 x_3  x_5
        }
    and $z=t/s$.
    Thus, the space of kinematic variables is given by
    \eq{
        Y=\{(z_1,\dots,z_{10})\in\Affine^{10}\mid z_1=\dots=z_9=1\}.
    }
    The two parameters $\ep, \delta \ll 1$ are analytic regulators in the sense
    of \cite{speer1969theory,Speer1971}.

    The holonomic rank of $M_A(\beta;Y)$ is $3$, as seen counting the number of
    solutions to \eqref{likelihood_equation}, but the holonomic rank of the GKZ
    system $M_A(\beta;\Affine^{10})$ is $9$.
    \hfill$\blacksquare$
\end{example}

\subsection{Zero-dimensional \texorpdfstring{$\cR$}{}-modules and Pfaffian systems}

Let $\cR_Y$ denote the ring of differential operators with coefficients in the field of \emph{rational} functions on $Y$ (this generalizes $\cD_Y$, which only contains polynomial coefficient functions).
The field of rational functions on $Y$ is denoted by $\cK_Y$.
When there is no fear of confusion, we write $\cR$ (resp. $\cK$) instead of $\cR_Y$ (resp. $\cK_Y$).
An $\cR$-module naturally arises from a $\cD_Y$-module $\CM$ by taking a tensor product $\cR\otimes_{\cD_Y}\CM$.
Note that given two $\cD_Y$-modules $\cM_1$ and $\cM_2$, they may fail to be isomorphic as $\cD_Y$-modules even if $\cR\otimes_{\cD_Y}\cM_1$ and $\cR\otimes_{\cD_Y}\cM_2$ are isomorphic as $\cR$-modules.
However, as we will see, it is enough to deal with the $\cR$-module structure to derive a Pfaffian system.

A zero-dimensional $\cR$-module 
$\cM$ is, by definition, a left $\cR$-module whose dimension over $\cK$ is finite.
The following proposition holds \cite[Theorem 6.9.1]{dojo}:

\begin{proposition}  \label{prop:from_holonomic_to_zero_dim_ideal}
If $\cM$ is a holonomic $\cD_Y$-module, then $\cR\otimes_{\cD_Y}\cM$ is a zero-dimensional $\cR$-module.
\end{proposition}
\noindent When a $\cD_Y$-module $\CM$ is holonomic, its \emph{holonomic rank} is \emph{algebraically} defined to be $\dim_{\cK}\cR\otimes_{\cD_Y}\CM$.
\thmref{th:rest-holonomic} will show agreement between the analytic and algebraic definitions of the holonomic rank.

Given a basis $\{e_1,\dots,e_r\}$ of a zero-dimensional $\cR$-module $\cM$, there exist \emph{Pfaffian matrices}
$P_i(z)=(p_{ij}^k)_{j,k=1}^r\in \cK^{r\times r}$, $i=1,\dots,m$, such that the identities
\begin{equation}\label{eqn:conn_mat}
    \pd{i}e_j=\sum_{k=1}^rp_{ij}^ke_k
\end{equation}
hold true.
It follows that these matrices are subject to the integrability condition:
\begin{equation}\label{eqn:integrability_condition}
    \pd{i} \, P_j(z) - \pd{j} \, P_i(z)
    =
    P_i(z)\cdot P_j(z) - P_j(z) \cdot P_i(z)
    \>.
\end{equation}
Conversely, due to \eqref{eqn:conn_mat}, a set of matrices $\{P_1(z),\dots,P_m(z)\}$ satisfying the integrability condition \eqref{eqn:integrability_condition} uniquely determines a left $\cR$-module on $\cM=\oplus_{i=1}^r\cK e_i$.

For an $r$-dimensional column vector of \emph{functions} $\MI=\MI(z)$, the system of PDEs
\begin{equation}\label{eqn:pfaffian_system}
    \pd{i}\MI(z)=P_i(z) \cdot \MI(z)
\end{equation}
is called a \emph{Pfaffian system}.
In the context of Feynman integrals, $\MI$ would correspond to a basis of master integrals.

\begin{example}\rm
Let us consider the case of an ODE to give intuition on how an \emph{algebraic} basis $e$, consisting of differential operators with rational function coefficients, is related to an \emph{analytic} basis of functions $\MI$.
For the moment, we set $m=1$.
Take a univariate differential operator $D = \sum_{i=0}^{r-1} a_i(z) \pd{}^i$, where $\pd{} := \frac{\pd{}}{\pd{} z}$ and $a_i \in \cK$.
$D$ then induces an $\cR$-module $\cM_D := \cR/\cR D$.
For an unknown function $f = f(z)$, the ODE
\eq{
    D \, f = 0
}
is equivalent to the Pfaffian system
\eq{
    \pd{} I = P \cdot I
    \quad , \quad
    P =
    \lrsbrk{
        \begin{array}{ccccc}
            0&1&\cdots&\cdots&0\\
            0&0&1&\cdots&0\\
            \vdots&&&\ddots&\vdots\\
            -a_0&-a_1&\cdots&\cdots&-a_{r-1}
        \end{array}
    }
    \quad , \quad
    \MI = \sbrk{f, \pd{} f, \ldots, \pd{}^{r-1} f}\tr \, .
}
The corresponding algebraic basis is $e = \sbrk{1, \pd{}, \ldots, \pd{}^{r-1}} \subset \cM_D$.

More generally, consider any $m>0$ and a left ideal $\cJ\subset\cR$ generated by differential operators $D_1,\dots,D_\mu\in\cR$.
Assume that the quotient $\cR$-module $\cR/\cJ$ is zero-dimensional.
Then we can take a finite subset of monomials of partial derivatives $\{ e_i=\pd{}^{\alpha_i}\}_{i=1}^r$ such that its equivalence classes in $\cR/\cJ$ form a basis.
For an unknown function $f(z)$, the system of PDEs
\eq{
    D_i \, f = 0\ \ \ i=1,\dots,\mu
}
is equivalent to the Pfaffian system \eqref{eqn:pfaffian_system} with $I=\sbrk{e_1 f,\dots,e_r f}\tr$.
Let us mention that any zero-dimensional $\cR$-module is isomorphic to a quotient $\cR/\cJ$ by a left ideal $\cJ\subset\cR$.
This is proved along similar lines of \cite[Theorem 1]{leykin-cyclic}\footnote{The statement in this reference concerns $\cD$-modules, but one can write the corresponding proof for $\cR$-modules simply by replacing the length of a module by its dimension over $\cK$.}.
    \hfill$\blacksquare$
\end{example}

The following example gives a free basis for the GKZ system.
It is essential for linking an operator basis to a basis of differential forms, which in turn defines integrands for a basis of Euler integrals.
\begin{example} \rm
    Let us consider the GKZ $\cD$-module $M_A(\beta;\Affine^N)$, where we
    recall that $N$ equals to the number of monomials in \eqref{g(z;x)}.
    After choosing a term order in the ring $\cR$, a set of standard monomials ${\Std}$ is well-defined \cite[Chapter 6]{dojo}.
    Each element $e_i = \pd{}^{\alpha_i}$ of ${\Std}$ corresponds to a cohomology class represented by a monomial differential form
    \eq{
        \beta_0 \, (\beta_0-1)\dots(\beta_0-|\alpha|+1)
        \>
        \frac{x^{A \, \alpha}}{g^{|\alpha|}}
        \>
        \frac{dx}{x}
        \>,
    }
    where we set $A \, \alpha:=\sum_{i=1}^N\alpha_i \cdot a_i$.
    These cohomology classes provide a free basis of $M_A(\beta;\Affine^N)$ on a non-empty open subset of $\Affine^N$.
    A Pfaffian system for $\{e_i\}$ can e.g.~be obtained via a Gr\"obner basis computation \cite[Example 6.2.1]{dojo}.
    \hfill$\blacksquare$
\end{example}

We note that Pfaffian matrices are basis-dependent.
If we choose a different basis $\{\tilde{e}_1,\dots,\tilde{e}_r\}$ specified by an invertible matrix $G=(g_{i}^j)_{i,j=1}^r$ such that ${e}_i=\sum_{j=1}^r  g_i^j\, \tilde{e}_{j}$, then the new Pfaffian matrix is determined by a \emph{gauge transformation}:
\begin{equation}
    G[P_i]:=G^{-1}(P_iG-\pd{i}G).
    \label{eq:gauge_transform}
\end{equation}

\subsection{Secondary-like equation}

This section contains theorems which will be used in \secref{subsec: Unequal to equal mass sunrise} to extract certain symmetry relations among Feynman integrals.

For a left $\cR$-module $\cM$, we write ${\rm End}_{\cR}(\cM)$ for the set of $\cR$-morphisms $\varphi:\cM\to\cM$.
We recall that a left $\cR$-module $\cM$ is said to be {\it simple} if any of its $\cR$-submodules equal $\cM$ or $0$.
A proof of the following theorem is given in \appref{app:Schur}.
\begin{theorem}\label{thm:Schur}
    If an $\cR$-module $\cM$ is simple, then ${\rm End}_{\cR}(\cM)$ is a one-dimensional $\C$-vector space.
\end{theorem}
\noindent In \secref{subsec: Unequal to equal mass sunrise}, we use the contraposition of this theorem: if $\dim_{\C}{\rm End}_{\cR}(\cM)>1$, then the $\cR$-module $\cM$ has a submodule $\cM^\prime\subset\cM$ which is neither $\cM$ nor $0$.
It is also important to note that the converse to \thmref{thm:Schur} is true for semi-simple $\cR$-modules.
An $\cR$-module is said to be semi-simple if for any $\cR$-submodule $\cN$ of $\cM$, there exists an $\cR$-submodule $\cN^\prime$ such that $\cM=\cN\oplus\cN^\prime$.
\begin{theorem}\label{thm:semi_simple}
    When the exponent $\beta$ is generic and real, the $\cR$-module $\cR\otimes_{\cD_Y}M_A(\beta;Y)$ is a semi-simple $\cR$-module.
\end{theorem}
\begin{theorem}\label{thm:Schur2}
    Suppose that an $\cR$-module $\cM$ is semi-simple.
    Then $\cM$ is simple if and only if ${\rm End}_{\cR}(\cM)$ is a one-dimensional $\C$-vector space.
\end{theorem}
\noindent See the proofs of these theorems in appendices~\multiref{app:semi_simple, app:Schur2} respectively.

In practice, it is useful to describe the set ${\rm End}_{\cR}(\cM)$ in terms of a differential equation.
Fixing a basis $\{ e_1,\dots,e_r\}$ of $\cM$, any $\cK$-linear map $\varphi:\cM\to\cM$ is represented by an $r\times r$ matrix $\Phi$ with entries in $\cK$.
Starting from \eqref{eqn:conn_mat} and using that $\varphi(\pd{i} e_j) = \pd{i}\varphi(e_j)$, a short calculation shows that $\varphi$ is $\cR$-linear if and only if the following differential equations hold:
\begin{equation}\label{eqn:secondary_like_equation}
    \pd{i}\Phi = \Phi\cdot P_i\tr - P_i\tr\cdot\Phi
    \quad , \quad
    i=1,\dots,m
    \> ,
\end{equation}
which we call the \textit{secondary-like equation} (in analogy with the secondary equation from \cite{matsubaraheo2019algorithm}). 

\subsection{Restriction of a \texorpdfstring{$\cD$}{}-module}
\label{subsec:restriction_of_a_D-module}

Consider an affine subspace $Y^\prime\subset Y$, and let $\cO_{Y^\prime}$ denote the coordinate ring of $Y^\prime$.
More precisely, $\cO_{Y'}$ is given by the quotient $\cO_Y/\cI_{Y^\prime}$, where $\cI_{Y^\prime}$ is the ideal of $\cO_Y$ which vanishes on $Y^\prime$.
When $Y^\prime=\{ z_1=\cdots=z_{m^\prime}=0 \}$, $\cI_{Y^\prime}$ is generated by $z_1,\dots,z_{m^\prime}$, in which case $\cO_{Y^\prime}\simeq\C[z_1,\dots,z_m]/\cI_{Y^\prime}=\C[z_{m^\prime+1},\dots,z_m]$.
As a note on notation, primed symbols, such as $Y'$, will henceforth be associated to restrictions.

For a $\cD_Y$-module $\CM$, we set
\begin{equation}\label{eq:definition_of_restriction}
    \iota^*\CM:=\cO_{Y^\prime}\otimes_{\cO_Y}\CM.
\end{equation}
We write $\iota:Y^\prime\hookrightarrow Y$ for the natural inclusion.
By equation \eqref{nteq:pullback} of the appendix, $\iota^*\CM$ can also be written as
\begin{equation}
    \iota^*\CM=\frac{\CM}{z_1\CM+\ldots+z_{m^\prime}\CM}.
\end{equation}
The left action of $\cD_Y$ on $\CM$ naturally induces that of $\cD_{Y^\prime}$ on $\iota^*\cM$.
We thus have 

\begin{definition}
The \emph{restriction} of a $\cD_Y$-module $\CM$ to $Y^\prime$ is defined by the $\cD_{Y^\prime}$-module $\iota^*\cM$.
\end{definition}

Recall that the $\cD_Y$-module $M_A(\beta;Y)$ from \eqref{de_Rham_cohomology_group} characterizes the vector space structure of Feynman integrals. The following proposition thus links restrictions to our study of Feynman integrals.

\begin{proposition}
    Let $Y^\prime$ be an affine subspace of $Y$.
    Then there is an isomorphism $\iota^*M_A(\beta;Y)\simeq M_A(\beta;Y^\prime)$ of $\cD_{Y^\prime}$-modules.
\end{proposition}

\begin{proof}
    It follows immediately from the definition of $M_A(\beta;Y)$ and the base change formula of \cite[Theorem 8.4]{Borel}
\end{proof}

\begin{definition}\label{defn:rational_restriction}
The $\CR_{Y'}$-\emph{restriction} of $\cM$ to $z_1 = \ldots = z_{m'} = 0$ is defined by $\CR_{Y'} \otimes_{\CD_{Y'}} \iota^*\CM$.
If no confusion arises regarding the specification of $\cR_{Y'}$, we simply call $\CR_{Y'} \otimes_{\CD_{Y'}} \iota^*\CM$ the \emph{rational restriction} of $\cM$ to $z_1 = \ldots = z_{m'} = 0$.
\end{definition}

\noindent Pfaffian matrices for rational restrictions at the level of $\cD$-modules
can be obtained algorithmically as follows.
First, compute a basis $e$ of the rational restriction 
$\CR_{Y'} \otimes_{\CD_{Y'}} \iota^*\CM$  as a vector space over 
the rational function field $\cK_{Y'}$ on $Y'$.
Second, express $\pd{i} e$ as a linear combination of the basis $e$.
The coefficients of $e$ in this linear combination gives the Pfaffian matrix.
For more details, see \exref{ex:4-dim-sol}, \thmref{th:th5}, \thmref{th:th6}
and \algref{alg:alg3}.

Next we relate restrictions to solution spaces of PDEs.
Elements of $\cD_Y$ act on the holomorphic functions of any open set $U \subset Y$ by $\pd{i} \bullet f=\frac{\pd{} f}{\pd{} z_i}$.
We denote by $\OO^{an}(U)$%
\footnote{The superscript $an$ stands for analytic; see the discussion in~\remref{rmk:history} for clarification.}
the set of the holomorphic functions on $U$.
$\OO^{an}(U)$ is equipped with a left $\CD_Y$-module structure by this action.
Let $\cI$ be a left ideal of $\cD_Y$.
It is well-known that when $\CM=\CD_Y/\cI$ is holonomic,
i.e.~when $\cI$ is a holonomic ideal,
then the space of holomorphic solutions on $U$
\begin{equation}
    \Hom_{\CD_Y}(\CM,\OO^{an}(U))
    =
    \lrbrc{
        f \in \OO^{an}(U)
        \bigm|
        \ell \bullet f = 0 \mbox{ for all $\ell \in \cI$}
    }
    \label{eq:hom-sol}
\end{equation}
is a finite-dimensional vector space over $\CC$ (see \appref{sec:nut}).
Let $c$ be a point in $Y$.
We denote by $B(c,R)$ the open ball with the center $c$ and radius $R$.
There exists a number $R_0>0$ such that
$\mathrm{dim}_\CC \, \Hom_{\cD_Y}(\CM, \OO^{an}(B(c,R)))$ is
a constant for any $R \leq R_0$.
We call it the space of the holomorphic solutions of $\CM$ at $z=c$.
The following result is fundamental.

\begin{theorem} {\rm See the expositions \cite[Chapter 6]{dojo}, \cite{coutinho}, \cite[Section 5.2]{SST}.}
\label{th:rest-holonomic}
\begin{enumerate}
\item The left $\CD_{Y^\prime}$-module $\iota^*\CM$ is a holonomic $\cD_{Y^\prime}$-module.
\item
Let $\CM$ be regular holonomic and let $c\in Y^\prime$ be a point outside of the singular locus of $\iota^*\CM$.
Then the holonomic rank of $\iota^*\CM$ agrees with the dimension
of the space of holomorphic solutions $\mathrm{dim}_\CC \, \Hom_{\cD_{Y}}(\CM,\OO^{an}(B(c,R_0)))$ when $R_0>0$ is sufficiently small.
In particular, if $\CM$ takes the form $\CM=\cD_Y/\cI$ for some left ideal $\cI\subset\cD_Y$, we have
\begin{equation} \label{eq:sol_rank}
    \mathrm{dim}_\CC\, \Hom_{\cD_Y}\Bigbrk{\CM, \OO^{an}\bigbrk{B(c,R_0)}}
    =
    \mathrm{dim}_\CC\, \frac{\cD_Y}{\cI + (z_1-c_1)\cD_Y + \ldots + (z_m-c_m)\cD_Y}
\end{equation}
by setting $\ypcodim=\ydim$. In other words, we set $Y^\prime = \{ c\}$.
\item If $c\in Y$ does not belong to the singular locus of $\CM$, $\mathrm{dim}_\CC\, \Hom_{\cD_Y}(\CM, \OO^{an}(B(c,R_0)))$ is equal to the holonomic rank of $\CM$.
\end{enumerate}
\end{theorem}

Let $\CD_{Y^\prime}^k := \overbrace{\CD_{Y^\prime} \times \cdots \times \CD_{Y^\prime}}^k$ be a free $\cD_{Y'}$-module.
It is well-known that the restriction $\iota^* \CM$ can be expressed as
$\iota^* \CM=\cD^k_{Y^\prime}/\cJ$ for some $k$,
where
$\cJ$ is a left $\CD_{Y^\prime}$-submodule of $\cD_{Y'}^k$ \cite{dojo, coutinho}.
The restriction can also be expressed in terms of a G\"obner basis as in \eqref{eq:rest-by-gb2}.
The holonomic rank of $\iota^* \CM$ can be determined by computing a Gr\"obner basis of $\CR_{Y^\prime}\cJ$ in a free left $\CR_{Y^\prime}$-module $\CR_{Y^\prime}^k$%
\footnote{For introductions to Gr\"obner bases of a submodule in a free module,
see \cite[Section 3.5]{dojo}, \cite[Chapter 3]{adams}, and \appref{sec:nut}.
The first two references discuss modules over a polynomial ring, but it is easy to
generalize their exposition to left modules over $\CD_{Y^\prime}$ or $\CR_{Y^\prime}$.}.

For a given holonomic $\CD_{Y^\prime}$-module $\CD_{Y'}^k/\cJ$,
there exists a left ideal $\cI$ of $\CD_{Y'}$ such that
$\CD_{Y'}/\cI$ is isomorphic to $\CD_{Y'}^k/\cJ$ as a $\CD_{Y'}$-module.
An algorithmic construction of $\cI$, relying on the so-called \emph{cyclic vector}, is given in \cite{leykin-cyclic}.
Our \remref{rmk:cyclic} gives an easy, albeit inexact, derivation.
The cyclic vector also allows for the construction of an ideal associated to a Pfaffian system.

\section{Restriction of a Pfaffian system}\label{sec:restriction_of_a_pfaffian_system}

In this section, we present our algorithm for the restriction of a generic Pfaffian
system.
In~\secref{ssec:integrable-connection}, we revise key points in the theory
of integrable connections.
Then in~\secref{subsec:Deligne_extension} and~\secref{ssec:restriction-solutions}, we
show the restriction procedure for a logarithmic
connection and the corresponding solution of the Pfaffian system.
Finally, in~\secref{sec:relation_to_Feynman_integrals}, we apply the developed theory to
the Pfaffian systems coming from Feynman integrals.

\subsection{Integrable connection}
\label{ssec:integrable-connection}
The arguments we present in this section are Zariski local, meaning that we always work on a Zariski open subset $Y_0$ of the total space $Y$.
Let $\cO_{Y_0}$ denote the ring of regular functions on $Y_0$.
When $Y_0$ is a complement of a hypersurface $\{ f=0\}$ defined by $f\in\cO_Y\setminus \{0\}$, $\cO_{Y_0}$ is given by $\CC[z_1,\dots,z_m,\frac{1}{f}]$.
We set $\cD_{Y_0}:=\cO_{Y_0}\otimes_{\cO_Y}\cD_Y$.
The ring structure of $\cD_Y$ induces that of $\cD_{Y_0}$.
An \emph{integrable connection} $\CM$ on $Y_0$ is a free $\cD_{Y_0}$-module of finite rank as an $\cO_{Y_0}$-module\footnote{
    In the standard literature, an integrable connection is usually defined as
    a $\cD_Y$-module which is locally free as an $\cO_Y$-module. However, a
    locally free $\cO_Y$-module is a free module when it is restricted to an
    open subset. As we are only interested in the structure of $\cR$-modules
    in this paper, the global structure of a locally free module does not play
    any role. Hence, our definition of an integrable connection does not
    involve locally free modules.
}.
More concretely, $\CM$ is isomorphic to $\cO_{Y_0}^r$ as an $\cO_{Y_0}$-module
for some integer $r$. The action of $\pd{i}$ (for $i=1, \ldots, m$) is
determined by an $r\times r$ matrix $P_i(z)$ with entries in $\cO_{Y_0}$ according to the
formula \eqref{eqn:conn_mat}, where $\{e_1,\dots,e_r\}$ is any $\cO_{Y_0}$-free
basis. The integrability condition \eqref{eqn:integrability_condition} follows straightforwardly for the
matrices $P_i\brk{z}$ obtained in this way.
Conversely, any set of matrices $\{P_i\}_{i=1}^r$ with
entries in $\cO_{Y_0}$ satisfying the integrability condition
\eqref{eqn:integrability_condition} gives rise to a connection
\eq{
    \nabla_i = \pd{i} - P_i
}
by the formula \eqref{eqn:conn_mat}.
In this representation, the integrability condition translates into the
commutator relation $[\nabla_i, \nabla_j] = 0\>$.

Any holonomic $\cD_Y$-module is an integrable connection on a non-empty open
subset $Y_0\>$. For a $\cD_Y$-module $\CM\>$, we define a left
$\cD_{Y_0}$-module $\CM|_{Y_0}$ by $\CM|_{Y_0} \defas \cO_{Y_0}\otimes_{\cO_Y}\CM$.
As the following proposition shows, studying the $\cR$-module structure of a
$\cD_Y$-module is the same as restricting it onto an open subset wherein it is an
integrable connection.

\begin{proposition}
  Let $\CM_1,\CM_2$ be a pair of $\cD_Y$-modules.
  $\CM_1$ and $\CM_2$ are isomorphic as $\cR$-modules if and only if there is a non-empty open subset $Y_0 \subset Y$ such that $\CM_1|_{Y_0}$ and $\CM_2|_{Y_0}$ are both integrable connections and isomorphic as $\cD_{Y_0}$-modules.
\end{proposition}

Next, we recall the definition of a logarithmic connection. To simplify the
exposition, let $Y^\prime$ be an affine hyperplane of $Y$. By an affine linear
change of coordinates, we may assume that $Y^\prime = \{z_1=0\}\>$. We write
$\cD_Y(\log Y^\prime)$ for the subring of $\cD_Y$ generated as
\begin{align}
    \cD_Y(\log Y^\prime) \defas
    \CC\vev{z_1,\dots,z_m, \, z_1\pd{1} \, ,\pd{2},\dots,\pd{m}}
    \>.
\end{align}
The use of the symbol $\log Y^\prime$ is justified because \eqref{eqn:conn_mat1} below has a logarithmic singularity along $Y^\prime$.
We further introduce $Y_0^\prime:=Y_0\cap Y^\prime$ and assume
$Y_0^\prime\neq\varnothing\>$. The relations between the sets $Y\>$,
$Y^\prime\>$, $Y_0\>$, and $Y_0^\prime$ are shown in \figref{fig:Ven}.
\begin{figure}
    \centering
    \includegraphicsbox{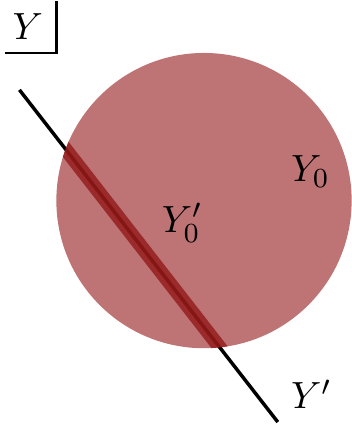}
    \caption{The different spaces $Y,Y^\prime,Y_0$ and $Y_0^\prime$ wherein the $z$-variables live. The coordinates $\{z_1,\ldots,z_m\}$ and $\{z_{m'+1}, \ldots, z_m\}$ parameterize $Y$ and $Y'$ respectively. $Y_0$ is an open subset in $Y$ containing $Y'$. The intersection $Y_0' = Y_0 \cap Y'$ is shown as a dark red line.}
    \label{fig:Ven}
\end{figure}

Next we define
\eq{
    \cD_{Y_0}(\log Y^\prime):=\cO_{Y_0}\otimes_{\cO_Y}\cD_Y(\log Y^\prime)
    \>.
}
The ring structure of $\cD_Y(\log Y^\prime)$ induces that of $\cD_{Y_0}(\log Y^\prime)$.
\begin{definition}\rm
    \label{def:logarithmic_connection}
    A \emph{logarithmic connection} $\CM$ on $Y_0$ along $Y^\prime$ is a 
    $\cD_{Y_0}(\log Y^\prime)$-module that is free as an $\cO_{Y_0}$-module.
    More concretely, $\CM$ is isomorphic to $\cO_{Y_0}^r$ as an $\cO_{Y_0}$-module
    for some integer $r$, and the action of $\pd{i}$ is given by an $r\times
    r$ matrix $P_i(z)$ with entries in $\cO_{Y_0}$ by the formulas
    \begin{align}
        z_1 \pd{1} \, e_j
        &= \sum_{k=1}^r p_{1j}^k(z) \, e_k
        \label{eqn:conn_mat1}
        \\
        \pd{i} \, e_j
        &= \sum_{k=1}^r p_{ij}^k(z) \, e_k
        \> , \quad i=2,\dots,m
        \label{eqn:conn_mat2}
    \end{align}
    where $\{e_1,\dots,e_r\}$ is some $\cO_{Y_0}$-free basis.
\end{definition}
Setting
\eq{
    z' = (z_2,\ldots,z_m)
    \>,
}
the residue matrix of $\CM$ along $Y^\prime$ is given by
$P_1\brk{0, z^\prime} = \bigbrk{p_{1j}^k\brk{0, z^\prime}}_{j, k = 1}^r\>$.
The residue matrix with respect to a different free basis is obtained by a conjugation of $P_1(0, z^\prime)\>$.
Therefore, matrix multiplication by $P_1(0,z^\prime)$ from the left is a
well-defined linear operator on
$\cO_{Y_0^\prime}\otimes_{\cO_{Y_0}}\CM \simeq \cO_{Y_0^\prime}^r$, as it does
not depend on a particular choice of the free basis.
Given a logarithmic connection $\CM$ on $Y_0$ along $Y^\prime$ we associate to
it a $\cD_Y$-module $\cD_Y\otimes_{\cD_Y(\log Y^\prime)} \cM\>$, which is again
denoted by the symbol $\CM$ when there is no fear of confusion.
As an $\cR$-module, $\CM$ corresponds to a Pfaffian system \eqref{eqn:pfaffian_system} with
\begin{align}
  P_1 \ \ \text{replaced by}\ \ \frac{1}{z_1}P_1(z) \, ,
  \label{eqn:connection_matrix_log}
\end{align}
where the factor of $1/z_1$ comes from \eqref{eqn:conn_mat1}.

An integrable connection $\CM$ on $Y_0\setminus Y^\prime$ is said to be
\emph{regular} along $Y^\prime$ if there exists a logarithmic connection
$\widetilde{\CM}$ on $Y_0$ along $Y^\prime$ whose restriction onto $Y_0\setminus Y^\prime$ is $\CM$.
This terminology is compatible with the regular holonomicity of a
$\cD_Y$-module in the following sense: if $\CM$ is a regular holonomic
$\cD_Y$-module and $\CM|_{Y_0\setminus Y^\prime}$ is an integrable connection,
then $\CM$ is regular.  Although the notion of regularity is well-defined,
$\widetilde{\CM}$ is not uniquely determined by $\CM$.
Let us finally mention
that the regularity of $\cM$ is equivalent to the corresponding Pfaffian
matrices taking the form \eqref{eqn:connection_matrix_log} after a gauge
transformation.

\subsection{Restriction of a logarithmic connection in normal form}
\label{subsec:Deligne_extension}

A matrix is called \emph{resonant}%
\footnote{
    This is different from resonance in GKZ systems.
}
when it has two distinct eigenvalues with an integral difference.
When $\CM$ is a regular connection along $Y^\prime$ defined on $Y_0\setminus
Y^\prime$, there exists a logarithmic connection $\widetilde{\CM}$ along
$Y^\prime$ on $Y_0$ such that
\begin{enumerate}
    \item The restriction of $\widetilde{\CM}$ onto $Y_0\setminus Y^\prime$ is $\CM$,
    \item The residue matrix $P_1(0,z^\prime)$ is non-resonant and the only integral eigenvalue, if it exists, is $0$.
\end{enumerate}
We say that such an $\widetilde{\CM}$ is in \emph{normal form}.
If $Y' = \{z_1 = 0\}$, the normal form Pfaffian matrices associated to
$\widetilde{\CM}$ will expand as
\eq{
    \frac{P_1}{z_1} &= \sum_{n=-1}^\infty P_{1,n}(z') \, z_1^n \\
    P_i &= \sum_{n= 0}^\infty P_{i,n}(z') \, z_1^n
    \quad , \quad
    i = 2, \ldots, m \, .
}
In other words, $P_1$ has a logarithmic singularity on $Y'$, and the other
Pfaffian matrices are finite there \brk{see~\appref{subsec:Moser_Reduction} for
more details and proofs}.

The existence of a normal form of a logarithmic connection is well-known \cite[Proposition 5.4]{Deligne}.
Working at the level of Pfaffian matrices, normal form can be found via the
standard Moser reduction algorithms of \cite{Barkatou-Jaroschek-Maddah-2017, Barkatou-1997},
which we review in \appref{subsec:Moser_Reduction}.

The next proposition follows from~\defref{def:logarithmic_connection}:
\begin{proposition}\label{prop:logarithmic_connection_as_D_module}
    Let $\CM$ be a logarithmic connection along $Y^\prime$ on $Y_0$.
    Let $\{e_1,\dots,e_r\}$ be a free basis of $\CM$ and let $P_i(z)$ be the matrices determined by \eqref{eqn:conn_mat1} and \eqref{eqn:conn_mat2}.
    Let $\cN$ be the left $\cD_{Y_0}$-submodule of $\cD^r_{Y_0}$ generated by the row vectors of
    $\mId \, z_1\pd{1} - P_1(z)$ and
    $\mId \, \pd{i} - P_i(z)$
    for $i=2,\dots,m$, with $\mId$ being the identity matrix.
    Then the left $\cD_{Y_0}$-module $\CM$ is isomorphic to $\cD_{Y_0}^r/\cN$.
\end{proposition}
Let $\CM$ be a logarithmic connection on $Y_0$ along $Y^\prime$, and suppose non-resonance, so that none of the eigenvalues of the residue matrix $P_1(0,z^\prime)$ is a
positive integer. The restriction $\iota^*\CM$ as an integrable
connection can be calculated as follows\footnote{
    Recall that $\iota$ is the natural inclusion introduced
    in~\secref{subsec:restriction_of_a_D-module}.
}.
We fix a free basis $\{e_1,\dots,e_r\}$ of $\CM$ and identify it
with $\cO_{Y_0}^r$ as an $\cO_{Y_0}$-module.
We set
\eq{
    {\CM}^\prime:=\cO_{Y_0^\prime}^r \,/\, \cN^\prime
    \>,
}
where $\cN^\prime$ denotes an
$\cO_{Y_0^\prime}$-submodule of $\cO_{Y_0^\prime}^r$ spanned by the row vectors of
$P_1(0,z^\prime)\>$.
From the integrability condition
\eqref{eqn:integrability_condition}, we have
\eq{
  \pd{i} \, P_1(0,z^\prime)
  =
  \bigsbrk{P_{i}(0,z^\prime), \, P_1(0,z^\prime)}
  \> , \quad
  \text{for $i = 2,\ldots,m$}
  \>.
  \label{eq:integrability_condition_rest}
}
It follows easily that the action of $\cD_{Y_0^\prime}$ on $\cO_{Y_0^\prime}^r$ given by
\eq{
  \pd{i} \, e_j
  =
  \sum_{k=1}^r p_{ij}^k(0,z^\prime) \, e_k
  \label{eqn:action_on_free_basis}
}
induces an action on the quotient ${\CM}^\prime\>$.

\begin{proposition}\label{prop:logarithmic_restriction}
    The restriction $\iota^*\CM$ is isomorphic to the integrable connection $\CM^\prime$ as a $\cD_{Y_0^\prime}$-module.
\end{proposition}
\noindent See the proof in~\appref{app:logarithmic_restriction}.

As the rank of the residue matrix $P_1(0,z^\prime)$ is constant on
$Y_0^\prime\>$, ${\CM}^\prime$ is free as an $\cO_{Y_0^\prime}$-module (with $Y_0^\prime$ being replaced by a smaller Zariski open subset if necessary),
so we may denote by $Q_i\brk{z^\prime}$ the Pfaffian matrices associated to this integrable connection.
In what follows, we give a procedure for finding the Pfaffian matrices $Q_i(z')$ associated to the restriction module $\CM'$.

Let $r^\prime$ be the holonomic rank of $\CM^\prime\>$.
For an element $m$ of $\cO_{Y_0^\prime}^r$, we write $[{m}]$ for the equivalence class in $\cM^\prime$.
We suppose the existence of a special, free basis $\{f_1,\dots, f_r\}$ for
$\cO_{Y_0^\prime}^r$ with associated Pfaffian matrices $\widetilde{Q}_2(z'), \ldots, \widetilde{Q}_m(z')$.
This basis is assumed to have the following special property: $\{ f_{r^\prime+1},\dots,f_r\}$ is a basis of
$\cN^\prime$.
The Pfaffian matrix $Q_i(z^\prime)$ of $\CM^\prime$ with respect to the basis
$\{[{f}_1],\dots,[{f}_{r^\prime}]\}$ is then given by the first
$r^\prime\times r^\prime$ block of the matrix $\widetilde{Q}_i(z')$.

In general, our given basis $e_i$ with associated Pfaffian matrices $P_2(0,z'), \ldots, P_m(0,z')$ will not automatically yield this block matrix structure.
In other words, we know $e_i$ and $P_i(0,z')$ but we do not know the special basis $f_i$ and $\widetilde{Q}_i(z')$.
Let us thus proceed to construct the gauge transformation mapping $e_i$ to the special basis $f_i$.

Given our free basis $\{ e_1,\dots, e_r\}$ of $\cO_{Y_0^\prime}^r$, we can always
find an $r^\prime \times r$ matrix $B=(b_{ij})$ with entries in $\cO_{Y_0^\prime}\>$ such that
$\{ [f_1],\dots,[f_{r^\prime}]\}$ with $f_i=\sum_{j=1}^rb_{ij}e_j$ forms a free basis of $\CM^\prime\>$.
Then, by the definition of $\CM^\prime$, there are exactly $r-r^\prime$
linearly independent rows of the residue matrix $P_1(0,z^\prime)\>$.
Let $R$ be the matrix containing these independent rows.
Writing $R = (r_{ij})_{i = r'+1, \ldots, r}^{j = 1, \ldots, r}$, then $f_{i}=\sum_{j=1}^rr_{ij}e_j$ forms a free basis of $\cN^\prime$.
Aligning the rows of $B$ and $R$ produces an invertible $r\times r$ matrix $M$:
\eq{
    M =
    \lrsbrk{
        \begin{array}{c}
            \mtop{B} \\
            \chline
            \mbot{R}
        \end{array}
    }
    =
    \vcenter{\hbox{
        \includegraphicsbox{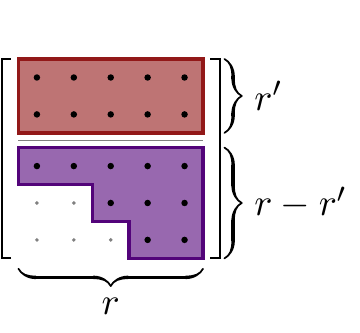}
    }}
    \>.
    \label{eq:M-def}
}
\noindent
In vector notation, we now have $f = M \cdot e$.

Borrowing notation from \eqref{eq:gauge_transform}, we thus obtain the Pfaffian matrix $\widetilde{Q}_i(z^\prime)$ associated to $\{f_1, \ldots, f_r\}$ from the known matrices $P_i(0,z')$ via a gauge transformation by $M^{-1}$:
\eq{
    \widetilde{Q}_i\brk{z^\prime}
    =
    \brk{M^{-1}}\bigsbrk{P_i\brk{0, z^\prime}}
    =
    \bigbrk{
        M \cdot P_i\brk{0,z^\prime}
        +
        \pd{i} \, M
    } \cdot M^{-1}
    \>,
    \quad i = 2, \ldots, m.
  \label{eqn:L}
}
The sought-after Pfaffian matrix $Q_i(z')$ for the restriction ideal $\cM'$, written in the basis $B$, is finally obtained from the first $r' \times r'$ block of $\widetilde{Q}_i(z')$:
\begin{empheq}[box=\fbox]{align}
    \bigbrk{
        \pd{i} M
        + M \cdot P_i\brk{0, z^\prime}
    } \cdot M^{-1} =
    \lrsbrk{
        \begin{array}{cc}
            \mtop{Q_i\brk{z^\prime}} & \mtop{\star} \\
            \mbot{\mZero} & \mbot{\star}
        \end{array}
    }.
    \label{eq:gauge-transform}
\end{empheq}
This formula for obtaining $Q_i(z')$ will be essential for calculating restrictions of Pfaffian systems in
\secref{sec:examples}.
Note that the lower-left block is zero because $\{ f_{r^\prime+1},\dots,f_r\}$ is a free basis of $\cN^\prime$ (see also \appref{sssec:M-matrix} for a matrix-based derivation of this fact).
This $\mZero$-block can be used as an intermediate check on practical computations.

Let us conclude this section by stating an assumption used throughout the rest of this paper.
First, recall that on a non-empty Zariski open subset $Y_0$ of $Y$, the $\cD_{Y_0}$-module $M(\beta;Y)|_{Y_0}$ is a regular integrable connection.

\begin{quote}
    \underline{Assumption}:
    Let $\beta$ be generic. The restriction of a normal form of
    $M_A(\beta;Y)|_{Y_0}$ to $Y^\prime$ is isomorphic to $M_A(\beta;Y^\prime)$
    as an $\cR_{Y^\prime}$-module.
\end{quote}
By \propref{prop:logarithmic_restriction} and \appref{subsec:Moser_Reduction}, the restriction of a normal form of
$M_A(\beta;Y)|_{Y_0}$ to $Y^\prime$ can be computed by means of linear algebra.
By this assumption, we can then compute the restriction of $\cD$-modules
$M_A(\beta;Y^\prime)\simeq \iota^*M_A(\beta;Y)$ from the data of
$\cR\otimes_{\cD_Y}M_A(\beta;Y)\>$.

\begin{remark} \rm
    Although we are presently unable to prove this assumption, we have verified
    in all our examples that the corank of the residue matrix is always
    identical to the expected holonomic rank shown in \thmref{thm:Gauss_Manin}.
    Moreover, we compute the $Q_i$ matrices in various examples and observe that
    they agree with the Pfaffian matrices derived independently using IBP
    software.
\end{remark}

\subsection{Restriction of solutions to a Pfaffian system}
\label{ssec:restriction-solutions}
We use the same notation as the previous subsection.
In this subsection, we investigate the boundary value problem of Pfaffian systems in normal form.
This can be seen as a singular boundary value problem; see \cite[Section 2]{haraoka2002integral},  \cite{takayama1992propagation} and references therein.
\begin{paragraph}{Holomorphic restriction}
Let $\CM$ be a logarithmic connection on $Y_0$ along $Y^\prime$ and let $\{e_1,\dots,e_r\}$ be its free basis.
Take $\{\phi_1,\dots,\phi_r\}$ to be the dual basis to $\{e_1,\dots,e_r\}$, i.e.~each $\phi_i$ is an $\cO_{Y_0}$-linear map from $\CM$ to $\mathcal{O}_{Y_0}$ satisfying $\phi_i(e_j)=\delta_{ij}$.
By the definition of the action of $\pd{i}$ on $\CM$,
we can write the space of holomorphic solutions of $\cM$ on an open
set $U$ \brk{recall the definition~\eqref{eq:hom-sol}} as
\begin{equation}
    {\rm Hom}_{\cD_{Y_0}}(\CM,\cO^{an}(U))
    =
    \lrbrc{
        \sum_{i=1}^r\MI_i(z)\,\phi_i
        \biggm|
        \sbrk{
            \MI_1(z),\dots,\MI_r(z)
        }\tr\text{ satisfies \eqref{eqn:pfaffian_system} with \eqref{eqn:connection_matrix_log}}
    }\>.
    \label{eqn:Hom}
\end{equation}
The vector of functions $\MI(z)=[\MI_1(z),\dots,\MI_r(z)]\tr$ is a solution\footnote{
    Indeed, equation~\eqref{eqn:Hom} implies that for any $m \in \cM$
    we have the homomorphism property
    $
        \pd{i} \bigbrk{\sum_{k=1}^r\MI_k(z)\,\phi_k\brk{m}}
        =
        \sum_{k=1}^r\MI_k(z)\,\phi_k\brk{\pd{i} m}
    $.
    So if we choose $m = e_j\>$, then
    $
        \pd{i} \bigbrk{\sum_{i=k}^r\MI_k(z)\,\phi_k\brk{e_j}}
        =
        \sum_{i=k}^r\MI_k(z)\,\phi_k\brk{\pd{i} e_j},
    $
    resulting in
    $
        \pd{i} \MI_j(z)
        =
        \sum_{k=1}^r\MI_k(z)\,
        p_{i j}^k(0,z^\prime)
        \>,
    $
    after we use $\phi_k\brk{e_j} = \delta_{kj}$
    and~\eqref{eqn:action_on_free_basis}.
    Thus we see that the functions $\MI_j$ indeed satisfy the Pfaffian system
    \eqref{eqn:pfaffian_system}.
}
to the Pfaffian system~\eqref{eqn:pfaffian_system}
and it will be associated to Feynman master integrals later on.

Set $U^\prime:=Y^\prime\cap U$.
We derive a realization of ${\rm Hom}_{\cD_{Y_0^\prime}}(\CM^\prime,\cO^{an}(U^\prime))$ similar to \eqref{eqn:Hom}.
Because $\cM$ is logarithmic, we can expand $P_1(z)$ as
\begin{equation}\label{eqn:expansion_of_P_1}
    P_1(z)=P_{1,-1}(z^\prime)z_1^{-1}+P_{1,0}(z^\prime)+P_{1,1}(z^\prime)z_1+\ldots
\end{equation}
The column vector
$
    \MI(z') \defas \MI(0,z^\prime)
    =\sbrk{
        \MI_1(0,z^\prime),\dots,\MI_r(0,z^\prime)
    }\tr
$
then satisfies the following system of PDEs:
\begin{empheq}[box=\fbox]{align}
    \label{eqn:restriction_equation}
    & \pd{i} \MI(z^\prime) = P_{i,0}(z^\prime) \cdot \MI(z^\prime)
    \>, \quad
    i=2,\dots,m \\[5pt] \nonumber
    & P_{1,-1}(z^\prime) \cdot \MI(z^\prime)=0
    \>.
\end{empheq}
Thus, we have an identification
\begin{equation}
    {\rm Hom}_{\cD_{Y_0^\prime}}(\CM^\prime,\cO^{an}(U^\prime))
    =
    \lrbrc{
        \sum_{i=1}^r \MI_i(0, z^\prime)\,\phi_i
        \biggm|
        \MI(z^\prime)\text{ is subject to \eqref{eqn:restriction_equation}}
    }\>.
\end{equation}
We call \eqref{eqn:restriction_equation} the \emph{holomorphic restriction} of the Pfaffian system $\pd{i} \MI(z) = P_i \cdot \MI(z)$, as it pertains to solution vectors $\MI(z)$ which are holomorphic at $z_1 = 0$.
Note that the first condition involving the residue matrix $P_{1,-1}(z')$ imposes constraints on the solution vector $\MI(z')$. These extra relations lead to a drop in rank in comparison to the unrestricted Pfaffian system.

Solutions to Pfaffian systems \eqref{eqn:connection_matrix_log} and solutions of \eqref{eqn:restriction_equation} are related to each other in the following sense:

\begin{theorem}\label{thm:restriction}
    Suppose that the Pfaffian system \eqref{eqn:connection_matrix_log} is
    non-resonant.  Then there is a one-to-one correspondence between the
    solutions $\MI(z)$ of \eqref{eqn:connection_matrix_log} being holomorphic
    at $(z_1,z')=(0,0)$ and the holomorphic solutions $\MI(z^\prime)$ of
    \eqref{eqn:restriction_equation}.  The correspondence is given by the
    boundary value map $\MI(z)\mapsto \MI(0,z^\prime)$.
\end{theorem}
\noindent See the proof in~\appref{app:restriction}.

\end{paragraph}

\begin{paragraph}{Logarithmic restriction}
So far, we have only focused on limits of Pfaffian systems for which the solution $\MI$ is holomorphic.
However, in physics applications, it is interesting to study solutions that develop logarithmic singularities in the restriction limit $z_1 \to 0$.
Let us hence define $\widetilde{\cO}^{(N)}$ to be the vector space of functions of the form
\eq{
    f_0(z)+f_1(z)\log z_1+\ldots +f_N(z)\log^Nz_1 \, ,
}
where the $f_i(z)$ denote convergent power series at the origin.
Then we have an identification
\begin{equation}
    {\rm Hom}_{\cD_{Y_0}}(\CM,\widetilde{\cO}^{(N)})
    =
    \lrbrc{
        \sum_{i=1}^r\MI_i(z)\,\phi_i
        \biggm|
        \MI_i(z)\in\widetilde{\cO}^{(N)}, \,
        \sbrk{I_1, \ldots, I_r}\tr
        \text{ satisfies \eqref{eqn:pfaffian_system} with \eqref{eqn:connection_matrix_log}}
    }\>.
    \label{eqn:Hom_log}
\end{equation}

We fix a free basis $\{e_1,\dots,e_r\}$ of $\CM$ to identify it with $\cO_{Y_0}^r$ as an $\cO_{Y_0}$-module.
For any non-negative integer $N$, we set ${\CM}^\prime_{N}:=\cO_{Y_0}^r/{\rm Im}\,P_1(0,z^\prime)^{N+1}$.
As discussed in the construction of $\CM^\prime$ in \secref{subsec:Deligne_extension}, it follows that the action of $\cD_{Y_0^\prime}$ on $\cO_{Y_0}^r$ given by \eqref{eqn:action_on_free_basis} induces that on the quotient ${\CM}^\prime$.

Consider now a vector expression of the form
\eq{
    \MI\supbrk{0}(z)+\MI\supbrk{1}(z)\log z_1+\ldots+\MI\supbrk{N}(z)\log^Nz_1.
}
We present the following generalization of \eqref{eqn:restriction_equation}, which is a system of PDEs for the vector functions $\MI\supbrk{i}(0,z')=\MI\supbrk{i}(z')$:
\begin{empheq}[box=\fbox]{align}
    \begin{array}{l l}
        \pd{i} \, \MI\supbrk{n}(z^\prime) = P_{i,0}(z^\prime) \cdot \MI\supbrk{n}(z^\prime)\>,
        & \quad i=2,\dots,m,\ n=0,\dots,N
        \\[10pt]
        P_{1,-1}(z^\prime) \cdot \MI\supbrk{n}(z^\prime) = (n+1) \MI\supbrk{n+1}(z^\prime)\>,
        & \quad n = 0, \ldots, N
    \end{array}
    \label{eq:log_constraint}
\end{empheq}
with $\MI\supbrk{N+1}(z^\prime) = 0$.
We call \eqref{eq:log_constraint} the \emph{logarithmic restriction} of \eqref{eqn:connection_matrix_log} at order $N$.
We have an identification
\begin{equation}
    {\rm Hom}_{\cD_{Y_0^\prime}}(\CM^\prime_N,\widetilde{\cO}^{(N)})
    =
    \lrbrc{
        \sum_{i=1}^r\MI\supbrk{i}(z^\prime)\,\phi_i
        \biggm|
        \MI\supbrk{i}(z^\prime)\in\cO^{an}, \, \MI\supbrk{i}\text{ is subject to \eqref{eq:log_constraint}}
    }\>.
\end{equation}

\begin{theorem}\label{thm:log_restriction}
    Suppose that the Pfaffian system \eqref{eqn:connection_matrix_log} is  non-resonant.
    Then there is a one-to-one correspondence between the solutions $\MI(z)$ to
    \eqref{eqn:connection_matrix_log} of the form $\MI\supbrk{0}(z)+\MI\supbrk{1}(z)\log z_1+\ldots+\MI\supbrk{N}(z)\log^Nz_1$ and the holomorphic solutions $\MI\supbrk{n}(z^\prime)$ of \eqref{eq:log_constraint}.
    More explicitly, the correspondence is given by the boundary value map $\MI(z)\mapsto \MI\supbrk{0}(0,z^\prime)$.
\end{theorem}
\noindent See the proof in~\appref{app:log_restriction}.
\end{paragraph}

\subsection{Relation to Feynman integrals}\label{sec:relation_to_Feynman_integrals}

Let us summarize how the preceding sections relate to the study of Feynman integrals.

Fix a single Feynman integral.
Its associated de Rham cohomology group $M_A(\beta;Y)$ (see \eqref{de_Rham_cohomology_group}) carries the structure of a holonomic $\cD_Y$-module, so the previous results on $\cD$-modules and Pfaffian systems also apply to $M_A(\beta;Y)$.
In particular, there exist an $r$-dimensional vector of master integrals $\MI=\MI(z)$ and Pfaffian matrices $P_i(z)$ such that $\pd{i} \MI = P_i(z) \cdot \MI(z)$ holds for $i = 1, \ldots, m$.
Assume that this Pfaffian system has a singularity at $z_1 = 0$.
We recall the notation $z' = (0, z_2, \ldots, z_m)$.

According to the discussion of \secref{subsec:Deligne_extension}, there exists a normal form of $M_A(\beta;Y)$. In practice, this means that we can use Moser reduction to find a set of new Pfaffian matrices $\widetilde{P}_i(z)$ such (i) $\widetilde{P}_1(z)$ has a simple pole at $z_1=0$, (ii) $\widetilde{P}_2(z_1), \ldots, \widetilde{P}_m(z)$ are holomorphic along $\{z_1=0\}$, and (iii) the spectrum of the residue matrix $\widetilde{P}_{1,-1}(z')$ is non-resonant (i.e.,~no difference of eigenvalues is integral) with zero being the unique integral eigenvalue.
By abuse of notation, we keep using the notation $\MI$ for the solution vector in this new basis.

We are interested in a system of \emph{simpler} PDEs which holds in the limit $z_1 \to 0$.
This problem is split into two cases: (i) the solution vector $\MI$ is finite at $z_1 = 0$, or (ii) the solution vector is $\log$-singular at $z_1 = 0$.

\begin{paragraph}{Case (i): Holomorphic restriction}
The relevant set of PDEs was given in \eqref{eqn:restriction_equation}.
However, this is written in a redundant fashion because we have not resolved the constraint%
\footnote{This constraint can be thought of as a collection of IBP relations which hold in the limit.}
\eq{
    \widetilde{P}_{1,-1} \cdot \MI(0,z') = 0 \, ,
    \label{eq:holomorphic_residue_constraint}
}
where the limit $z_1 \to 0$ on $\MI(z_1,z')$ is taken at the \emph{integrand} level.
There are $\fun{Rank}{\widetilde{P}_{1,-1}}$ many relations above, so we would like a restricted PDE system which is manifestly of dimension $r' = r - \fun{Rank}{\widetilde{P}_{1,-1}}$.
Assume we know an $r'$-dimensional basis $\Mi(z')$ for the restricted system (in \secref{sec:1L_bhabha} and \appref{sssec:choice-B-matrix}, we show how to construct such a basis from the original, $r$-dimensional one).
This basis fixes the matrix $B$ in \eqref{eq:M-def} according to the relation $\Mi = B \cdot \MI(0,z')$.
Moreover, the matrix $R$ in \eqref{eq:M-def} equals $\rowReduce{\widetilde{P}_{1,-1}}$,
where the operation \code{RowReduce} includes the deletion of zero-rows.
So we now know the matrix $M$ from that same formula.
Equation \eqref{eq:gauge-transform} finally yields the $r' \times r'$-dimensional Pfaffian matrices $Q_i(z')$ satisfying
\eq{
    \pd{i} \Mi = Q_i \cdot \Mi
    \quad , \quad
    i = 2, \ldots, m.
}
This restricted Pfaffian system is simpler than the original one in the sense that it has a smaller rank, and it depends on one less variable.
If needed, once $\Mi$ has been determined, it is possible to recover $\MI$ by inverting the relation
\eq{
    M \cdot
    \MI(0,z')
    =
    \lrsbrk{
        \begin{array}{c}
            B \cdot \MI(0,z') \\
            \hline
            R \cdot \MI(0,z')
        \end{array}
    }
    =
    \lrsbrk{
        \begin{array}{c}
            \Mi \\
            \hline
            0 \\[-5pt]
            \vdots \\
            0
        \end{array}
    } \, ,
}
where the block of $r-r'$ zeros below $\Mi$ on the RHS arises due to \eqref{eq:holomorphic_residue_constraint}.

Let us summarize these steps in an algorithm.

\begin{algorithm}
{\rm (Pfaffian-level holomorphic restriction to $z_1=0$)} \rm \ \label{alg:holo_rest} \\[2pt]
    \underline{Input}: The $r$-dimensional Pfaffian system $\pd{i} \MI(z) = P_i(z) \cdot \MI(z), \, i = 1, \ldots m$. \\[5pt]
    \underline{Output}: The solution vector $\MI(0,z')$.
    \begin{algorithmic}[1]

        \vspace{0.2cm}
        \State Gauge transform $P_i \to \widetilde{P}_i := G[P_i]$ to \emph{normal form} via Moser reduction.

        \vspace{0.2cm}
        \State Find the residue matrix $\widetilde{P}_{1,-1}$ and its row-reduced form $R := \rowReduce{\widetilde{P}_{1,-1}}$. Set $r' = r - \fun{Rank}{R}$.

        \vspace{0.2cm}
        \State Choose an $r'$-dimensional basis $\Mi$. This yields the $r' \times r$-dimensional matrix $B$ via $\Mi = B \cdot \MI(0,z')$. If there is no obvious choice for $\Mi$, then $B$ can be set to $B = \fun{NullSpace}{R}$.

        \vspace{0.2cm}
        \State Construct the invertible block matrix
        $
        M =
        \lrsbrk{
            \begin{array}{c}
                B \\
                \chline
                R
            \end{array}
        }
        $.

        \vspace{0.2cm}
        \State For each $i = 2, \ldots, m$, compute
        $
        \bigbrk{
            \pd{i} M
            + M \cdot P_i\brk{0, z^\prime}
        } \cdot M^{-1}
        $
        and save the upper-left $r' \times r'$ block matrices $Q_i(z')$.

        \vspace{0.2cm}
        \State Solve the Pfaffian system $\pd{i} \Mi(z') = Q_i(z') \cdot \Mi(z')$ for $\Mi$.

        \vspace{0.2cm}
        \State \Return $\MI(0,z') = M^{-1} \cdot \sbrk{\Mi  \, | \, 0 \cdots 0\,}^T$.
    \end{algorithmic}
\end{algorithm}

We visualize this algorithm in~\figref{fig:ker} by showing the building blocks of the restricted
Pfaffian system and its solution introduced in~\secref{subsec:Deligne_extension}.
The \mtop{red} surface depicts every possible solution $I(0, z_2)$ which
satisfies the constraint $\widetilde{P}_{1,-1}(z_2) \cdot I(0, z_2) = 0$, i.e.~the
nullspace of $\widetilde{P}_{1,-1}$.  For a fixed value of $z_2$, this nullspace
spans a $1$-dimensional vector space, shown as \mtop{red} lines.  As
$z_2$ varies, these lines trace out the \mtop{red} surface.  This
surface thus portrays a basis for the restricted Pfaffian system, as
given by the row-span of the $\mtop{B}$ matrix
in~\protect\eqref{eq:M-def}.
For fixed $z_2$, the \mbot{purple} lines represent the $1$-dimensional
span of the row vectors in $\widetilde{P}_{1,-1}$.  As $z_2$ varies, we obtain the
\mbot{purple} surface.  This surface hence depicts the $\mbot{R}$
matrix in~\protect\eqref{eq:M-def}.  Orthogonality of the \mbot{purple}
surface w.r.t.~the \mtop{red} surface ensures invertibility of the $M$
matrix in~\protect\eqref{eq:M-def}.
The \yel{yellow} curve illustrates a \emph{particular} solution vector
$
    \MI\brk{0, z_2} \defas \sbrk{\MI_1\brk{0, z_2}, \MI_2\brk{0,z_2}}\tr\>
$
given some initial condition $I(0,0)$.
Note that this solution is constrained to lie on the \mtop{red} surface.

\begin{figure}[H]
    \centering
    \includegraphicsbox[scale=.75]{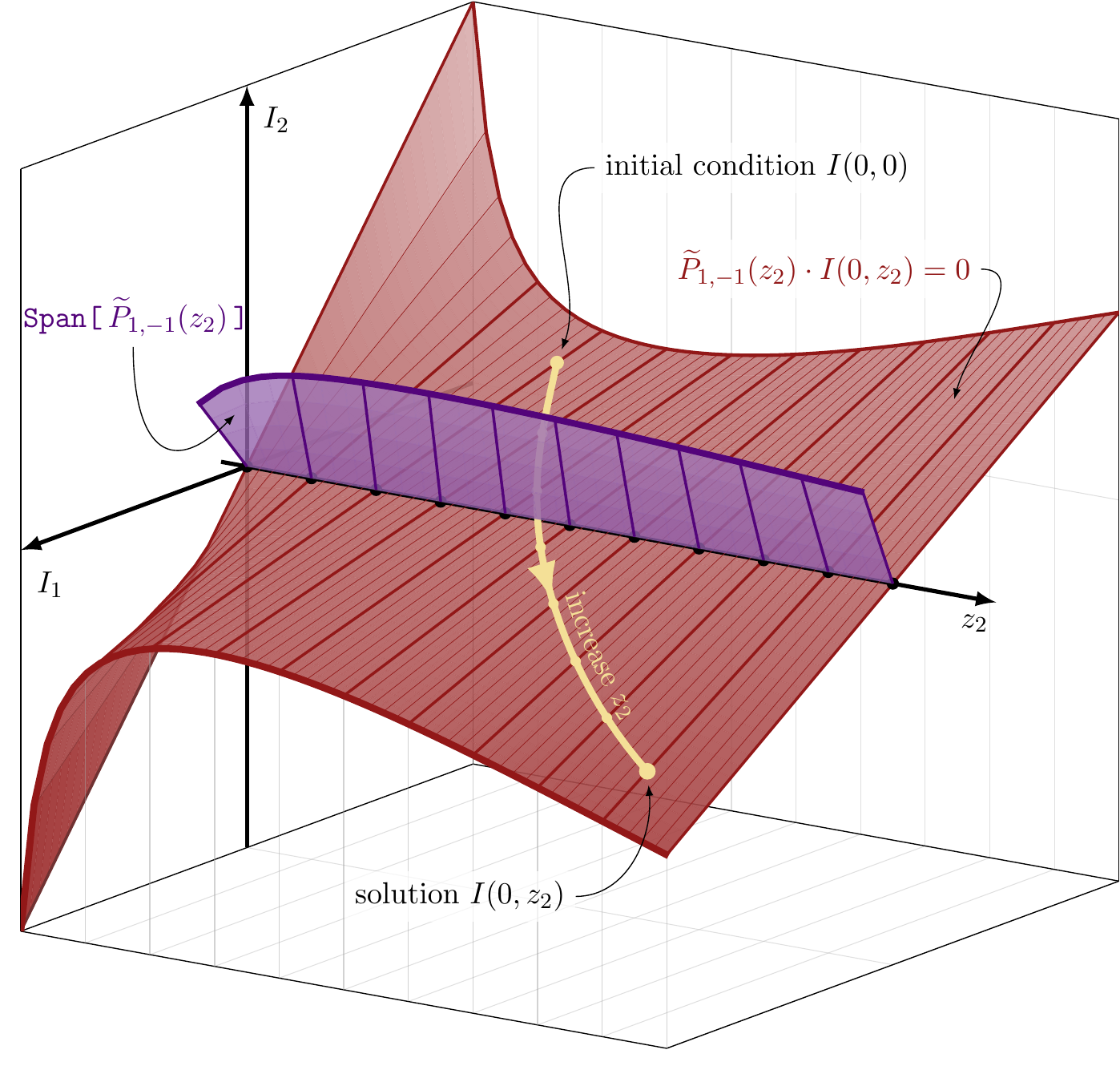}
    \caption{
        \protect\algref{alg:holo_rest}
        exemplified as the
        restriction $z_1 \to 0$ of a rank $r = 2$ Pfaffian system in two
        variables $z_1$ and $z_2$. \\
    }
    \label{fig:ker}
\end{figure}

\end{paragraph}

\begin{paragraph}{Case (ii): Logarithmic restriction}
By the general theory of PDEs \cite{haraoka2020linear}, we expect the solution vector $\MI(z)$ to take the following form in the limit $z_1 \to 0$:
\eq{
    \MI\supbrk{\log}(z)
    \quad = \quad
    \sum_{\lambda \, \in \, \spec{\widetilde{P}_{1,-1}}}
    \hspace{-0.3cm} z_1^\lambda \hspace{0.2cm}
    \sum_{n=0}^{N_\lambda} \MI\supbrk{\lambda,n}(z') \log^n(z_1) .
    \label{eq:log_singular_I}
}
The outer sum runs over the unique eigenvalues $\lambda$ of the residue matrix $\widetilde{P}_{1,-1}$, and for each $\lambda$ there is a collection of $N_\lambda+1$ unknown vectors $\MI\supbrk{\lambda,0}, \ldots, \MI\supbrk{\lambda,N_\lambda}$.
The integer $N_\lambda$ depends on the structure of the Jordan blocks in the Jordan decomposition of $\widetilde{P}_{1,-1}$.

For fixed $\lambda,$ each $\MI\supbrk{\lambda,n}$ is subject to the system of PDEs \eqref{eq:log_constraint}.
This system includes a list of constraints
\eq{
    (\widetilde{P}_{1,-1}-\lambda) \cdot \MI\supbrk{\lambda,n} =
    (n+1) \MI\supbrk{\lambda,n+1}
    \quad , \quad
    n = 1, \ldots, N_\lambda,
}
which leads to a drop in rank.
Equation \eqref{eq:log_constraint} was specifically written for the eigenvalue $\lambda = 0$, but in the general case there is a shift by $-\lambda$ as on the LHS.
We observe that the constraints above imply
\eq{
    \MI\supbrk{\lambda, n} =
    (n+1)! \, (\widetilde{P}_{1,-1}-\lambda)^{n+1} \cdot \MI\supbrk{\lambda,0}
    \quad , \quad
    n = 1, \ldots, N_\lambda,
    \label{log_singular_constraint}
}
for which reason we only need to compute $\MI\supbrk{\lambda,0}$.
Setting $n = N_\lambda$ above, we have
\eq{
    \label{eq:N_lambda_constraints}
    (\widetilde{P}_{1,-1}-\lambda)^{N_\lambda+1} \cdot \MI\supbrk{\lambda,0} = 0 \, ,
}
giving $\fun{Rank}{\widetilde{P}_{1,-1}^{N_\lambda+1}}$ many constraints among the entries of the solution vector $\MI\supbrk{\lambda,0}$.

For each $\lambda$, we calculate $\MI\supbrk{\lambda,0}$ by repeating the procedure from case (i).
The first step is to find an $r'_\lambda$-dimensional basis $\Mi\supbrk{\lambda}$, where
$r'_\lambda = r - \fun{Rank}{(\widetilde{P}_{1,-1}-\lambda)^{N_\lambda+1}}$, so as to fix the matrix $B\supbrk{\lambda}$ in
\eq{
    M\supbrk{\lambda} =
    \lrsbrk{
        \begin{array}{c}
            B\supbrk{\lambda} \\
            \chline
            R\supbrk{\lambda}
        \end{array}
    }.
}
It is important that $B\supbrk{\lambda}$ be chosen such that $M\supbrk{\lambda}$ is invertible.
The $R\supbrk{\lambda}$ block encodes the constraints from the residue matrix in \eqref{eq:N_lambda_constraints}:
\eq{
    R\supbrk{\lambda} = \rowReduce{(\widetilde{P}_{1,-1}-\lambda)^{N_\lambda+1}} \, .
}
Because of \eqref{eq:N_lambda_constraints}, we thus get a block of zeros in
\eq{
    M\supbrk{\lambda} \cdot
    \MI\supbrk{\lambda,0}
    =
    \lrsbrk{
        \begin{array}{c}
            B\supbrk{\lambda} \cdot \MI\supbrk{\lambda,0} \\
            \hline
            R\supbrk{\lambda} \cdot \MI\supbrk{\lambda,0}
        \end{array}
    }
    =
    \lrsbrk{
        \begin{array}{c}
            \Mi\supbrk{\lambda} \\
            \hline
            0 \\[-5pt]
            \vdots \\
            0
        \end{array}
    } \ .
    \label{log_singular_M}
}
This can be inverted to recover $\MI\supbrk{\lambda,0}$ once we have computed $\Mi\supbrk{\lambda}$.
Inserting $M\supbrk{\lambda}$ into \eqref{eq:gauge-transform}, we obtain a collection of $r'_\lambda \times r'_\lambda$-dimensional Pfaffian matrices $Q\supbrk{\lambda}_i(z')$ such that
\eq{
    \pd{i} \Mi \supbrk{\lambda} =
    Q\supbrk{\lambda}_i \cdot \Mi \supbrk{\lambda}
    \quad , \quad
    i = 2, \ldots, m.
    \label{log_singular_Pfaffian_system}
}

The equations \eqref{log_singular_constraint,log_singular_M,log_singular_Pfaffian_system} together determine the unknown vectors $\MI\supbrk{\lambda,n}$ in \eqref{eq:log_singular_I}, apart from fixing the boundary constants stemming from solutions of the Pfaffian system \eqref{log_singular_Pfaffian_system}.
These constants can e.g.~be fixed by comparing with the original integral $\MI(z)$ at special numerical points.

Summarizing all the steps above, we have
\begin{algorithm}
{\rm (Pfaffian-level logarithmic restriction to $z_1=0$)} \rm \ \label{alg:log_rest} \\[2pt]
    \underline{Input}: The $r$-dimensional Pfaffian system $\pd{i} \MI(z) = P_i(z) \cdot \MI(z), \, i = 1, \ldots m$. \\[5pt]
    \underline{Output}: The solution vector $\MI\supbrk{\log}(z)$.
    \begin{algorithmic}[1]

        \vspace{0.2cm}
        \State Gauge transform $P_i \to \widetilde{P}_i := G[P_i]$ to \emph{normal form} via Moser reduction.

        \vspace{0.2cm}
        \State Find the residue matrix $\widetilde{P}_{1,-1}$.

        \vspace{0.2cm}
        \For{each unique $\lambda \in \spec{\widetilde{P}_{1,-1}}$}

            \vspace{0.2cm}
            \State Compute
            $
                R\supbrk{\lambda} = \rowReduce{(\widetilde{P}_{1,-1}-\lambda)^{N_\lambda+1}}
            $
            and set $r'_\lambda = r - \fun{Rank}{R\supbrk{\lambda}}$.

            \vspace{0.2cm}
            \State Choose an $r'_\lambda$-dimensional basis $\Mi\supbrk{\lambda}$. This fixes an $r'_\lambda \times r$-dimensional matrix $B\supbrk{\lambda}$ via $\Mi\supbrk{\lambda} = B\supbrk{\lambda} \cdot \MI\supbrk{\lambda,0}$. One choice is $B\supbrk{\lambda} = \fun{NullSpace}{R\supbrk{\lambda}}$.

            \vspace{0.2cm}
            \State Construct the invertible block matrix
            $
            M\supbrk{\lambda} =
            \lrsbrk{
                \begin{array}{c}
                    B\supbrk{\lambda} \\
                    \chline
                    R\supbrk{\lambda}
                \end{array}
            }
            $.

            \vspace{0.2cm}
            \State For each $i = 2, \ldots, m$, compute
            $
            \bigbrk{
                \pd{i} M\supbrk{\lambda}
                + M\supbrk{\lambda} \cdot P_i\brk{0, z^\prime}
            } \cdot (M\supbrk{\lambda})^{-1}
            $
            and save the upper-left $r'_\lambda \times r'_\lambda$ block matrices $Q\supbrk{\lambda}_i(z')$.

            \vspace{0.2cm}
            \State Solve the Pfaffian system $\pd{i} \Mi\supbrk{\lambda}(z') = Q_i\supbrk{\lambda}(z') \cdot \Mi\supbrk{\lambda}(z')$ for $\Mi\supbrk{\lambda}$.

            \vspace{0.2cm}
            \State Compute
            $
            \MI\supbrk{\lambda,0}(z') = (M\supbrk{\lambda})^{-1} \cdot \sbrk{\Mi\supbrk{\lambda} \, | \, 0 \cdots 0\,}^T
            $

            \vspace{0.2cm}
            \State Compute
            $
            \MI\supbrk{\lambda, n}(z') =
            (n+1)! \, (\widetilde{P}_{1,-1}-\lambda)^{n+1} \cdot \MI\supbrk{\lambda,0}
            $
            for each $n = 1, \ldots, N_\lambda$.
        \EndFor
        \State \Return The sum \eqref{eq:log_singular_I}.
    \end{algorithmic}
\end{algorithm}

This is a novel approach for computing logarithmic corrections to scattering events.
A pedagogical example is given in \secref{sec:1L_bhabha}, where we compute logarithmic corrections to a one-loop box diagram contributing to Bhabha scattering.
In that example we imagine that $z_1$ denotes a small mass, such as that of the electron.
While the goal of the example is to compute $\MI\supbrk{\lambda,\bullet}(z')$, the actual, full, solution $\MI(z)$ would read
\eq{
    \MI(z) \ \ =
    \sum_{\lambda \, \in \, \spec{\widetilde{P}_{1,-1}}}
    \hspace{-0.3cm} z_1^\lambda \hspace{0.2cm}
    \sum_{n=0}^{N_\lambda} \MI\supbrk{\lambda,n}(z_1,z^\prime) \log^n(z_1) \, ,
}
where $\MI\supbrk{\lambda,n}(z_1,z^\prime)$ is the analytic function satisfying
$\MI\supbrk{\lambda,n}(0,z^\prime)=\MI\supbrk{\lambda,n}(z^\prime)$.
Using the recursion relations~\eqref{eqn:define_I_k, eq:MI-Gamma} from the proofs of \thmref{thm:restriction} and
\thmref{thm:log_restriction}, presented in \appref{sec:proofs}, it it possible
to reconstruct $\MI\supbrk{\lambda,n}(z_1,z^\prime)$ from
$\MI\supbrk{\lambda,n}(z^\prime)$ as a power series in $z_1$.  In other words,
one would thereby obtain terms of the form $z_1^a \log^b(z_1)$ for $a,b \in
\ZZ_{>0}$.  Explicit calculations of this nature are left for future work.
\end{paragraph}

\section{Restrictions as \texorpdfstring{$\cD$}{}-modules by the Macaulay matrix method}
\label{sec:NTsec1_MMMR}

In this section, we will discuss rational restrictions from the $\cD$-module
point of view.
The first subsection gives an algorithm for finding bases of rational restrictions.
The second subsection uses these bases to construct rational restrictions via the Macaulay matrix method.

\subsection{Restrictions as \texorpdfstring{$\CD$}{}-modules} \label{sec:NTsec1}

Recall the notion of rational restrictions from \defref{defn:rational_restriction}.
There, $Y$ is the $\ydim$-dimensional complex plane and
$Y'$ is $(\ypdim)$-dimensional complex plane in $Y$ defined by
$z_{1}= \ldots = z_{\ypcodim}=0$.
The following proposition will be used to prove the correctness of
\algref{alg:alg2} later.
A proof is given in \appref{appendix:proof:prop:prop1}.
\begin{proposition} \label{prop:prop1}
Let $\cI$ be a regular holonomic ideal of $\CD_Y$.
We have the isomorphism of left $\CR_{Y'}$-modules
\begin{equation} \label{eq:weyl-closure-and-restriction}
 \CR_{Y^\prime} \otimes_{\CD_{Y^\prime}} \CD_Y/(\cI+ \sum_{j=1}^{\ypcodim} z_j \CD_Y)
\simeq
 \CR_{Y^\prime} \otimes_{\CD_{Y^\prime}} \CD_Y/((\CR_Y \cI)\cap \CD_Y + \sum_{j=1}^{\ypcodim} z_j \CD_Y) \, .
\end{equation}
\end{proposition}
\noindent The ideal $(\CR_Y\cI) \cap \CD_Y$ is known as the Weyl closure of $\cI$ \cite[Definition 5.1]{Tsai2001algorithms}.

Let $\cI$ be a left holonomic ideal in $\CD_Y$.
The restriction of the left holonomic $\CD_Y$-module
$\CM=\CD_Y/\cI$ to $Y^\prime$ is defined by $\iota^* \CM=\CD_Y/(\cI + z_1 \CD_Y + \ldots + z_{m'} \CD_Y)$.
We denote the restriction $\iota^* \CM$ by $\CN_{Y^\prime}$.
When $m=m'$, we denote the restriction by $\CN_0$.

Let $w \in \ZZ^{\ydim}$ be a vector and let $w_i$ denote its $i$-th component.
Define the $(-w,w)$-\emph{order} for $c z^\alpha \pd{}^\beta$, $c \in \CC$, by
\begin{equation} \label{eq:def-ord}
 \ord_{(-w,w)} (c z^\alpha \pd{}^\beta) =  -w \cdot \alpha + w \cdot \beta \, ,
\end{equation}
where $z^\alpha = \prod_{i=1}^{\ydim} z_i^{\alpha_i}$,
$\pd{}^\beta = \prod_{i=1}^{\ydim} \pd{i}^{\beta_i}$
and
$\cdot$ is the standard inner product.
The order is defined for $\ell \in \CD_Y$
as the maximum of all monomial $(-w,w)$-orders appearing in $\ell$.
For $\ell, \ell' \in \CD_Y$, we have
$ \ord_{(-w,w)} (\ell \ell') = \ord_{(-w,w)} (\ell) + \ord_{(-w,w)}(\ell')$.
Putting $\theta_i=z_i \pd{i}$, we note that $\ord_{(-w,w)}(\theta_i)=0$.
Since $w$ defines the order among monomials,
it is called the \emph{weight vector}.

Now, let us briefly review an algorithm for
constructing the restriction $\CN_{Y'}$ of $\CM$
by following the exposition \cite[Section 5.2]{SST}.
Let $w \in \ZZ^{\ydim}$ be a weight vector such that
\begin{equation} \label{eq:weight-cond}
w_{\ypcodim+1}=\ldots=w_{\ydim}=0
\quad , \quad
w_i > 0
\quad , \quad
1 \leq i \leq \ypcodim.
\end{equation}
Let $b(u)$ be the $b$-\emph{function} (also called the indicial polynomial)
of the left ideal $\cI$ in $\CD_Y$
with respect to $w$.
More precisely,
$b(\sum_{i=1}^{\ypcodim} w_i z_i \pd{i})$
is a monic generator of the principal ideal
\begin{equation} \label{eq:def-bfunc}
 \CC\left[\sum_{i=1}^{\ypcodim} w_i z_i \pd{i}\right] \cap {\rm in}_{(-w,w)} \cI \, ,
\end{equation}
where ${\rm in}_{(-w,w)}(f)$ is the sum of the highest
order terms in $f \in \CD_Y$ such that
the weight for $z_i$  is $-w_i$ and
the weight for $\pd{i}$ is $w_i$ \cite[Section 1.1]{SST}%
\footnote{
We note that several computer algebra systems have commands for computing $b$-functions -
see in particular \cite[p. 194]{SST} for a definition of the $b$-function which is computer algebra friendly and \cite{dojo} for details on computer algebra systems which are useful in this context.
\asir{Example: {\tt import("nk\_restriction.rr"); nk\_restriction.generic\_bfct\_and\_gr([x*dx+y*dy-1, x*dx-y*dy-2],[x,y],[dx,dy],[1,0]); }
gives the $b$-function and GB for the ideal {\tt x*dx+y*dy-1, x*dx-y*dy-2} along $x=0$. }}.

Fix a weight vector $w$ satisfying (\ref{eq:weight-cond}).
Let $\cF_k$ be the set of all elements in $\CD_Y$
such that the $(-w,w)$-order is less than or equal
to $k$.
In other words, we have
\begin{equation} \label{eq:KM-filtration}
    \cF_k
    =
    \lrbrc{
        \sum c_{\alpha \beta} z^\alpha \pd{}^\beta
        \Bigm|
        -w \cdot \alpha + w \cdot \beta \leq k
    }\>,
\end{equation}
which is a free left $\CD_{Y'}$-module.

\begin{theorem} {\rm \cite{Oaku-1997}, \cite[Section 5.2]{SST}} \label{th:th2}
\begin{enumerate}
\item If $b(u)=0$ does not have a non-negative integral root, then the restriction
$\CN_{Y'}=0$.
\item Let $u_0$ be the maximal non-negative integral root of $b(u)=0$
and let
$ \{ g_1, \ldots, g_t \}$ be a $(-w,w)$-Gr\"obner basis of $\cI$
(see \cite[Section 1.1]{SST}).
Then the restriction $\CN_{Y'}$ is expressed as
\begin{equation}  \label{eq:rest-by-gb}
\CN_{Y'} \simeq \frac{\cF_{u_0}}
  { \sum_{\{i\in\{1,\ldots,t\} \,|\, \ord_{(-w,w)} g_i \leq u_0\}} \cF_{u_0-\ord_{(-w,w)}(g_i)} g_i + \sum_{j=1}^{\ypcodim} z_j \cF_{u_0+w_j} }
\end{equation}
as a left $\CD_{Y'}$-module.
\end{enumerate}
\end{theorem}

\noindent The right hand side of (\ref{eq:rest-by-gb}) can also be expressed as
\begin{equation} \label{eq:rest-by-gb2}
\frac{\sum_{ w \cdot \alpha \leq u_0} \CD_{Y'} \prod_{j=1}^{\ypcodim} \pd{j}^{\alpha_j}}
  { \sum_{\{i\in\{1,\ldots,t\} \,|\, \ord_{(-w,w)} g_i \leq u_0\}} \CD_{Y'}
    \sum_{ w \cdot \beta \leq u_0-\ord_{(-w,w)}(g_i)} :\prod_{j=1}^{\ypcodim}
    \pd{j}^{\beta_j} g_i:{|_{z_1=\ldots=z_{\ypcodim}=0}} } \, .
\end{equation}
The notation $:\ell:{|_{z_1=\ldots=z_{\ypcodim}=0}}$ means that we first order
each monomial in $\ell$ such that $\pd{i}$'s come to the right and $z_j$'s come
to the left (via $[\pd{i}, z_j] = \delta_{ij}$), whereafter we set $z_1 = \ldots=z_{\ypcodim}=0$.
For example, we have
$$
  :\pd{1} z_1:\big|_{z_1=0} = \>\> (z_1 \pd{1} + 1) \big|_{z_1=0} = 1.
$$
Note that $u_0$ may be set equal to a non-negative integer
larger than or equal to the maximal non-negative integral root
of $b(u)=0$.
This procedure will be utilized in the following sections.

\begin{remark} \rm \label{rem:monomial-basis}
Let us regard $\CR_{Y'} \otimes_{\CD_{Y'}} \CN_{Y'}$ as a vector space over
$\CC(z_{\ypcodim+1}, \ldots, z_{\ydim})$.
Considering a Gr\"obner basis in a free $\CR_{Y'}$-module for the denominator
of the expression (\ref{eq:rest-by-gb2}),
we can take a basis of the vector space
consisting only of monomials in $\pd{i}$'s.
\end{remark}

\noindent
\ntcomment{
Risa/Asir codes for examples in this note are at\\
\url{http://www.math.kobe-u.ac.jp/OpenXM/Math/rest-by-MM/2023-01-13-rest-by-mm}.
}

We proceed by constructing an algorithm for finding the Pfaffian system associated to the restriction $\CN_{Y'}$ via Gr\"obner bases for submodules.
This algorithm will lend itself to the Macaulay matrix restriction method presented in \secref{sec:MMMR} by providing a standard basis.

In order to simplify our presentation,
we first illustrate our algorithms for $\ydim=2$ and $\ypcodim=1$.
The generalization to more than $2$ variables is not difficult, and will be discussed later.
In the rest of this section, we set $x=z_1, y=z_2$ and $\CD'=\CD_{\ydim}, \CD=\CD_{\ypdim}$.

We want to restrict a given left ideal $\cI=\CD'\cdot \{f_1, \ldots, f_\mu\}$ in
$\CD'=\CC\langle x, y, \pd{x}, \pd{y} \rangle$
to $x=0$.
The restriction module is defined as
$$ \CD'/(\cI + x \CD'). $$
As stated in \thmref{th:rest-holonomic},
when $\CD'/\cI$ is a holonomic $\CD'$-module,
the restriction module is a holonomic $\CD$-module
with
$\CD=\CC\langle y, \pd{y} \rangle$.

Let $\cJ$ be a left submodule of $\CD^{k+1}$.
The holonomic rank of the left $\CD$-module $\CD^{k+1}/\cJ$ is, by definition, $
 \mathrm{dim}_{\CC(y)} \CR^{k+1}/(\CR \cJ)$,
where $\CR=\CC(y)\langle \pd{y} \rangle$.

Let $\ell \in \CD'$ and take $k$ to be a sufficiently large number.
We regard the operator
$$ \ell = \sum_{j=0}^k c_j(x,y,\pd{y}) \pd{x}^j
$$
as an element of $\CD[x]^{k+1}$ by
\begin{equation} \label{eq:def_vk}
 v_k(\ell) \defas (c_k, c_{k-1}, \ldots, c_1, c_0) \in \CD[x]^{k+1}.
\end{equation}
We regard $\CD[x]^{k+1}$ (resp. $\CD^{k+1}$) as a subset of $\CD[x]^{k+2}$
(resp. $\CD^{k+2}$)
by the inclusion
$$
\CD[x]^{k+1} \ni (c_k, c_{k-1}, \ldots, c_1, c_0) \mapsto
(0,c_k, c_{k-1}, \ldots, c_1, c_0) \in \CD[x]^{k+2}.
$$

Let $u_0$ be the maximal non-negative integral root of the $b$-function
of $\cI$ for the weight $w=(1,0)$.
If there is no such $u_0$, the restriction is $0$.
Let $u_1$ be the maximal $(-w,w)$-order of the elements in
a $(-w,w)$-Gr\"obner basis for $\cI$ (see \cite[Section 1.1]{SST} for details on Gr\"obner bases of Weyl algebras).
\begin{algorithm}{{\rm (Rational restriction to $x=0$)}}  \rm \ \label{alg:alg1} \\
    \underline{Input}: Generators $\{f_1, \ldots, f_\mu\}$ of a holonomic ideal $\cI$ in $\CD'$.
    Holonomic rank $r$ of the restriction of $\CD'/\cI$ to $x=0$.
    A positive integer $\gamma$ such that $\gamma \geq \max(u_0+1,u_1+1)$. \\
    \underline{Output}: Generators of a left submodule $\cJ\cap \CD^\gamma$ of $\CD^{\gamma}$ such that $\CR^{\gamma}/\CR (\cJ \cap \CD^\gamma) \simeq \CR \otimes_{\CD} \CD'/(\cI+x \CD')$, that is the rational restriction.
    \begin{algorithmic}[1]
        \State $w=(1,0)$
        \State $k = \gamma-1$
        \Repeat
        \State  $\cJ = \CD\cdot \Bigbrc{
            v_k \, \bigbrk{\,:\pd{x}^j f_i:{\big|_{x=0}}}
            \Bigm|
            \ord_{(-w,w)} (\pd{x}^j f_i) \leq k , 1 \leq i \leq N,
            \ j \in \NN_0
        }$
        \State $k \defas k+1$.
        \Until{$\mathrm{rank}\bigbrk{\CD^{\gamma}/\cJ \cap \CD^{\gamma}} = r$}
        \State \Return $\cJ \cap \CD^\gamma$
    \end{algorithmic}
\end{algorithm}

\noindent We can obtain generators of $\cJ \cap \CD^\gamma$ by
the POT (position over term) order%
\asir{The asir command to compute a Gr\"obner basis
by the POT order is
{\tt nd\_weyl\_gr(Generators,Variables,0,[1,0])}.
Example: {\tt nd\_weyl\_gr([[dx,dx],[1,dy],[1,dy*dy]],[x,y,dx,dy],0,[1,0])}.
}
in $\CD^{k+1}$
(see \appref{sec:nut} for Gr\"obner bases in a free module and
the POT order).
Once we get a set of generators of $\cJ \cap \CD^\gamma$,
the Pfaffian system of the restriction $\CD'/(\cI+x\CD')$
can be obtained from a Gr\"obner basis
of $\CR (\cJ \cap \CD^\gamma)$.
An improvement of this algorithm is given in \appref{appendix:improve}.

\begin{theorem} \label{th:th3}
\algref{alg:alg1} stops at some number $k$ and gives the correct
answer.
\end{theorem}
\noindent
The proof of this theorem is technical; see \appref{appendix:proof:th:th3}.

Let us illustrate \algref{alg:alg1} with a small input.
\begin{example}\rm  \label{ex:4-dim-sol}
Let $\cI$ be the left ideal generated by
\begin{eqnarray*}
f_1 &=& 2\theta_x \theta_y-1 = 2 x y \pd{x}\pd{y}-1, \\
f_2 &=& 2 \theta_y^2+\theta_x-\theta_y-1
     =  2 (y^2\pd{y}^2+y \pd{y})+x \pd{x}-y \pd{y}-1.
\end{eqnarray*}
The solution space for $\cI$ is spanned by
$\{x y^{1/2}, \,
x^{1/\sqrt{2}} y^{1/\sqrt{2}}, \,
x^{-1/\sqrt{2}} y^{-1/\sqrt{2}}\}$.
We will construct the restriction to $x=0$.
The listed solutions imply that the expected rank of the restriction is $1$,
    standing for the solution $x y^{1/2}$, which is the only solution holomorphic at $x = 0$.
The $b$-function for the weight $w=(1,0)$ is
$(u-1)(u^2-1/2)$, implying that $u_0=1$.
The maximum $(-w,w)$-order of the $(-w,w)$-Gr\"obner basis
of $\cI$ is $0$, hence $u_1=0$.
Thus we put $\gamma=2$.
We start with $k=1$ in \algref{alg:alg1}.
The submodule $\cJ$ is generated by
\begin{eqnarray*}
&&(0, - 1)=v_1(:f_1:|_{x=0}), \\[5pt]
&&(    2  {y}  \pd{y}- 1,0)=v_1(:\pd{x}f_1:|_{x=0})
 =v_1(2y\pd{x}\pd{y}+2xy\pd{x}^2\pd{y}-\pd{x} |_{x=0}),\\[5pt]
&&(0,     2   {y}^{ 2}    \pd{y}^{ 2} +  {y}  \pd{y}- 1)
 = v_1(:f_2:|_{x=0}),\\[5pt]
&&(    2   {y}^{ 2}    \pd{y}^{ 2} +  {y}  \pd{y},0)
 = v_1(:\pd{x}f_2:|_{x=0})
\end{eqnarray*}
The Gr\"obner basis of $\cJ$ with the POT order is
$$
\{ (0, - 1), \, (    2  {y}  \pd{y}- 1,0) \},
$$
whose holonomic rank is $1$.
The set of standard monomials of the Gr\"obner basis in $\cD^2$ is
$S=\{s_1:=(1,0) \in \cD^2 \}$.
Then, we have the $1 \times 1$ Pfaffian matrix $\left(\frac{1}{2y}\right)$
for $S$ as
\begin{equation*}
\pd{y}s_1=(\pd{y},0) \equiv \left(\frac{1}{2y},0\right)=\frac{1}{2y}s_1.
\end{equation*}
    \hfill$\blacksquare$
\end{example}

Let us now generalize to the case of several variables
$z=(z_1, \ldots, z_m)$.
Fix an integer $m$ and a weight vector $w$ satisfying (\ref{eq:weight-cond}).
We put
\begin{equation}
    \label{eq:def_ck}
    c(k) = \# E_{k}\>, \quad
    E_k = \bigbrc{
        \alpha
        \bigm|
        \alpha \cdot w \leq k, \,
        \alpha_1, \ldots,\alpha_{\ypcodim} \in \NN_0, \,
        \alpha_{\ypcodim+1}=\ldots=\alpha_{\ydim}=0
    }\>.
\end{equation}
We regard $\bigbrc{
    \pd{1}^{\alpha_1} \ldots \pd{\ypcodim}^{\alpha_{\ypcodim}} \bigm| \alpha \in E_{k}
}$ as the standard basis of $\CD_{Y'}[z_1,\ldots, z_{\ypcodim}]^{c(k)}$.
We define the map $v_k$ from $\CD_Y$ to $\CD_{Y'}[z_1,\ldots, z_{\ypcodim}]^{c(k)}$
by
\begin{equation}
v_k(\ell) = \sum_{\alpha \in E_k} \ell_\alpha \pd{1}^{\alpha_1} \ldots \pd{\ypcodim}^{\alpha_{\ypcodim}} \, ,
\end{equation}
where we decompose $\CD \ni \ell$ as
$\ell = \sum_{\alpha \in E_k} \ell_\alpha \pd{1}^{\alpha_1} \ldots \pd{\ypcodim}^{\alpha_{\ypcodim}}$,
$\ell_\alpha \in \CD_{Y'}[z_1, \ldots, z_{\ypcodim}]$.
\algref{alg:alg1} can be generalized using $v_k$.
The map $v_k$ takes values in a $c(k)$-dimensional vector space.
We start the procedure from $k \geq \max(u_1,u_0)$, where
$u_0$ is the maximal non-negative integral root of the $b$-function
with respect to $w$, and $u_1$ is the maximum of $\ord_{(-w,w)}$
for the $(-w,w)$-Gr\"obner basis of $\cI$.
Note that $c(k)=k+1$ when $\ydim=2, \, \ypcodim=1$.

The case $\ypcodim = \ydim$ will be used to find a rational restriction
by the Macaulay matrix method, so the following algorithm assumes this equality.
\begin{algorithm}{{\rm (Basis for the restriction to $z_1=\cdots=z_m=0$)}} \rm \ \label{alg:rest_to_pt} \\
    \underline{Input}: Generators $\{f_1, \ldots, f_\mu\}$ of a holonomic ideal $\cI$ in $\CD_Y$.
    The vector space dimension $r$ of holomorphic solutions at $z=(z_1, \ldots, z_m)=0$.
    A positive integer $\gamma$ such that $\gamma \geq \max(u_0,u_1)$. \\
    \underline{Output}: A $\C$-basis of the restriction $\CN_0$ to the point $z=(z_1, \ldots, z_m)=0$.
    \begin{algorithmic}[1]
        \State $w=(1,\ldots,1)$
        \State $k = \gamma$
        \Repeat
        \State  $\cJ =
            \CC\cdot \Bigbrc{
                v_k\,\bigbrk{\,:\pd{}^\alpha f_i:{\big|_{z=0}}}
                \Bigm|
                \ord_{(-w,w)}\lrbrk{\prod_{j=1}^{\ydim} \pd{j}^{\alpha_{j}} f_i} \leq k,
                \ \alpha \in \NN_0^{\ydim}
            } \subseteq \CC^{c(k)}
        $
        \State $k \defas k+1$
        \Until{$\mathrm{dim}\bigbrk{\CC^{c(\gamma)}/\cJ \cap \CC^{c(\gamma)}} = r$}
        \State \Return A vector space basis of $\CC^{c(\gamma)}/\cJ \cap \CC^{c(\gamma)}$.
    \end{algorithmic}
    \end{algorithm}
\noindent Here we identify the vector space $\CC^{c(k)}$
with the vector space spanned by $\pd{}^\alpha, \alpha \in E_k$.
For example, when $m=2$, $k=2$, and $\pd{}^\alpha$ is sorted as
$(\pd{1}^2,\pd{1}\pd{2},\pd{2}^2,\pd{1},\pd{2},1)$,
then $(0,0,0, 0,0,1)$ stands for $1$,
$(0,0,0, 0,1,0)$ for $\pd{2}$,
and
$(1,0,0, 0,0,0)$ for $\pd{1}^2$.
Note that this algorithm only utilizes numerical linear algebra, so
it is enough to compute a basis of the linear space
$\CC^{c(\gamma)}/\cJ \cap \CC^{c(\gamma)}$.
The numerical matrix representing $\cJ$ is large in general, so in our implementations we perform the steps $6$ and $7$ over a finite field instead of $\CC$, which greatly speeds up the computation.
Note also that we do not know $\gamma$ beforehand without computing a Gr\"obner basis, for which reasons we also choose $\gamma$ probabilistically in our implementation.

\begin{example}\rm  \label{ex:y2-x3-1-2}
We denote $z_1, z_2$ by $x, y$ respectively.
The function $(y^2-x^3)^{1/2}$ spans the holomorphic solution space
of the operators
$\{2 y \pd{x}+3 x^2 \pd{y}, \, 2 x \pd{x}+3 y \pd{y}-3\}$
around $(x,y)=(1,2)$.
Let us compute the restriction to this point.
We make the change of variables $x \defas x+1$ and $y \defas y+2$
and apply \algref{alg:rest_to_pt} with $w=(1,1)$ and
$m=2$.
We have $b(u) = u$, so $u_0=0$.
The $(-w,w)$-initial terms of the $(-w,w)$-Gr\"obner basis are
\begin{align}
    \bigbrc{
        \pd{x} + \ldots,
        \pd{y} + \ldots,
        x \pd{y}^2 + \ldots,
        x^2 \pd{x}^3 + \ldots
    }\>.
    \nonumber
\end{align}
The $(-w,w)$-orders all equal $1$.
Hence, we can take $u_1=1$.
We have
$E_0=\{(0,0)\}$ and
$E_1=\{(1,0), \, (0,1), \, (0,0)\}$.
Therefore, we can set $\gamma=1$ and $c(\gamma)=2$.
The input operators are
\eq{
    f_1 &= 2 (y+2) \pd{x}+3 (x+1)^2 \pd{y} \\
    f_2 &= 2 (x+1) \pd{x}+3 (y+2) \pd{y}-3.
}
Then the generators of $\cJ$ for $k=1$ are
\eq{
    v_k(f_1|_{x=y=0})=(4,3,0)
    \quad , \quad
    v_k(f_2|_{x=y=0})=(2,6,-3) \, ,
}
with the vector $v_k$ being indexed by $(\pd{x},\pd{y},1)$.
The row echelon form of
$\lrsbrk{
    \begin{array}{ccc}
     4 & 3 & \mzero \\
     2 & 6 & -3 \\
    \end{array}
}
$
is
$
\lrsbrk{
    \begin{array}{ccc}
        -2 & \mzero & -1 \\
        \mzero  & -3 & 2 \\
    \end{array}
},
$
meaning that the rank of $\CC^3/\cJ \cap \CC^3$ is $1$.
Note that $\gamma=1$ does not work even when we increase $k$ to larger values.
    \hfill$\blacksquare$
\end{example}

\subsection{Macaulay matrix method for restrictions}
\label{sec:MMMR}

In this subsection, we consider the case when $Y^\prime\subset Y$ is an irreducible,
 closed and smooth subvariety.
Building on the work of \cite{Chestnov:2022alh}, the aim of this section is to give a Macaulay matrix method for rational restrictions, as per \defref{defn:rational_restriction},
when $Y^\prime$ is either $z_1= \ldots = z_{\ypcodim}=0$
or $Y^\prime$ is an irreducible hypersurface.
The rational restriction gives a description of
holomorphic solutions on $Y^\prime$.

Let $\MM=\DD_Y/\cI$ be a regular holonomic $\DD_Y$-module.
The space $Y$ can be decomposed into algebraic varieties $Y_\alpha$
such that $Y=\sqcup Y_\alpha$ and
the space of holomorphic solutions of $\MM$ on $Y_\alpha$
is a locally constant sheaf.
In other words, all holomorphic solutions on $Y_\alpha$
are analytic continuations of holomorphic solutions at a point in $Y_\alpha$ along paths in $Y_\alpha$.
In particular, the number of holomorphic solutions at a point $c \in Y_\alpha$ only depends on $Y_\alpha$ and not $c$.
This decomposition is called a \emph{stratification}.
Stratifications were proved to exist in a broader context in the work of M. Kashiwara \cite{kashiwara-1975}.
For instance, a stratification of the system in the \exref{ex:4-dim-sol} is
\eq{
    \Bigbrk{\CC^2 \setminus V(xy)}
    \,\bigsqcup\,
    \Bigbrk{V(x) \setminus (0,0)}
    \,\bigsqcup\,
    \Bigbrk{V(y) \setminus (0,0)}
    \,\bigsqcup\,
    \brk{0,0}.
}
Here, $V(f)$ is the variety defined by $f=0$.
The number of holomorphic solutions on these stratifications are respectively $4, 1, 0$ and $0$.

Consider the restriction $\cN:=\iota^*\CM$ and let $W \subset Y^\prime$ be the maximum-dimensional stratum of $Y'$.
$W$ is known to be a Zariski open set, and it does not intersect with the singular locus of  $\cN$.
We suppose that the origin $0$ belongs to $W$.
\begin{lemma} \label{lemma:basis-lemma}
Let $S=\{ s_i\}$ be a finite set consisting of monomials in $\pd{}$.
Assume that a set consisting of equivalence classes represented by elements of $S$ is a $\CC$-basis for a $\CC$-vector space $\DD_Y/(\cI + \sum_{i=1}^{\ydim} z_i \DD_Y)$.
Then the set of equivalence classes in $\cN|_W$ represented by elements of the form $1\otimes s_i\in \cN|_W$, for some $s_i\in S$, gives a free basis of $\cN|_W$ as an $\OO_{Y^\prime}$-module.
\end{lemma}
\noindent It follows from this lemma that $S$ gives a basis for the rational
restriction of $\MM$ to $Y^\prime$.
The lemma is well-known folklore in the theory of $\DD$-modules
and we give a constructive and elementary proof in the
appendices \multiref{appendix:basis-lemma, appendix:strata-and-comprehensive-gb}.

Let $\{f_1, \ldots, f_\mu\}$ be a set of generators for a regular holonomic ideal
$\cI$ in $\cD_Y$.
We make a parallel change of variables
$(z_1, \ldots, z_{\ydim})=(z'_1, \ldots, z'_{\ydim})+(c_1, \ldots, c_{\ydim})$
by choosing a numerical vector $(c_1, \ldots, c_{\ydim})$ in $W \subset Y^\prime$
such that we can apply \lemref{lemma:basis-lemma}.
We dispose of the old coordinates $z_i$ and denote the new coordinates $z'_i$ by $z_i$.
A set of monomials $S$ as in \lemref{lemma:basis-lemma} can be obtained by \algref{alg:rest_to_pt}.
Since we do not know the stratum $W$, the chosen variable shift is probabilistic.
If we unluckily choose a point outside of $W$, the method fails.
However, the relative measure of $Y^\prime \setminus W$ is $0$, so the algorithm succeeds with probability $1$.
The following theorem easily follows from \lemref{lemma:basis-lemma}.

\begin{theorem} \label{th:th5}
Let $S$ be as in \lemref{lemma:basis-lemma}, and regard it as a column vector.
There exists an $r \times \mu$ matrix $Q_i$
with entries in $\CR_{Y^\prime}[\pd{1},\ldots,\pd{\ypcodim}]$
and an $r \times r$ matrix $P_i$
with entries in $\CC(z'):=\CC(z_{\ypcodim+1}, \ldots, z_{\ydim})$
such that
\begin{equation} \label{eq:pfaffian-of-restriction}
\pd{i} S = P_i \cdot S \ + :Q_i \cdot [f_1, \ldots, f_\mu]\tr:\big|_{z_1=\ldots=z_{\ypcodim}=0}
\end{equation}
holds in $\CR_{Y^\prime}[\pd{1},\ldots,\pd{\ypcodim}]$.
\end{theorem}

\noindent The expression \eqref{eq:pfaffian-of-restriction} leads to the following Macaulay matrix method for
finding a Pfaffian system associated to the restriction
to $z_1 = \ldots = z_{\ypcodim} = 0$.

\begin{algorithm}{{\rm (Rationally restricted Pfaffian system via the Macaulay matrix)}}
    \label{alg:alg3} \rm \\
    \underline{Input}: The basis $S$. Generators $\{f_1, \ldots, f_\mu\}$ of a regular holonomic
    ideal $\cI$ in $\CD_m$.
    \\
    \underline{Output}: Pfaffian matrices of the restriction
$\CR_{Y^\prime} \otimes_{\CD_{Y^\prime}} \CD_Y/(\cI+z_1 \CD_Y + \ldots + z_{\ypcodim} \CD_Y)$.
\begin{algorithmic}[1]
    \State $\RStd$=$S$.
    \State Call {\tt find\_macaulay\_and\_pfaffian}($\{f_1, \ldots, f_\mu\}$, $\RStd$, $\ypcodim$).
    \State \Return The Pfaffian matrices $P_i$,  $i=m'+1, \ldots, m$.
\end{algorithmic}
Here
{\tt find\_macaulay\_and\_pfaffian}($\{f_1, \ldots, f_n\}$, $\RStd$, $\ypcodim$)
is Algorithm 1 of \cite{Chestnov:2022alh}
with the following modification in steps 2 and 6:
\begin{equation}  \label{eq:restriction-of-MMM}
    M_D := \> \lrsbrk{
        \begin{array}{c|c}
            M_{\Ext} &
            M_{\RStd}
        \end{array}
    } \ \Big|_{z_1= \ldots = z_{\ypcodim}=0}
    \>.
\end{equation}
\end{algorithm}

\noindent The Macaulay matrix $M_D$ is constructed from the coefficients
of certain normally ordered differential operators.
These operators come from acting on the generators of $\cI$ with a set of monomials
${\Mons} := \{ \prod_{d_1+\cdots+d_{\ydim} \leq D, d_i \in \ZZ_0 }
\pd{1}^{d_1} \dots \pd{\ydim}^{d_{\ydim}} \}\>$,
where the integer $D$ is called the degree of the Macaulay matrix.
The columns of $M_D$ are thus indexed by ${\Mons}$.
Setting ${\Ext} := {\Mons} \setminus {\RStd}$, we obtain matrix blocks
$M_{\RStd}$ and $M_{\Ext}$ that are indexed by the monomials in ${\RStd}$ and
${\Ext}$ respectively.
Note that ${\RStd}$ can be a subset of
a standard basis ${\Std}$ of $\cR_m/\cR_m \cI\>$.

With these modifications in comparison to \cite{Chestnov:2022alh}, let us briefly describe the Macaulay matrix method for finding Pfaffian systems.
First, define the matrices $C_{\Ext}$ and $C_{\RStd}$, with entries consisting of $1$'s and $0$'s, by the expression
\eq{
    \pd{i} \, {\RStd} =:
    C_{\Ext} \cdot {\Ext} + C_{\RStd} \cdot {\RStd} \, .
}
Here ${\RStd}$ and ${\Ext}$ are regarded as column vectors.
The Macaulay matrix method then instructs us to solve the linear equation
\begin{equation} \label{eq:MM_equation_for_restriction}
  C_{\Ext} - C \cdot M_{\Ext}|_{z_1=\ldots=z_{\ypcodim}=0}=0
\end{equation}
for an unknown matrix $C$.
The Pfaffian matrix in direction $\pd{i}$ is then given by
\eq{
    P_i(z') = C_{\RStd}-C \cdot M_{\RStd}|_{z_1=\ldots=z_{\ypcodim}=0} \, .
}
See \cite[Section 4]{Chestnov:2022alh} for additional details.

\begin{remark} \rm \label{rem:choice_of_Std}
Suppose that we know the holonomic rank $r$ of the restriction.
One strategy to
find $S$ is to use the probabilistic method of \algref{alg:rest_to_pt},
as explained after \lemref{lemma:basis-lemma}.
A second strategy is to try all possible standard monomials until
the Macaulay matrix method succeeds.
\end{remark}

\begin{example}\rm
We consider a left ideal generated by
\eq{
    &\pd{x} (\theta_x + c_{1}-1) - (\theta_x+\theta_y+a) (\theta_x+b_{1}) \\
    &\pd{y} (\theta_y + c_{2}-1) - (\theta_x+\theta_y+a) (\theta_y+b_{2}) \, ,
}
where $\theta_x = x \pd{x}$ and $\theta_y = y \pd{y}$.
This system annihilates the Appell $F_2$ function (see \cite{encyclopedia} for a definition).
Let us try to find the restriction to $x=0$ by applying
\algref{alg:rest_to_pt} with $\gamma=2, k=4$ and \algref{alg:alg3} with $\gamma=1$ and $r=2$.
We have implemented the algorithms in \soft{Risa/Asir} in the package {\tt mt\_mm.rr} \cite{url-asir}.
The output $\RStd$ (a standard basis for the rational restriction) is
$\sbrk{1, \pd{y}}$ and
the Pfaffian matrix {\tt P2} is constructed by the Macaulay matrix method as follows.
\begin{Verbatim}[fontsize=\footnotesize]
import("mt_mm.rr")$
Ideal = [(-x^2+x)*dx^2+(-y*x)*dx*dy+((-a-b1-1)*x+c1)*dx-b1*y*dy-b1*a,
         (-y^2+y)*dy^2+(-x*y)*dy*dx+((-a-b2-1)*y+c2)*dy-b2*x*dx-b2*a]$
Xvars = [x,y]$
//Rule gives a probabilistic determination of RStd (Std for the restriction)
Rule=[[y,y+1/3],[a,1/2],[b1,1/3],[b2,1/5],[c1,1/7],[c2,1/11]]$
Ideal_p = base_replace(Ideal,Rule);
RStd=mt_mm.restriction_to_pt_(Ideal_p,Gamma=2,KK=4,[x,y] | p=10^8);
RStd=reverse(map(dp_ptod,RStd[0],[dx,dy]));
Id = map(dp_ptod,Ideal,poly_dvar(Xvars))$
MData = mt_mm.find_macaulay(Id,RStd,Xvars | restriction_var=[x]);
//For larger problems, use FiniteFlow instead of find_pfaffian
P2 = mt_mm.find_pfaffian(MData,Xvars,2 | use_orig=1);
\end{Verbatim}
\noindent The output {\tt P2} is
\eq{
    P_2=
    \lrsbrk{
        \begin{array}{cc}
            0&  1 \\
            \frac{    -   {b}_{2}   {a}} {   {y}  (  {y}- 1)}&  \frac{   -      (   {a}+ {b}_{2}+ 1)  {y} + {c}_{2}} {   {y}  (  {y}- 1)} \\
        \end{array}
    } \, ,
}
which appears in the Pfaffian system for the restriction as
$ \pd{y} S = P_2 \cdot S$, $S=[1, \pd{y}]\tr$.

Programs for larger problems are posted at \cite{dataAndProgramsOfThisPaper}.
    \hfill$\blacksquare$
\end{example}

\bigbreak

The method presented above can also be applied to restrictions to irreducible hypersurfaces.
Suppose that $Y^\prime$ is the vanishing locus of an irreducible polynomial $L\in \CC[z]=\CC[z_1, \ldots, z_{\ydim}]$ and that $Y^\prime$ is non-singular.
We then have the following theorem, which is proved in \appref{appendix:proof-of-th6}:

\begin{theorem} \label{th:th6}
Let $S$ be a set as in \lemref{lemma:basis-lemma}
and regard it as a column vector.
There exists an $r \times \mu$ matrix $Q_i$
with entries in $\DD_Y$,
an $r \times r$ matrix $P_i$
with entries in $\CC[z]$
and a polynomial $q_i \in \CC[z]$
such that
\begin{equation} \label{eq:pfaffian-of-restriction-to-L}
q_i \pd{i} S = P_i \cdot S
\ + :Q_i \cdot [f_1, \ldots, f_\mu]\tr: \quad \mathrm{mod}\, L
\end{equation}
holds in $\CD_Y$.
Here $\CD_Y \ni \sum c_\alpha(z) \pd{}^\alpha = 0 \ \mathrm{mod}\,L$ means that
each $c_\alpha(z)$ is divisible by $L$.
\end{theorem}

\noindent We call $P_i/q_i$ the Pfaffian matrix on the codimension-$1$ stratum $W$.

Let us now give a procedure for obtaining $P_i$ and $q_i$.
Our column vector $S$ will be denoted by $\RStd$.
In this context, ${\Mons}$ consists of the monomials in $\pd{}$
appearing in $:\pd{}^\alpha f_j:$ which construct the Macaulay matrix.
Working modulo $L$, the equation \eqref{eq:MM_equation_for_restriction} is now modified to
\eq{
    C_{\Ext} - C \cdot M_{\Ext} \equiv 0
    \quad \text{mod } L \, ,
    \label{eq:mm_ext_eq}
}
where the unknown matrix $C$ is solved for   over the fraction field of the quotient ring $\cO_{Y^\prime}=\CC[z]/L\CC[z]$.
In other words, we find a $C$ with rational function entries
and an $|{\RStd}| \times |{\Ext}|$-dimensional matrix $H$ with {\it polynomial entries}
such that
\begin{equation}  \label{eq:4.12m}
  q_i (C_{\Ext}-C \cdot M_{\Ext}) = H L
\end{equation}
holds, where $q_i$ is a polynomial which is relatively prime to $L$
that cancels the denominators of the elements of $C$.
Note that the entries of $M_{\Ext}$ are polynomial in $z$.
Then (\ref{eq:4.12m}) can be solved by a syzygy computation in the polynomial ring $\CC[z]$.
For example, when $C_{\Ext}=[1,0]$, $M_{\Ext}=\lrsbrk{\begin{array}{cc} m_{11} & m_{12} \\ m_{21} & m_{22} \\ \end{array}}$,
$H=[H_1,H_2]$ and $C=[c_1/q_i,c_2/q_i]$, we may solve
the syzygy equation
\eq{
  q_i \lrsbrk{\begin{array}{c} 1 \\ 0 \\ \end{array}}
  -c_1 \lrsbrk{\begin{array}{c} m_{11} \\ m_{12} \\ \end{array}}
  -c_2 \lrsbrk{\begin{array}{c} m_{21} \\ m_{22} \\ \end{array}}
  -H_1 \lrsbrk{\begin{array}{c} L \\ 0 \\ \end{array}}
  -H_2 \lrsbrk{\begin{array}{c} 0 \\ L \\ \end{array}}
  = \lrsbrk{\begin{array}{c} 0 \\ 0 \\ \end{array}}
}
for the unknowns $H_1, \, H_2, \, c_2, \, c_2$ and $q_i$.
This can be solved via a Gr\"obner basis computation \cite{adams}.
Finally, the matrix
\eq{
    \label{eq:Pi_qi}
    \frac{\widetilde P_i}{q_i}=C_{\RStd}- {C} \cdot M_{\RStd}
}
gives a Pfaffian matrix on the codimension-$1$ stratum $W$.
Note that the entries of $\widetilde P_i/q_i$ belong to the fraction field
of the ring $\cO_{Y^\prime}$.

In \secref{subsec:ex:F4}, we use this procedure to find the restriction to the hypersurface singularity $(x-y)^2-2 (x+y)+1=0$ of the Appell $F_4$ hypergeometric system (see \cite{encyclopedia} for a definition)
and Horn's hypergeometric function
$H_i$, $i=1, \dots, 7$.
As far as we are aware, the results on $H_i$ were not known before.

\section{Examples}\label{sec:examples}
We have explained general algorithms for restriction of Pfaffian systems in \secref{sec:restriction_of_a_pfaffian_system}, and restrictions for regular holonomic $\CD$-modules in
\secref{sec:NTsec1_MMMR}.
Now we apply these methods to Pfaffian systems associated to Feynman integrals and GKZ-hypergeometric systems.
Data and code for these examples can be found in \cite{dataAndProgramsOfThisPaper}.

When no confusion arises, $0$-entries in matrices and vectors are denoted by $\mzero$ in this section.

\subsection{From GKZ to \soft{N}-box}\label{sec:gkz_nbox}
In this first example, we compute the restriction of a GKZ system at the level
of Pfaffian matrices, i.e.~by the holomorphic restriction \algref{alg:holo_rest}.
\subsubsection{Setup}
Let $\MI_{\nu_1,\ldots,\nu_5}$ denote the $2$-loop massless N-box integral family presented in \exref{ex:nbox}.
We study this family as a restriction of a GKZ system:
\eq{
    A = \lrsbrk{
        \begin{array}{ccccc ccccc}
            1 & 1 & 1 & 1 & 1 & 1 & 1 & 1 & 1 & 1\\
            1 & 1 & 1 & \mzero & \mzero & \mzero & \mzero & \mzero & 1 & \mzero\\
            1 & \mzero & \mzero & 1 & 1 & 1 & \mzero & \mzero & 1 & 1\\
            \mzero & \mzero & \mzero & 1 & \mzero & \mzero & 1 & 1 & \mzero & 1\\
            \mzero & 1 & \mzero & \mzero & 1 & \mzero & 1 & \mzero & 1 & \mzero\\
            \mzero & \mzero & 1 & \mzero & \mzero & 1 & \mzero & 1 & \mzero & 1
        \end{array}
    }
    \quad
    \xrightarrow[i \, \neq \, 10]{z_i = 1}
    \quad
    \includegraphicsbox{figures/nbox}
}

While the Lee-Pomeransky polynomial for the GKZ system has generic monomial coefficients $z_1, \ldots, z_{10}$, to match the proper N-box Feynman integral family we must to take the limits
\eq{
    z_{1}, \ldots, z_{9} \to 1
    \quad , \quad
    z_{10} \to \frac{t}{s} \, .
    \label{nbox_limit}
}
As noted in \exref{ex:nbox}, this restriction drops the rank from $9$ to $3$.
The goal of this example is to obtain the $3 \times 3$ Pfaffian matrix $Q_{10}(z_{10})$ by the holomorphic restriction procedure.

Before employing the restrictions at the level of Pfaffians, let us first use the homogeneity property described in \cite[Appendix A]{Chestnov:2022alh} to rescale 6 ($=5+1$, where 5 is the number of integration variables) variables to 1:
\eq{
    z_1, \, z_2, \, z_3, \, z_4, \, z_5, \, z_9 \to 1 \, .
}
In the notation of \cite[Appendix A]{Chestnov:2022alh}, this corresponds to choosing a \emph{simplex} $\sigma = [1,2,3,4,5,9]$.
This rescaling is the maximal restriction which does not change the rank of the GKZ system.

By the method of \cite{Hibi-Nishiyama-Takayama-2017}, we obtain $9$ standard monomials for the rescaled GKZ system:
\eq{
    {\Std} =
    \lrsbrk{
        \pd{8}^2, \> \pd{6} \pd{10}, \> \pd{7} \pd{10}, \>
        \pd{10}^2, \> \pd{6}, \> \pd{7}, \> \pd{8}, \>
        \pd{10}, \> 1
    }\tr \, .
    \label{eq:nbox_std}
}
The Macaulay matrix method of \cite{Chestnov:2022alh} yields a Pfaffian system
\eq{
    \pd{i} {\Std} = P_i \cdot {\Std}
    \quad , \quad
    i = 6, 7, 8, 10.
    \label{eq:pfaff-nbox}
}
We want to apply the restriction $z_6, z_7, z_8 \to 1$ onto this system.

The $3$-dimensional integral basis $\MI$ for the restricted system is taken to be
\eq{
    \lrsbrk{
        \begin{array}{c}
            \MI_{02202} \\
            \MI_{22020} \\
            \MI_{11111}
        \end{array}
    }
    =
    G_0 \cdot \MI
    \, ,
    \label{eq:masters-nbox}
}
where the mass dimension is factored out via $G_0$, a diagonal matrix containing powers of $s$.
By \cite[Proposition A.1]{Chestnov:2022alh}, we can translate this integral basis into a $\cD_Y$-module basis, finding
\eq{
    \label{eq:nbox_restricted_basis}
    &
    \MI_{02202} \quad \to \quad
    \frac{c_1}{\ep(\ep-1)} \pd{10}^2
    \\ & \nonumber
    \MI_{22020}\quad \to \quad
    \frac{c_2}{\ep(\ep-1)}
    \Bigsbrk{
        \ep (2+5\de) (1+2\ep+5\ep\de) +
        2 z_{10} (1+2\ep+5\ep\de) \pd{10} +
        z_{10}^2 \pd{10}^2
    }
    \\ & \nonumber
    \MI_{11111} \quad \to \quad
    \frac{c_3}{\ep(\ep-1)}
    \Bigsbrk{
        (1+\ep+2\de\ep) \pd{10} +
        z_{10} \pd{10}^2 -
        z_7 \pd{7} \pd{10} +
        z_6 \pd{6} \pd{10}
    } \, ,
}
where we recall that $\ep$ and $\de$ are analytic regulators stemming from the integral representation \eqref{eq:nbox_LP_integral}.
The prefactors
\eq{
    c_1 =
    \tfrac{s^{-4 - \ep - 9 \delta \ep} \Gamma\brk{2 - \ep}}{\Gamma\brk{\delta \ep}^6 \Gamma\brk{-3 \ep - 9 \delta \ep} \Gamma\brk{2 + \delta \ep}^3} \, , \>
    c_2 =
    \tfrac{s^{-4 - \ep - 9 \delta \ep} \Gamma\brk{2 - \ep}}{\Gamma\brk{\delta \ep}^6 \Gamma\brk{-3 \ep - 9 \delta \ep} \Gamma\brk{2 + \delta \ep}^3} \, , \>
    c_3 =
    \tfrac{\ep^2 (1 + 2 \ep) s^{-4 - \ep - 9 \delta \ep} \Gamma\brk{2 - \ep}}{\Gamma\brk{\delta \ep}^4 \Gamma\brk{1 - 3 \ep - 9 \delta \ep} \Gamma\brk{1 + \delta \ep}^5}
}
come from the Lee-Pomeransky representation together with factors of $s$ from $G_0$.

\subsubsection{Normal form}

We choose to do the restriction in the order%
\footnote{We have checked that the final result does not depend on the order of the individual limits.}
$z_6$, $z_7$, $z_8$\ .
For the application of the restriction method, we must be careful in checking that the Pfaffian system is in normal form (recall \secref{subsec:Deligne_extension}).
We observe that the system is in normal form w.r.t.~the variables $z_6$ and $z_7$.
However, the system is not logarithmic w.r.t.~$z_8$:
\eq{
    P_{8}\big|_{z_6, z_7 \to 1} =
    \brk{z_8 - 1}^{-2} \> P_{8, -2}\big|_{z_6, z_7 \to 1}
    + \brk{z_8 - 1}^{-1} \> P_{8, -1}\big|_{z_6, z_7 \to 1}
    + \ldots
}
This 2nd order pole can be cured by Moser reduction, in this case a gauge transformation by
\eq{
    G = \diag{\brk{z_8 - 1}^{-1}, 1, \ldots, 1}
    \quad , \quad
    G[P_i] = G^{-1}(P_i \cdot G - \pd{i} G) \, .
}
We check that the transformed residue matrices $G\sbrk{P_{i}}_{-1}$, for $i = 6,7,8$, all have non-resonant spectra.
Hence, the Pfaffian system is now in normal form w.r.t.~all the variables we seek to restrict.
To simplify the notation in the following, let us
\begin{align}
    \text{replace $P_i \to G_1\sbrk{P_i}$ in this example}
    \nonumber
\end{align}
to denote the gauge transformed Pfaffian system.

\subsubsection{Restriction}

The restriction procedure implores us to build the matrix
\eq{
    M =
    \lrsbrk{
        \begin{array}{c}
            B
            \\
            \chline
            R
        \end{array}
    }
}
coming from \eqref{eq:M-def}.
The block matrix $B$ translates between the unrestricted, $9$-dimensional basis and the restricted, $3$-dimensional basis.
This requires us to expand the basis \eqref{eq:nbox_restricted_basis} in terms of the $\Std$ from \eqref{eq:nbox_std}%
\footnote{It can happen that monomials in $\pd{i}$ appear in the restricted basis which are not part of $\Std$. This situation is treated in \cite[Section 3.2]{Chestnov:2022alh}.},
leading to the $3 \times 9$ matrix
\eq{
    B =
    \frac{1}{\ep(\ep-1)}
    \lrsbrk{
        \begin{array}{ccccc ccccc}
            \mzero & \mzero & \mzero & 1 & \mzero & \mzero & \mzero & \mzero & \mzero
            \\[4pt]
            \mzero & \mzero & \mzero & z_{10}^2 & \mzero & \mzero & \mzero & 2z_{10} (1 + 2 \ep + 5 \delta \ep) & \ep(2 + 5 \delta) (1 + 2 \ep + 5 \delta \ep)
            \\[4pt]
            \mzero & z_6 & -z_7 & z_{10} & \mzero & \mzero & \mzero & 1 + \ep + 2 \delta \ep & \mzero
        \end{array}
    } \, ,
}
with columns labeled by $\Std$, and rows labeled by $I$.
The block matrix $R$ is built from the independent rows of the residue matrices:
\eq{
    R =
    \RowReduce{\lrsbrk{
        \begin{array}{c}
            P_{6, -1}
            \\
            P_{7, -1}
            \\
            P_{8, -1}
        \end{array}
    }}
    \ .
    \label{eq:nbox-R}
}
Here we have stacked the matrices on top of each other, row reduced, and finally deleted zero-rows.

Next we insert $M$ into \eqref{eq:gauge-transform}:
\eq{
    \Bigbrk{
        \pd{10} M +
        M \cdot P_{10, 0}
    }
    \cdot M^{-1}
    =
    \lrsbrk{
        \begin{array}{cc}
            Q_{10} & \star
            \\
            \mZero\subsm{6 \times 3} & \star
        \end{array}
        \label{eqn:MGR}
    } \, ,
}
where $P_{10, 0}$ is defined by the limit $P_{10} \big |_{z_6, z_7, z_8 \to
0}\>$, which is a well-defined because we are in normal form.
Sending the analytic regular $\de \to 0$ in $Q_{10}$, we finally obtain:
\eq{
    Q_{10} = \lrsbrk{
        \begin{array}{ccc}
                -\frac{2 \brk{1 + \ep}}{z_{10}} & \mzero & \mzero \\
                \mzero & \mzero & \mzero \\
                -\frac{z_{10} \ep}{6 \brk{1 + z_{10}}} &
                \frac{\ep}{6 z_{10} \brk{1 + z_{10}}} &
                -\frac{z_{10} + 2 \ep}{z_{10} \brk{1 + z_{10}}}
        \end{array}
    } \, .
    \label{eq:deq-nbox-answer}
}
This is the $3$-dimensional Pfaffian matrix for the restricted system.
This result agrees with an independent calculation performed with the IBP software \soft{LiteRed}\cite{Lee:2012cn,Lee:2013mka}.

\subsection{Logarithmic corrections to the one-loop Bhabha box integral}\label{sec:1L_bhabha}

Let us now turn to non-GKZ example to illustrate \algref{alg:log_rest}.

\subsubsection{Setup}

We would like to study logarithmic
corrections in the small-mass expansion of the following topology:
\eq{
    \nonumber
    \includegraphicsbox{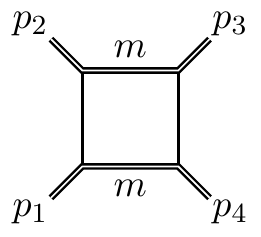}
    \quad
    \xrightarrow{m \to 0}
    \quad
    \includegraphicsbox{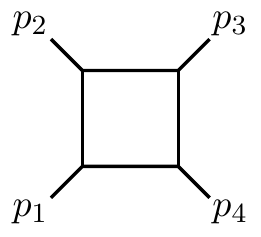}
}
The corresponding family of Feynman integrals in $d = 4-2\ep$ dimensions is
\eq{
    I_{\nu_1\nu_2\nu_3\nu_4}(m) =
    \int
    \frac{\dd^d k / i \pi^{d/2}}
    {
        \bigsbrk{ k^2 - m^2 }^{\nu_1}
        \bigsbrk{ (k+p_1)^2 }^{\nu_2}
        \bigsbrk{ (k+p_{12})^2 - m^2 }^{\nu_3}
        \bigsbrk{ (k+p_{123})^2 }^{\nu_4}
    } \, ,
}
with $p_K := \sum_{k \in K}p_k$ for $K \subset \{1,2,3,4\}$.
The kinematic variables are given by
\eq{
    p_1^2 = p_2^2 = p_3^3 = p_4^2 = m^2
    \quad , \quad
    s = p_{12}^2
    \quad , \quad
    t = p_{23}^2 \, .
    \label{eq:bhabha_kinematics}
}
This family of integrals contributes to Bhabha scattering at $1$-loop order.

\vspace{5pt}

After IBP reduction, we find $5$ master integrals, rescaled by a matrix $G_0$ to become unitless:
\eq{
    \lrsbrk{
        \begin{array}{c}
            I_{0010}(m) \\
            I_{0101}(m) \\
            I_{1010}(m) \\
            I_{0111}(m) \\
            I_{1111}(m)
        \end{array}
    }
    =:
    G_0
    \cdot \MI
    \quad , \quad
    G_0 =
    (-s)^{-\ep} \, \Diag{(-s), \, 1, \, 1, \, (-s)^{-1}, \, (-s)^{-2}} \, .
    \label{1L_bhabha_Y}
}
The master integrals $I$ obey a Pfaffian system
\eq{
    \label{1L_bhabha_pfaffian_system}
    \pd{i} \, \MI = P_i \cdot \MI
    \quad , \quad
    z_1 = \frac{m^2}{-s}
    \quad , \quad
    z_2 = \frac{-t}{-s}
}
in terms of two $5 \times 5$ Pfaffian matrices $P_i = P_i(z_1,z_2)$.

Our goal is to find an approximation to the master integrals $\MI$ which holds in the limit $m \to 0$, or equivalently
\eq{
    z_1 \to 0 \, .
}
The approximation is found by the restriction procedure laid out in \secref{sec:relation_to_Feynman_integrals}.
In particular, we will find and solve \emph{simpler} Pfaffian systems in comparison to \eqref{1L_bhabha_pfaffian_system}, in the sense that they will be of smaller rank and contain one variable less.

\subsubsection{Normal form}

In the basis $\MI$, we find that $P_1$ and $P_2$ respectively have 2nd and 1st order poles at $z_1=0$.
Moreover, the spectrum of the residue matrix of $P_1$ is resonant.
We therefore need to bring the Pfaffian system to normal form before we can
apply the restriction algorithm.
We find that a single gauge transformation
\eq{
    \MI \, \to \, G_1 \cdot \MI
    \quad , \quad
    G_1[P_i] := G_1^{-1}(P_i \cdot G_1 - \pd{i} G_1)
}
by a diagonal matrix
\eq{
    G_1 = \diag{z_1, 1, 1, 1, 1}
}
does the trick. To simplify the notation, we
\begin{align}
    \text{replace $P_i \to G_1\sbrk{P_i}$ in this example} \, .
\end{align}
Now the following expansions in $z_1$ hold true:
\eq{
    P_1 &= \frac{1}{z_1} \> P_{1, -1}(z_2) + O(z_1^0) \\
    P_2 &= z_1^0 \> P_{2, 0}(z_2) + O(z_1) \, ,
    \label{1L_bhabha_G1}
}
with $P_{1, -1}$ having unique eigenvalues $\{0,-\ep\}$, which is a non-resonant spectrum.

Jordan decomposition of the residue matrix $P_{1, -1}$ provides information on how to subdivide the approximated solution to $\MI$.
More precisely, given
\eq{
    \jordan{P_{1, -1}} =
    \lrsbrk{
        \begin{array}{ccccc}
            0      & \mzero & \mzero & \mzero & \mzero \\
            \mzero & 0      & \mzero & \mzero & \mzero \\
            \mzero & \mzero & 0      & \mzero & \mzero \\
            \mzero & \mzero & \mzero & -\ep   & 1      \\
            \mzero & \mzero & \mzero & \mzero & -\ep
        \end{array}
    } \, ,
    \label{eq:bhabha-1}
}
we observe two blocks: a $3 \times 3$ block associated to the eigenvalue $\lambda_1 = 0$ and a $2 \times 2$ block associated to the eigenvalue $\lambda_2 = -\ep$.
As in \eqref{eq:log_singular_I}, we then expect the approximated solution to take the form
\eq{
   \MI\supbrk{\log}(z_1,z_2)
   \quad \overset{z_1 \to 0}{=} \quad
   z_1^{\lambda_1}
   \MI\supbrk{\lambda_1}(z_2)
   \, + \,
   z_1^{\lambda_2}
   \lrsbrk{\MI\supbrk{\lambda_2,0}(z_2) + \MI\supbrk{\lambda_2,1}(z_2) \log(z_1)} \, .
   \label{bhabha_sum_of_eigenspaces}
}
The vector $\MI\supbrk{\lambda_1}$ will be found from the solution to a rank-$3$ Pfaffian system, and the vectors $\MI\supbrk{\lambda_2,0}, \MI\supbrk{\lambda_2,1}$ from the solution to a rank-$2$ Pfaffian system.
The logarithm in $z_1$ appears due to the $"1"$ in the superdiagonal of the $\lambda_2$-eigenvalue block.

\subsubsection{Restriction w.r.t.~eigenvalue \texorpdfstring{$\lambda_1=0$}{}}
\label{sssec:zero}

Here we describe how to find the $5$-dimensional vector $\MI\supbrk{\lambda_1}$ by solving a $3 \times 3$ subsystem in the variable $z_2$ only.
Why is this subsystem of rank $3$?
Recalling \eqref{eqn:restriction_equation}, this question is answered by the following constraints which hold in the limit $z_1 \to 0$:
\eq{
    \rowReduce{P_{2,-1}}
    \cdot
    \MI\supbrk{\lambda_1}
    &:=
    R\supbrk{\lambda_1}
    \cdot
    \MI\supbrk{\lambda_1}
    \\&=
    \lrsbrk{
        \begin{array}{ccccc}
            1      & \mzero & \mzero & \mzero                 & \mzero \\
            \mzero & 1      & \mzero & \frac{z_2 \ep}{1-2\ep} & \mzero
        \end{array}
    }
    \, \cdot \,
    \MI\supbrk{\lambda_1}
    \label{1L_bhabha_rowreduce}
    \\&=
    0 \, .
}
We may regard $\MI\supbrk{\lambda_1}$ as the piece of the full solution vector
$\MI$  which is holomorphic in the massless limit, wherefore we expect every
entry of $\MI\supbrk{\lambda_1}$ to be proportional to $\MI_i(m=0)$, for some
$i = 1, \ldots, 5$, with the massless limit being taken at the \emph{integrand}
level\footnote{
   This corresponds to the so-called ``hard region'' in the expansion by
   regions method~\cite{Beneke:1997zp}.
}.
The first row of \eqref{1L_bhabha_rowreduce} thus states that the massless tadpole $\MI_{0010}(0)$ vanishes, and the 2nd row yields an IBP relation between the massless $t$-channel bubble $\MI_{0101}(0)$ and the massless triangle $\MI_{0111}(0)$.
These two constraints drop the rank from $5$ to $3$.
Summing up, $\MI\supbrk{\lambda_1}$ is of the form
\eq{
    \MI\supbrk{\lambda_1} =
    G_0 \cdot
    \lrsbrk{
        \begin{array}{c}
            \mzero \\
            I_{0101}(0) \\
            I_{1010}(0) \\
            \frac{2\ep-1}{z_2 \ep} I_{0101}(0) \\
            I_{1111}(0)
        \end{array}
    } .
}

Next we find a 3-dimensional basis for the vector space associated to the eigenvalue $\lambda_1=0$.
One possible basis is
\eq{
    \nullSpace{P_{1,-1}}
    =
    \lrsbrk{
        \begin{array}{ccccc}
            \mzero & \frac{z_2 \ep}{2\ep-1} & \mzero & 1      & \mzero \\
            \mzero & \mzero                 & 1      & \mzero & \mzero \\
            \mzero & \mzero                 & \mzero & \mzero & 1
        \end{array}
    } \, ,
    \label{1L_bhabha_nullspace}
}
where the rows of this matrix, viewed as vectors, span the nullspace of $P_{1,-1}$.
The 1st row suggests a linear combination between $\MI_{0101}(0)$ and $\MI_{0111}(0)$ as a basis element.
We could use this as it is, but by virtue of \eqref{1L_bhabha_rowreduce} it will be simpler to rewrite $\MI_{0111}(0)$ in terms of $\MI_{0101}(0)$ up to a non-zero prefactor, which we are free to omit in our choice of basis.
The 2nd and 3rd rows of \eqref{1L_bhabha_nullspace} respectively instruct us to pick the $s$-channel bubble $I_{1010}(0)$ and the box $I_{1111}(0)$ as the last two basis elements.
Summing up, we have the following basis for the restricted Pfaffian system associated to the eigenvalue $\lambda_1$:
\eq{
    \Mi\supbrk{\lambda_1}
    :=
    (-s)^{\ep}
    \lrsbrk{
        \begin{array}{c}
            I_{0101}(0) \\
            I_{1010}(0) \\
            (-s)^2 I_{1111}(0)
        \end{array}
    } \, ,
}
with prefactors of $s$ stemming from $G_0$ in \eqref{1L_bhabha_Y}.
The $3 \times 5$ matrix relating the bases $\Mi\supbrk{\lambda_1}$ and $\MI\supbrk{\lambda_1}$ is thus
\eq{
    B\supbrk{\lambda_1} =
    \lrsbrk{
        \begin{array}{ccccc}
            \mzero & 1      & \mzero & \mzero & \mzero \\
            \mzero & \mzero & 1      & \mzero & \mzero \\
            \mzero & \mzero & \mzero & \mzero & 1
        \end{array}
    }
    \quad , \quad
    \Mi\supbrk{\lambda_1} = B\supbrk{\lambda_1} \cdot \MI\supbrk{\lambda_1} \, .
    \label{eq:B-lambda1}
}
Using \eqref{eq:M-def}, we proceed to build the invertible matrix
\eq{
    \def\arraystretch{1.2}
    M\supbrk{\lambda_1} =
    \lrsbrk{
        \begin{array}{c}
            B\supbrk{\lambda_1} \\
            \hline
            R\supbrk{\lambda_1}
        \end{array}
    } \, ,
    \label{eq:1L_bhabha_M0}
}
which has the property
\eq{
    \lrsbrk{
        \begin{array}{c}
            \Mi\supbrk{\lambda_1} \\
            \mzero \\
            \mzero
        \end{array}
    }
    =
    M\supbrk{\lambda_1} \cdot \MI\supbrk{\lambda_1} \, .
}
The $3 \times 3$ Pfaffian matrix $Q_2\supbrk{\lambda_1}$ associated to the vector $\Mi\supbrk{\lambda_1}$ is now found via \eqref{eq:gauge-transform}:
\eq{
    \Bigbrk{
        \pd{2} M\supbrk{\lambda_1}
        +
        M\supbrk{\lambda_1} \cdot P_{2,0}
    }
    \cdot  (M\supbrk{\lambda_1})^{-1}
    =
    \lrsbrk{
        \begin{array}{cc}
            Q_2\supbrk{\lambda_1}    & \star \\
            \mZero\subsm{2 \times 3} & \star
        \end{array}
    } \, ,
}
with $P_{2,0}$ defined in \eqref{1L_bhabha_G1}.
Explicitly, the Pfaffian system restricted to $z_1 = 0$ w.r.t.\ the eigenvalue $\lambda_1=0$ is
\eq{
    \pd{2} \Mi\supbrk{\lambda_1} = Q_2\supbrk{\lambda_1} \cdot \Mi\supbrk{\lambda_1}
    \quad , \quad
    Q_2\supbrk{\lambda_1} =
    \lrsbrk{
        \begin{array}{ccc}
            -\frac{\ep}{z_2} & \mzero & \mzero \\
            \mzero           & \mzero & \mzero \\
            \frac{2(1-2\ep)}{z_2^2(z_2+1)} &
            \frac{2(2\ep-1)}{z_2  (z_2+1)} &
            \frac{-z_2-\ep-1}{z_2 (z_2+1)}
        \end{array}
    } \, .
}
This agrees with the system one would find by performing IBP reduction directly on the massless integrals $\Mi\supbrk{\lambda_1}$, as can be checked with \soft{LiteRed}.
It is standard to solve such a system by passing to a canonical basis \cite{Henn:2013pwa} (see \cite{Henn:2014qga} for this specific example).

\subsubsection{Restriction w.r.t.~eigenvalue \texorpdfstring{$\lambda_2=-\ep$}{}}
\label{subsec:eigen_epsilon}
Next we determine the $5$-dimensional vectors $\MI\supbrk{\lambda_2,0}, \MI\supbrk{\lambda_2,1}$ by solving a rank-$2$ Pfaffian system.
The rank is $2$ because of \eqref{eq:N_lambda_constraints}, which imposes the following constraint on the solution vector $\MI\supbrk{\lambda_2,0}$:
\eq{
    \RowReduce{(P_{1,-1} - \lambda_2)^2}
    \cdot \MI\supbrk{\lambda_2,0}
    &:=
    R\supbrk{\lambda_2}
    \cdot \MI\supbrk{\lambda_2,0}
    \\&=
    \lrsbrk{
        \begin{array}{ccccc}
            1      & \mzero & \mzero & \frac{z_2 \ep}{1-\ep} & \frac{z_2 \ep}{2(1-\ep)} \\
            \mzero & 1      & \mzero & \mzero                & \mzero                   \\
            \mzero & \mzero & 1      & \mzero                & \mzero
        \end{array}
    }
    \cdot \MI\supbrk{\lambda_2,0}
    \\&=
    0 \, .
    \label{1L_bhabha_Psq_constraint}
}
Since there are $3$ relations imposed on the entries of $\MI\supbrk{\lambda_2,0}$, the rank drops from $5$ to $2$.
Let us further note that, according to \eqref{log_singular_constraint}, we obtain $\MI\supbrk{\lambda_2,1}$ for free once we know $\MI\supbrk{\lambda_2,0}$:
\eq{
    (P_{1,-1} - \lambda_2) \cdot \MI\supbrk{\lambda_2,0} = \MI\supbrk{\lambda_2,1} \, .
}

Let
\eq{
    \Mi\supbrk{\lambda_2} =
    \lrsbrk{
        \begin{array}{c}
            \Mi_{1}\supbrk{\lambda_2} \\
            \Mi_{2}\supbrk{\lambda_2}
        \end{array}
    } \,
}
denote the basis of the Pfaffian system associated to the eigenvalue
$\lambda_2$.  Unlike the zero eigenvalue case of~\secref{sssec:zero}, here we
no longer have a natural choice for the basis matrix $B\supbrk{\lambda_2}$ that
relates $\Mi\supbrk{\lambda_2}$ with $\MI\supbrk{\lambda_2}$. Luckily, when
buidling the transformation matrix $M\supbrk{\lambda_2}$ from
\eqref{eq:M-def}, we only need to ensure that the our choice of
$B\supbrk{\lambda_2}$ together with the $R\supbrk{\lambda_2}$
of \eqref{1L_bhabha_Psq_constraint} produce a full rank square matrix, so
let us take\footnote{
    See also the procedure described in~\appref{sssec:choice-B-matrix} that
    gives an alternative choice of the basis matrix.
}
\eq{
    B\supbrk{\lambda_2} \defas
    \lrsbrk{
        \begin{array}{ccccc}
            \mzero & \mzero & \mzero & \frac{z_2 \ep}{\ep-1} & \frac{z_2 \ep}{2 (\ep-1)} \\
            \mzero & \mzero & \mzero & 1 & \mzero \\
        \end{array}
    }
    \> .
}

Using \eqref{eq:M-def} to write
\eq{
    \def\arraystretch{1.2}
    M\supbrk{\lambda_2} =
    \lrsbrk{
        \begin{array}{c}
            B\supbrk{\lambda_2} \\
            \chline
            R\supbrk{\lambda_2}
        \end{array}
    }
    \implies
    \MI\supbrk{\lambda_2,0}
    =
    \bigbrk{M\supbrk{\lambda_2}}^{-1}
    \cdot
    \lrsbrk{
        \begin{array}{c}
            \Mi\supbrk{\lambda_2} \\[-5pt]
            \mzero         \\[-5pt]
            \mzero         \\[-5pt]
            \mzero
        \end{array}
    }
    =
    \lrsbrk{
        \begin{array}{c}
            \Mi_1\supbrk{\lambda_2} \\
            \mzero           \\
            \mzero           \\
            \Mi_2\supbrk{\lambda_2} \\
            \frac{2(\ep-1)}{z_2 \ep} \Mi_2\supbrk{\lambda_2} -
            2 \Mi_2\supbrk{\lambda_2}
        \end{array}
    } \, ,
}
we then insert $M\supbrk{\lambda_2}$ into \eqref{eq:gauge-transform} to obtain the $2 \times 2$ Pfaffian matrix associated to $\Mi\supbrk{\lambda_2}$:
\eq{
    \Bigbrk{
        \pd{2} M\supbrk{\lambda_2}
        +
        M\supbrk{\lambda_2} \cdot P_{2,0}
    }
    \cdot  \bigbrk{M\supbrk{\lambda_2}}^{-1}
    =
    \lrsbrk{
        \begin{array}{cc}
            Q_2\supbrk{\lambda_2}    & \star \\
            \mZero\subsm{3 \times 2} & \star
        \end{array}
    } \, .
}
Explicitly, the Pfaffian system restricted to $z_1 = 0$ w.r.t.~the eigenvalue $\lambda_2$ is
\eq{
    \pd{2} \, \Mi\supbrk{\lambda_2} =
    Q_2\supbrk{\lambda_2} \cdot \Mi\supbrk{\lambda_2}
    \quad , \quad
    Q_2\supbrk{\lambda_2} =
    \lrsbrk{
        \begin{array}{cc}
            \mzero              & \mzero \\
            \frac{\ep-1}{z_2^2} & \frac{-1}{z_2}
        \end{array}
    } \, .
    \label{eq:final_pfaffian}
}
The general solution is
\eq{
    \Mi_1\supbrk{\lambda_2} = c_1(\ep)
    \quad , \quad
    \Mi_2\supbrk{\lambda_2} = \frac{c_1(\ep) (1-\ep) \log (z_2)}{z_2} + \frac{c_2(\ep)}{z_2} \, .
}
The constants $c_{1}(\ep),c_{2}(\ep)$ will be fixed by matching with the full solution vector $\MI$.

\subsubsection{Full result}

Let us combine the solution vectors from the previous two sections.
Given $\Mi\supbrk{\lambda_1}$ and $\Mi\supbrk{\lambda_2}$, the vectors $\MI\supbrk{\lambda,\bullet}$ from \eqref{bhabha_sum_of_eigenspaces} are determined by
\eq{
    \MI\supbrk{\lambda_1}
    &=
    \bigbrk{ M\supbrk{\lambda_1} }^{-1}
    \cdot
    \lrsbrk{
        \begin{array}{c}
            \Mi\supbrk{\lambda_1} \\[-5pt]
            \mzero      \\[-5pt]
            \mzero
        \end{array}
    }
    \\
    \MI\supbrk{\lambda_2,0}
    &=
    \bigbrk{ M\supbrk{\lambda_2} }^{-1}
    \cdot
    \lrsbrk{
        \begin{array}{c}
            \Mi\supbrk{\lambda_2} \\[-5pt]
            \mzero         \\[-5pt]
            \mzero         \\[-5pt]
            \mzero
        \end{array}
    }
    \\[12pt]
    \MI\supbrk{\lambda_2,1}
    &=
    (P_{1,-1} - \lambda_2) \cdot \MI\supbrk{\lambda_2,0} \, .
}
Reinstating prefactors and gauge transformations, the expression \eqref{bhabha_sum_of_eigenspaces} becomes
\eq{
    \lrsbrk{
    \begin{array}{c}
        I_{0010}(m) \\[5pt]
        I_{0101}(m) \\[5pt]
        I_{1010}(m) \\[5pt]
        I_{0111}(m) \\[5pt]
        I_{1111}(m)
    \end{array}
    }
    \, \overset{m \to 0}{=} \,
    z_1^{\lambda_1}
    G_0 \cdot G_1 \cdot
    \lrsbrk{
        \begin{array}{c}
            0                                \\[5pt]
            I_{0101}(0)                      \\[5pt]
            I_{1010}(0)                      \\[5pt]
            \frac{1-2\ep}{t \ep} I_{0101}(0) \\[5pt]
            I_{1111}(0)
        \end{array}
    }
    \, + \,
    z_1^{\lambda_2}
    G_0 \cdot G_1 \cdot
    \lrsbrk{
        \begin{array}{c}
            \frac{c_1}{\ep-1}                                                \\[5pt]
            0                                                                \\[5pt]
            0                                                                \\[5pt]
            \frac{c_1 \log(z_2/z_1) + c_2}{z_2}                              \\[5pt]
            \frac{-2}{\ep} \frac{c_1 [\ep \log(z_2/z_1) - 1] + \ep c_2}{z_2}
        \end{array}
    } \, ,
}
where $c_1 \to \frac{c_1}{\ep-1}$ has been rescaled to simplify the expression.
$c_1$ can be found by matching against the exact result for the tadpole:
\eq{
    \nonumber
    &\text{LHS:} \quad
    I_{0010}(m)
    =
    -\Gamma(\ep-1) (m^2)^{1-\ep} \\
    \label{1L_bhabha_sol}
    &\text{RHS:} \quad
    (-s)^{1-\ep} \cdot z_1^{1-\ep} \cdot \frac{c_1}{\ep-1}
    =
    (-s)^{1-\ep} \cdot \lrbrk{ \frac{m^2}{-s} }^{1-\ep} \cdot \frac{c_1}{\ep-1} \\
    &\implies
    c_1 = (1-\ep) \Gamma(\ep-1) \, .
    \nonumber
}
The constant $c_2$ can be found order by order in $\ep$ by matching with the LHS at a particular phase space point.
Expanding%
\footnote{
    Note that leading $\ep$-poles for $I_{1111}(m)$ and $I_{1111}(0)$ are $\ep^{-2}$ and $\ep^{-1}$ respectively.
    We therefore start the Laurent series for $c_2$ at $O(\ep^{-2})$ to cancel the $\ep^{-2}$ poles on the RHS.
}
$c_2 = \sum_{i=-2}^2 c_{2,i} \ep^i$ and matching with $I_{1111}(m)$ at $(s,t,m^2) = (-1, -9, 10^{-3})$ using \soft{AMFlow} \cite{Liu:2022chg}, we obtain
\eq{
    (c_{2,-2}, c_{2,-1}, c_{2,0}, c_{2,1}, c_{2,2}) =
    (1, -0.588997, -7.12155, -4.69031, -20.4068).
}
After inserting the solutions for $c_{1},c_{2}$ into the LHS of \eqref{1L_bhabha_sol}, we find good agreement at other phase space points.
For instance, at $(s,t,m^2) = (-1,-8,10^{-3})$ we find that, up to $O(\ep^3)$, the LHS agrees with $I_{0111}(m)$ within $0.01\%$ accuracy, and with $I_{1111}(m)$ within $0.001\%$ accuracy.

\subsection{From massive Bhabha to the massless double-box}

Here we apply the Pfaffian-level restriction protocol to a more complicated two-loop diagram to illustrate that the method is computationally cheap.

\subsubsection{Setup}
We study the massless limit of a planar 2-loop integral family
contributing to Bhabha scattering in $d = 4 - 2\ep$ dimensions:
\begin{align*}
    \includegraphicsbox{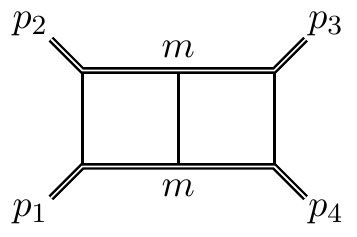}
    \quad
    \xrightarrow{m \to 0}
    \quad
    \includegraphicsbox{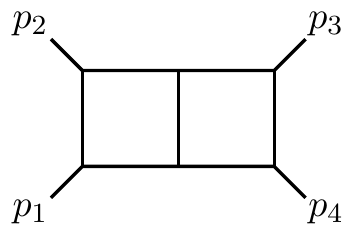}
\end{align*}
Let
\eq{
    & \MI_{\nu_1 \ldots \nu_9}(m) =
    \int
    \frac{\dd^d k_1 \dd^d k_2}
    {
        D_1^{\nu_1}
        \cdots
        D_9^{\nu_9}
    }
    \\[9pt] \nonumber
    & D_1 = k_1^2  - m^2
    \, , \,
    D_2 = (k_1 + p_1 + p_2)^2 - m^2
    \, , \,
    D_3 = k_2^2 - m^2 \\
    & D_4 = (k_2 + p_1 + p_2)^2 - m^2
    \, , \,
    D_5 = (k_1+p_1)^2
    \, , \,
    D_6 = (k_1-k_2)^2  \\ \nonumber
    & D_7 = (k_2-p_3)^2
    \, , \,
    D_8 = (k_2+p_1)^2
    \, , \,
    D_9 = (k_1-p_3)^2 \, ,
}
with $D_8$ and $D_9$ added because of irreducible scalar products.
The kinematics are as in \eqref{eq:bhabha_kinematics}.
Solutions for this integral family can be found in \cite{Henn:2013woa} (see also \cite{Duhr:2021fhk}).

As in the 1-loop example, we write two dimensionless variables
\eq{
    z_1 = \frac{m^2}{-s}
    \quad , \quad
    z_2 = \frac{s}{t} \, ,
}
and study the massless limit $m^2 \to 0$, i.e. the restriction
\eq{
    z_1 \to 0 \, .
}
For simplicity, in this section we only illustrate the holomorphic restriction, i.e.~we do not compute logarithmic corrections.

Inspired by \cite{Henn:2013woa}, we choose the following unrestricted basis:
\eq{
    G_0 \cdot \MI(m) = \bigsbrk{
        &
        \mism{2 0 2 0 0 0 0 0 0} \, , \,
        \mism{0 0 0 0 2 2 1 0 0} \, , \,
        \mism{0 0 1 0 2 2 0 0 0} \, , \,
        \mism{1 2 2 0 0 0 0 0 0} \, , \,
        \mism{0 1 2 0 0 2 0 0 0} \, , \,
        \mism{0 2 2 0 0 1 0 0 0} \, , \,
        \nln
        &
        \mism{0 0 1 1 1 2 0 0 0} \, , \,
        \mism{0 1 2 0 1 1 0 0 0} \, , \,
        \mism{1 0 0 0 1 1 2 0 0} \, , \,
        \mism{1 2 1 2 0 0 0 0 0} \, , \,
        \mism{0 1 1 0 1 1 1 0 0} \, , \,
        \mism{0 1 1 0 1 2 1 0 0} \, , \,
        \nln
        &
        \mism{0 2 1 0 1 1 1 0 0} \, , \,
        \mism{0 2 2 0 1 1 1 0 0} \, , \,
        \mism{1 0 1 1 1 1 0 0 0} \, , \,
        \mism{1 0 1 2 1 1 0 0 0} \, , \,
        \mism{1 0 2 0 1 1 1 0 0} \, , \,
        \mism{1 0 1 0 1 1 2 0 0} \, , \,
        \nln
        &
        \mism{1 1 0 0 1 1 2 0 0} \, , \,
        \mism{1 1 1 0 1 1 1 0 0} \, , \,
        \mism{1 1 1 0 1 1 1 -1 0} \, , \,
        \mism{1 1 1 1 1 1 1 0 0} \, , \,
        \mism{1 1 1 1 1 1 1 -1 0}
    }\tr
    \, ,
    \label{eq:masters-bhabha}
}
where $G_0$ is a diagonal matrix containing factors of $s$ which render the master integrals unitless.
We also choose a basis $\Mi$ for the restricted system, which is associated massless double-box topology $I_{1 \ldots 1}(m=0)$, where it is understood that $m=0$ is taken at the level of the integrand.
The $8 \times 23-$dimensional binary matrix $B$ relating the two bases as
\eq{
    \Mi = B \cdot I(m=0)
    \label{eq:bhabha_J}
}
is given by
\eq{
    B =
    \lrsbrk{
        \tiny
        \begin{array}{ccccccccccccccccccccccc}
            \mzero & 1 & \mzero & \mzero & \mzero & \mzero & \mzero & \mzero & \mzero & \mzero & \mzero & \mzero & \mzero & \mzero & \mzero & \mzero & \mzero & \mzero & \mzero & \mzero & \mzero & \mzero & \mzero \\
            \mzero & \mzero & \mzero & \mzero & 1 & \mzero & \mzero & \mzero & \mzero & \mzero & \mzero & \mzero & \mzero & \mzero & \mzero & \mzero & \mzero & \mzero & \mzero & \mzero & \mzero & \mzero & \mzero \\
            \mzero & \mzero & \mzero & \mzero & \mzero & \mzero & 1 & \mzero & \mzero & \mzero & \mzero & \mzero & \mzero & \mzero & \mzero & \mzero & \mzero & \mzero & \mzero & \mzero & \mzero & \mzero & \mzero \\
            \mzero & \mzero & \mzero & \mzero & \mzero & \mzero & \mzero & \mzero & \mzero & 1 & \mzero & \mzero & \mzero & \mzero & \mzero & \mzero & \mzero & \mzero & \mzero & \mzero & \mzero & \mzero & \mzero \\
            \mzero & \mzero & \mzero & \mzero & \mzero & \mzero & \mzero & \mzero & \mzero & \mzero & 1 & \mzero & \mzero & \mzero & \mzero & \mzero & \mzero & \mzero & \mzero & \mzero & \mzero & \mzero & \mzero \\
            \mzero & \mzero & \mzero & \mzero & \mzero & \mzero & \mzero & \mzero & \mzero & \mzero & \mzero & \mzero & \mzero & \mzero & \mzero & \mzero & \mzero & \mzero & 1 & \mzero & \mzero & \mzero & \mzero \\
            \mzero & \mzero & \mzero & \mzero & \mzero & \mzero & \mzero & \mzero & \mzero & \mzero & \mzero & \mzero & \mzero & \mzero & \mzero & \mzero & \mzero & \mzero & \mzero & \mzero & \mzero & \mzero & 1 \\
            \mzero & \mzero & \mzero & \mzero & \mzero & \mzero & \mzero & \mzero & \mzero & \mzero & \mzero & \mzero & \mzero & \mzero & \mzero & \mzero & \mzero & \mzero & \mzero & \mzero & \mzero & 1 & \mzero \\
        \end{array}
    }
    \, .
    \label{eq:bhabha_B}
}

We have the following Pfaffian data for the massive and massless systems:
\eq{
    \def\arraystretch{1.2}
    \begin{array}{cccl}
        \text{System} & \text{Matrices} & \text{Rank} & \text{Variables}
        \\
        \hline
        \text{Massive} & \brc{P_1, P_2} & 23 & \brc{z_1, z_2}
        \\
        \text{Massless} & \brc{Q_2} & 8 & \brc{z_2}
    \end{array}
    \label{eq:pfaff-bhabha}
}
Our goal is to derive the massless differential equation matrix $Q_2$ as a restriction of the massive system $\brc{P_1, P_2}$.

\subsubsection{Normal form}
The massive Pfaffian system~\eqref{eq:pfaff-bhabha} for the
basis~\eqref{eq:masters-bhabha} is not in normal form. The first issue is
that the Pfaffians $P_i$ have higher order singularities in the restriction limit:
\eq{
    P_1 &=
    z_1^{-3} \> P_{1, -3} +
    z_1^{-2} \> P_{1, -2} +
    z_1^{-1} \> P_{1, -1} +
    \ldots
    \\[5pt]
    P_2 &=
    z_1^{-2} \> P_{2, -2} +
    z_1^{-1} \> P_{2, -1} +
    P_{1, 0} +
    \ldots \ .
}
To cure it, we perform Moser reduction via two gauge transformations by
\eq{
    G_1 = \diag{%
        \underbrace{z_1, \ldots, z_1}_{\text{8 times}},
        1, \ldots, 1%
    }%
    \quad , \quad
    G_2 = \diag{%
        \underbrace{z_1, \ldots, z_1}_{\text{12 times}},
        1, \ldots, 1%
    }
}
whereby $P_1$ becomes logarithmically singular at $z_1=0$,
\eq{
    &G_{21}\sbrk{P_1} = z_1^{-1} \> G_{21}\sbrk{P_{1}}_{-1} + \ldots
    \ .
    \label{eq:moser-bhabha}
}
Here we introduced the short-hand notation
\eq{
    G_{i \ldots j} \defas G_i \cdot \ldots \cdot G_j
    \ .
}
The transformations~\eqref{eq:moser-bhabha} still leave $G_{21}\sbrk{P_2}$ with an unwanted singularity at $z_1 = 0$:
\eq{
    G_{21}\sbrk{P_2} = z_1^{-1} \> G_{21}\sbrk{P_{2}}_{-1} + G_{21}\sbrk{P_{2}}_{0} + \ldots \ .
}
This issue arises because the residue matrix $G_{21}\sbrk{P_{1}}_{-1}$ is resonant, meaning that its spectrum contains eigenvalues with integral difference:
\eq{
    \Spec{G_{21}\sbrk{P_{2}}_{-1}}
    &=
    \bigbrc{
        -2, \> -2, \> -2, \> -1, \> -1, \> 0, \> 0, \> 0
    }
    \nln
    & \cup \> \bigbrc{
        -2 - \ep, \> -2 - \ep, \> -1 - \ep, \> -1 - \ep, \> -\ep, \> -\ep, \> -\ep
    }
    \nln
    & \cup \> \bigbrc{
        -3 - 2 \ep, \> -2 - 2 \ep, \> -2 - 2 \ep, \> -1 - 2 \ep, \> -1 - 2 \ep, \> -1 - 2 \ep, \> -2 \ep
    }
    \nln
    & \cup \> \bigbrc{
        -2 - 4 \ep
    }
    \ .
    \label{eq:resonance-bhabha}
}
These resonances are cured via an additional gauge transformation by a
matrix $G_3$, resulting in the non-resonant spectrum that reads
\eq{
    \Spec{G_{321}\sbrk{P_{1}}_{-1}}
    = \bigbrc{
        0,
        \ldots,
        - \ep,
        \ldots,
        -2 \ep,
        \ldots,
        -2 - 4 \ep,
        \ldots
    }
    \, ,
    \label{eq:spec-bhabha}
}
where we only wrote unique eigenvalues.
We review the algorithmic procedure to construct such a $G_3$ in
\appref{subsec:Moser_Reduction}~\footnote{
    In our basic computer implementation, we perform sequential Jordan
    decompositions of the residue matrix, which can be computed in a matter of
    seconds using \soft{Mathematica} on a laptop.
}.

With this final composition of gauge transformations applied to our Pfaffian
system, we now have that
\eq{
    G_{321}\sbrk{P_2} = G_{321}\sbrk{P_{2}}_{0} + \ldots \ ,
}
which means that the system is in normal form.

\subsubsection{Restriction}
Given our choice of restricted basis, the matrix $R$ from \eqref{eq:M-def} is simply
\eq{
    R =
    \RowReduce{G_{321}\sbrk{P_{1}}_{-1}}
    =
    \lrsbrk{
        \begin{array}{c|c}
            \mZero\subsm{15 \times 8} & \mId\subsm{15 \times 15}
        \end{array}
    }
    \ .
}
Had we chosen a restricted basis which was not a subset of $\MI(m=0)$ (cf.~\eqref{eq:masters-bhabha}), say a Laporta basis where all non-zero propagator powers equal to $\pm 1$, then $R$ would generally depend on $z_2$ and $\ep$.
Given the matrix $B$ from \eqref{eq:bhabha_B}, we are now ready to build the restriction transformation matrix~\eqref{eq:M-def}:
\eq{
    \def\arraystretch{1.2}
    M = \lrsbrk{
        \begin{array}{c}
            B \cdot G_{321}
            \\
            \chline
            R
        \end{array}
    }
    \, .
}
The upper matrix block encodes information about the restricted basis, and the lower block contains linear relations (IBPs) which hold in the massless limit.

We finally compute the gauge transformation~\eqref{eq:gauge-transform}
\eq{
    \Bigbrk{
        \pd{2} M
        + M \cdot G_{321}\sbrk{P_{2}}_{0}
    } \cdot M^{-1}
    =
    \lrsbrk{
        \begin{array}{cc}
            Q_2 & \star
            \\
            \mZero\subsm{15 \times 8} & \star
        \end{array}
    }
    \ ,
}
and obtain the $8$-dimensional restricted Pfaffian matrix $Q_2$ associated to the massless double-box topology.
In other words, $Q_2$ satisfies $\pd{2} \Mi = Q_2 \cdot \Mi$, with $J$ given in \eqref{eq:bhabha_J}.
We have verified the form of $Q_2$ by an independent calculation using the IBP software \soft{Kira} \cite{Klappert:2020nbg}.

\subsection{Unequal to equal-mass sunrise}

This example studies symmetry relations of Feynman integrals from the point of view of Pfaffian-level restriction.

\label{subsec: Unequal to equal mass sunrise}
\subsubsection{Setup}
We examine the equal-mass limit of the 2-loop 3-mass sunrise integral
family:
\begin{align*}
    \includegraphicsbox{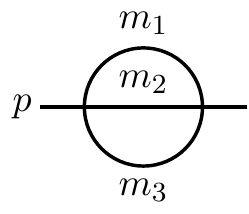}
    \quad
    \xrightarrow{m_i \to m}
    \quad
    \includegraphicsbox{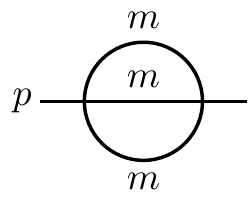}
\end{align*}
The family of momentum space Feynman integrals is
\eq{
    \MI_{\nu_1\nu_2\nu_3}(m_1,m_2,m_3) =
    \int
    \frac{\dd^d k_1 \dd^d k_2}
    {
        \bigsbrk{ k_1^2 - m_1^2         }^{\nu_1}
        \bigsbrk{ k_2^2 - m_2^2         }^{\nu_2}
        \bigsbrk{ (k_1 + k_2 + p)^2 - m_3^2  }^{\nu_3}
    } \, .
    \label{eq:sunrise_momentum_space}
}
We choose the following $3$ dimensionless kinematic variables:
\eq{
    z_1 = \frac{m_1^2}{s}
    , \quad
    z_2 = \frac{m_2^2}{s}
    , \quad
    z_3 = \frac{m_3^2}{s}
    , \quad
    \text{where}
    \quad
    s = p^2 \, .
}
We are interested in the equal-mass limit, that is
\eq{
    m_i \to m,
    \quad
    z_i \to z \defas \frac{m^2}{s},
    \quad
    \text{for $i = 1, 2, 3$} \, .
    \label{eq:equal-mass-limit}
}
The details of the two Pfaffian systems before and after restriction are as follows:
\eq{
    \def\arraystretch{1.2}
    \begin{array}{cccl}
        \text{System} & \text{Matrices} & \text{Rank} & \text{Variables}
        \\
        \hline
        \text{Unequal mass} & \brc{P_1, P_2, P_3} & 7 & \brc{z_1, z_2, z_3}
        \\
        \text{Equal mass} & \brc{Q_z} & 3 & \brc{z}
    \end{array}
    \label{eq:pfaff-sunrise}
}
For the unequal-mass integral family we choose the following basis of 
master integrals:
\eq{
    G_0 \cdot \MI \brk{m_1,m_2,m_3} = \bigsbrk{
        \MI_{011} , \,
        \MI_{101} , \,
        \MI_{110} , \,
        \MI_{111} , \,
        \MI_{211} , \,
        \MI_{121} , \,
        \MI_{112}
    }\tr \, ,
    \label{eq:masters-sunrise}
}
where, as before, $G_0$ is a diagonal matrix containing factors of $s$ which
render the integrals unitless.
In the equal-mass case, our target basis reads
\eq{
    \Mi(m) = \bigsbrk{
        \MI_{011}(m,m,m) , \,
        \MI_{111}(m,m,m) , \,
        \MI_{211}(m,m,m)
    }\tr \, ,
}
so we introduce the following binary matrix 
\eq{
    B = \lrsbrk{
        \begin{array}{ccc ccc c}
            1 & \mzero & \mzero & \mzero & \mzero & \mzero & \mzero
            \\
            \mzero & \mzero & \mzero & 1 & \mzero & \mzero & \mzero
            \\
            \mzero & \mzero & \mzero & \mzero & 1 & \mzero & \mzero
        \end{array}
    } 
    \label{eq:basis-sunrise}
}
relating the bases by $J(m) = B \cdot I(m,m,m)$.

\subsubsection{Restriction}
The unequal-mass Pfaffian system turns out to have the following denominators:
\eq{
    &\fun{Denominator}{P_i} = \bigbrc{z_1, z_2, z_3,
        \\
        &
        \quad 1 - 4 z_1 + 6 z_1^2 - 4 z_1^3 + z_1^4 - 4 z_2 + 4 z_1 z_2 + 4 z_1^2 z_2 - 4 z_1^3 z_2 +
        \nln
        & \qquad
        6 z_2^2 + 4 z_1 z_2^2 + 6 z_1^2 z_2^2 - 4 z_2^3 - 4 z_1 z_2^3 + z_2^4 - 4 z_3 + 4 z_1 z_3 + 4 z_1^2 z_3 -
        \nln
        & \qquad
        4 z_1^3 z_3 + 4 z_2 z_3 - 40 z_1 z_2 z_3 + 4 z_1^2 z_2 z_3 + 4 z_2^2 z_3 + 4 z_1 z_2^2 z_3 - 4 z_2^3 z_3 +
        \nln
        & \qquad
        6 z_3^2 + 4 z_1 z_3^2 + 6 z_1^2 z_3^2 + 4 z_2 z_3^2 + 4 z_1 z_2 z_3^2 + 6 z_2^2 z_3^2 - 4 z_3^3 - 4 z_1 z_3^3 -
        4 z_2 z_3^3 + z_3^4
        \nn
    } \, .
}
The last polynomial reduces to $\brk{z-1}^3 \brk{9z-1}$ under the
identification~\eqref{eq:equal-mass-limit}, so none of the denominators vanish
in the equal-mass limit! This example is hence not a restriction to a singular
locus.
In particular, we have no residue matrix to utilize, so we need to slightly modify the restriction procedure to derive the equal-mass rank-3 Pfaffian system.

To this end, we begin by observing that some of the
MIs~\eqref{eq:masters-sunrise} become identical in the equal-mass limit, e.g.\
$\MI_{011}\brk{m, m, m} = \MI_{101}\brk{m, m, m}$%
\footnote{In the momentum space representation \eqref{eq:sunrise_momentum_space}, this can be shown by shifts $k_i \to k_i + p$ in the loop momenta. These shifts have no effect on the integrated result, because the integration contour is over all of Minkowski space.}.
Let us collect the four such \emph{symmetry relations} arising in the
equal-mass limit into the matrix
\eq{
    R = \lrsbrk{
        \begin{array}{ccc ccc c}
            1 & -1 & \mzero & \mzero & \mzero & \mzero & \mzero
            \\
            1 & \mzero & -1 & \mzero & \mzero & \mzero & \mzero
            \\
            \mzero & \mzero & \mzero & \mzero & 1 & -1 & \mzero
            \\
            \mzero & \mzero & \mzero & \mzero & 1 & \mzero & -1
        \end{array}
    } \, .
    \label{eq:symmetry-matrix}
}
In analogy with the residue matrix relations from
\eqref{eqn:restriction_equation}, we have
\eq{
    R \cdot \MI(m,m,m) = 0 \, .
    \label{eq:symmetry-sunrise}
}
Now the holomorphic restriction protocol can proceed as before.
Namely, taking the equal-mass limit\footnote{
    In practice, we take $P_{z, 0} = \bigbrk{P_1 + P_2 + P_3}\big|_{z_i \to
    z}\>$, which is obtained from the unequal mass
    system~\eqref{eq:pfaff-sunrise} by shifting the variables $\brc{z_1, z_2,
    z_3} \mapsto \brc{z, z + x_1, z + x_2}$ and then taking the smooth limit of
    $x_1, x_2 \to 0$.
} of the initial Pfaffian system, we compute
\eq{
    \bigbrk{
        \pd{z} M
        + M \cdot P_{z, 0}
    } \cdot M^{-1}
    =
    \lrsbrk{
        \begin{array}{cc}
            {Q}_{z} & \star
            \\
            \mZero\subsm{4 \times 3} & \star
        \end{array}
    }
    \> , \quad \text{with} \quad
    M = \lrsbrk{
        \begin{array}{c}
            B
            \\
            \chline
            R
        \end{array}
    } 
}
and $B$ given in \eqref{eq:basis-sunrise}.
This gives
\eq{
    Q_z = \lrsbrk{
        \begin{array}{ccc}
            \frac{-2 \ep}{z} & \mzero & \mzero
            \\
            \mzero & \mzero & 3
            \\
            \frac{2 \ep^2}{(z-1) z^2 (9z-1)} & -\frac{(1 + 2 \ep) (1 + 3 \ep) (3z-1)}{(z-1) z (9z-1)} & -\frac{1 + \ep - 20 z - 30 \ep z + 27 z^2 + 45 \ep z^2}{(z-1) z (9z-1)}
        \end{array}
    } \, ,
    \label{eq:deq-sunrise-answer}
}
in agreement with an independent \soft{LiteRed} calculation.

\subsubsection{Deriving symmetry relations}
\label{sssec:derive-symmetry-relations}
In this previous calculation, we merely \emph{guessed} the symmetry relation matrix \eqref{eq:symmetry-matrix}.
In the following we show how it can be derived.

Let us recall the secondary-like equation \eqref{eqn:secondary_like_equation}, which in the
present context reads
\eq{
    \pd{z} \Phi + P_{z, 0}\tr \cdot \Phi - \Phi \cdot P_{z, 0}\tr = 0
    \> ,
}
for some unknown matrix $\Phi$.
The \soft{IntegrableConnections} \cite{integrableconnection} library in
\soft{Maple} can find rational solutions to such PDEs via the function
\code{RationalSolutions} (see also the function \code{EigenRing}).
We find five non-trivial solutions
\eq{
    &\brc{\Phi_1, \ldots, \Phi_5}
    =
    \\[3pt]\nn
    &\lrbrc{
        \lrsbrk{\tiny
            \begin{array}{ccccccc}
            \mzero & -1 & 1 & \mzero & \mzero & \mzero & \mzero \\
            \mzero & \mzero & \mzero & \mzero & \mzero & \mzero & \mzero \\
            \mzero & 1 & -1 & \mzero & \mzero & \mzero & \mzero \\
            \mzero & \mzero & \mzero & \mzero & \mzero & \mzero & \mzero \\
            \mzero & \mzero & \mzero & \mzero & \mzero & -1 & 1 \\
            \mzero & \mzero & \mzero & \mzero & \mzero & \mzero & \mzero \\
            \mzero & \mzero & \mzero & \mzero & \mzero & 1 & -1 \\
            \end{array}
        }
        ,
        \lrsbrk{\tiny
            \begin{array}{ccccccc}
            \mzero & 1 & \mzero & \mzero & \mzero & \mzero & \mzero \\
            1 & \mzero & \mzero & \mzero & \mzero & \mzero & \mzero \\
            \mzero & \mzero & 1 & \mzero & \mzero & \mzero & \mzero \\
            \mzero & \mzero & \mzero & 1 & \mzero & \mzero & \mzero \\
            \mzero & \mzero & \mzero & \mzero & \mzero & 1 & \mzero \\
            \mzero & \mzero & \mzero & \mzero & 1 & \mzero & \mzero \\
            \mzero & \mzero & \mzero & \mzero & \mzero & \mzero & 1 \\
            \end{array}
        }
        ,
        \lrsbrk{\tiny
            \begin{array}{ccccccc}
            \mzero & 1 & \mzero & \mzero & \mzero & \mzero & \mzero \\
            \mzero & 1 & \mzero & \mzero & \mzero & \mzero & \mzero \\
            1 & -1 & 1 & \mzero & \mzero & \mzero & \mzero \\
            \mzero & \mzero & \mzero & 1 & \mzero & \mzero & \mzero \\
            \mzero & \mzero & \mzero & \mzero & \mzero & 1 & \mzero \\
            \mzero & \mzero & \mzero & \mzero & \mzero & 1 & \mzero \\
            \mzero & \mzero & \mzero & \mzero & 1 & -1 & 1 \\
            \end{array}
        }
        ,
        \lrsbrk{\tiny
            \begin{array}{ccccccc}
            \mzero & 1 & \mzero & \mzero & \mzero & \mzero & \mzero \\
            \mzero & \mzero & 1 & \mzero & \mzero & \mzero & \mzero \\
            1 & \mzero & \mzero & \mzero & \mzero & \mzero & \mzero \\
            \mzero & \mzero & \mzero & 1 & \mzero & \mzero & \mzero \\
            \mzero & \mzero & \mzero & \mzero & \mzero & 1 & \mzero \\
            \mzero & \mzero & \mzero & \mzero & \mzero & \mzero & 1 \\
            \mzero & \mzero & \mzero & \mzero & 1 & \mzero & \mzero \\
            \end{array}
        }
        ,
        \lrsbrk{\tiny
            \begin{array}{ccccccc}
            1 & -1 & \mzero & \mzero & \mzero & \mzero & \mzero \\
            \mzero & \mzero & \mzero & \mzero & \mzero & \mzero & \mzero \\
            -1 & 1 & \mzero & \mzero & \mzero & \mzero & \mzero \\
            \mzero & \mzero & \mzero & \mzero & \mzero & \mzero & \mzero \\
            \mzero & \mzero & \mzero & \mzero & 1 & -1 & \mzero \\
            \mzero & \mzero & \mzero & \mzero & \mzero & \mzero & \mzero \\
            \mzero & \mzero & \mzero & \mzero & -1 & 1 & \mzero \\
            \end{array}
        }
    } \, .
}
We are interested in eigenspaces of these matrices, so let us first record
their eigenvalues:
\eq{
    \nonumber
    \spec{\Phi_1} &= \brc{-1, -1, 0, 0, 0, 0, 0}
    \> ,
    \\
    \nonumber
    \spec{\Phi_2} &= \brc{-1, -1, 1, 1, 1, 1, 1}
    \> ,
    \\
    \spec{\Phi_3} &= \brc{0, 0, 1, 1, 1, 1, 1}
    \> ,
    \\
    \nonumber
    \spec{\Phi_4} &= \brc{e^{2 / 3 \pi \ii}, e^{2 / 3 \pi \ii}, e^{4 /
    3 \pi \ii}, e^{4 / 3 \pi \ii}, 1, 1, 1}
    \> ,
    \\
    \nonumber
    \spec{\Phi_5} &= \brc{0, 0, 0, 0, 0, 1, 1}
    \> .
}
Next we collect the eigenvectors of $\Phi_i$ corresponding to an eigenvalue $\lambda$
into the rows of a matrix $V_i\supbrk{\lambda}\>$, meaning that it satisfies
\eq{
    \lrsbrk{\Phi_i - \lambda \mId}
    \cdot
    \lrbrk{V_i\supbrk{\lambda}}\tr
    = 0
    \quad , \quad
    i = 1, \ldots, 5 \, .
}
Now comes our main observation: the rows of the symmetry relations
matrix~\eqref{eq:symmetry-matrix} span
the same linear space as the rows of \brk{some of the} $V_i\supbrk{\lambda}$
matrices! We have
\eq{
    \rowReduce{R}
    =
    \rowReduce{\lrsbrk{
        \def\arraystretch{1.2}
        \begin{array}{c}
            V_1\supbrk{-1}
            \\
            V_2\supbrk{-1}
            \\
            V_3\supbrk{0}
            \\
            V_5\supbrk{1}
        \end{array}
    }}\>,
    \label{eq:symmetry-eigenspaces}
}
where, as before, the $\code{RowReduce}$ operation includes deletion of zero-rows.
The correspondence~\eqref{eq:symmetry-eigenspaces} is the reason why our choice
of the symmetry relations matrix~\eqref{eq:symmetry-matrix} actually gave the
correct restricted Pfaffian~\eqref{eq:deq-sunrise-answer}: the rows of this $R$
span an invariant subspace $\cN$ of the
$\cR=\C(z)\langle\partial_z\rangle$-module $\cM$ associated to the Pfaffian
matrix $P_{z,0}$ (see the discussion around~\thmref{thm:Schur}).
Equation \eqref{eq:deq-sunrise-answer} is thus the Pfaffian matrix for the
quotient $\cR$-module $\cM/\cN$ with respect to a basis $B$, and can be
determined by the restriction method as shown above.

\subsection{Restriction of GKZ systems by the Macaulay matrix method}

Next we present examples of Feynman integrals arising as restrictions of GKZ systems.
We employ the Macaulay matrix based algorithms from \secref{sec:NTsec1_MMMR}.
For simplicity, in all examples we omit $\Gamma$-prefactors coming from the Lee-Pomeransky representation.
To obtain Pfaffian matrices for the proper Feynman integrals, one would need to include these prefactors and then send $\de \to 0$.

Given $n$ integration variables in the Euler integral \eqref{f_Gamma(z)}, for every example below we fix the GKZ $\beta$-parameters to be the $(n+1)$-dimensional vector $\beta = \sbrk{\ep, -\ep \de, \ldots, - \ep \de}\tr$.
This follows the convention of \cite{Chestnov:2022alh}.

\subsubsection{One-loop triangle}\label{sec:gkz_1L_1M_triangle}

We first consider the one-loop triangle diagram with one internal mass $m$.
After rescaling the integration variables by $m^2$, the Lee-Pomeransky polynomial associated to this diagram takes the form
\eq{
    g(x;z) &= z_1 x_1 + z_2 x_2 + z_3 x_3 + z_4 x_1^2+z_5 x_1 x_2 + z_6 x_1 x_3 \\
    z_1 &= \ldots = z_5 = 1
    \quad , \quad
    z_6 = \frac{(p_1+p_2)^2}{m^2} 
}
where $p_1,p_2$ denote the incoming momenta of two massless particles.

We are interested in obtaining the Pfaffian system associated to the proper Feynman integral as a restriction of the GKZ system:
\begin{equation}
    A=\lrsbrk{
        \begin{array}{cccccc}
            1&  1&  1&  1&  1&  1 \\
            1& \mzero& \mzero&  2&  1&  1 \\
            \mzero&  1& \mzero& \mzero&  1& \mzero \\
            \mzero& \mzero&  1& \mzero& \mzero&  1 \\
        \end{array}
    }
    \quad
    \xrightarrow[i \, \neq \, 6]{z_i = 1}
    \quad
    \includegraphicsbox{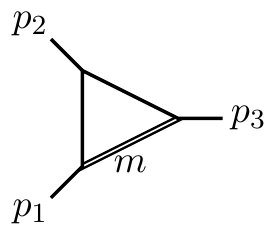}
\end{equation}
With all variables $z_1,\ldots,z_6$ being generic, the GKZ system has rank $3$.
However, an Euler characteristic computation (see \thmref{thm:Gauss_Manin}) gives $\chi = 2$ for the case $z_1 = \ldots = z_5 = 1$ and generic $z_6$.
We thus seek to drop the rank from $3$ to $2$ via a $\cD_Y$-module restriction.

Following \cite[Appendix A]{Chestnov:2022alh}, we are free to immediately rescale $z_1 = \ldots z_4 = 1$ without changing the rank of GKZ system.
The remaining variables are thus $z_5$ and $z_6$, and we seek the restriction $z_5 = 1$.
We first find a standard basis when $z_5=1$ by applying \algref{alg:rest_to_pt} with the randomly chosen value $z_6=1/17$ and the algorithm parameters $\gamma=2$, $k=3$.
The output is a basis consisting of $2$ elements:
\begin{equation}
    \RStd = \sbrk{ \pd{6},1 }\tr \, ,
\end{equation}
which agrees with the rank being $2$.
We can obtain a Pfaffian matrix in the variable $z_6$ for this restricted basis by \algref{alg:alg3}.
The Macaulay matrix succeeds at degree $D=1$, and we get
$$P_6=\lrsbrk{
\begin{array}{cc}
 \frac{(     (     4  {\epsilon}  {\delta}+ {\epsilon}+ 1)  {z}_{6}-   4  {\epsilon}  {\delta}-  2  {\epsilon}- 1)} {   {z}_{6}  (1-{z}_{6})}&  \frac{ {\epsilon}^{ 2}    {\delta}  (   3  {\delta}+ 1)} {   {z}_{6}  (1-{z}_{6})} \\
 1& 0 \\
\end{array}
}.
$$

\subsubsection{Two-loop N-box}\label{sec:gkz_2L_1M_nbox}

The two-loop N-box with one internal mass has the following Lee-Pomeransky polynomial:
\eq{
    \nonumber
    g(z;x) = &
    z_1 x_1 x_2 +
    z_2 x_1 x_4 +
    z_3 x_1 x_5 +
    z_4 x_2 x_3 +
    z_5 x_2 x_4 +
    z_6 x_2 x_5 + \\ &
    z_7 x_3 x_4 +
    z_8 x_3 x_5 +
    z_9 x_1 x_2 x_4 +
    z_{10} x_1 x_4^2  +
    z_{11} x_1 x_4 x_5 +
    z_{12} x_2 x_3 x_4 + \\ \nonumber &
    z_{13} x_2 x_3 x_5 +
    z_{14} x_2 x_4^2  +
    z_{15} x_2 x_4 x_5 +
    z_{16} x_3 x_4^2  +
    z_{17} x_3 x_4 x_5  \, .
}
We study it as a restriction of a GKZ system:
\begin{equation}
    A=\lrsbrk{
     \begin{array}{ccccccccccccccccc}
       1&  1&  1&  1&  1&  1&  1&  1&  1&  1&  1&  1&  1&  1&  1&  1&  1 \\
       1&  1&  1& \mzero& \mzero& \mzero& \mzero& \mzero&  1&  1&  1& \mzero& \mzero& \mzero& \mzero& \mzero& \mzero \\
       1& \mzero& \mzero&  1&  1&  1& \mzero& \mzero&  1& \mzero& \mzero&  1&  1&  1&  1& \mzero& \mzero \\
       \mzero& \mzero& \mzero&  1& \mzero& \mzero&  1&  1& \mzero& \mzero& \mzero&  1&  1& \mzero& \mzero&  1&  1 \\
       \mzero&  1& \mzero& \mzero&  1& \mzero&  1& \mzero&  1&  2&  1&  1& \mzero&  2&  1&  2&  1 \\
       \mzero& \mzero&  1& \mzero& \mzero&  1& \mzero&  1& \mzero& \mzero&  1& \mzero&  1& \mzero&  1& \mzero &  1 \\
    \end{array}
    }
    \quad
    \xrightarrow[i \, \neq \, 9, \, 13]{z_i \to 1}
    \quad
    \includegraphicsbox{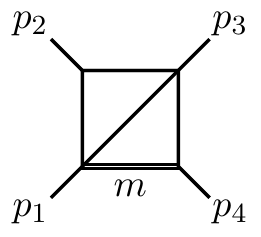}
\end{equation}

\vspace{5pt}
The 17-variable GKZ system has rank $33$ for generic $\ep, \de$.
To match with the proper Feynman integral, we seek the restriction
\eq{
    z_i = 1
    \, , \,
    i \neq 9, 13
}
with $z_9 = 1+\frac{s}{m^2}$ and $z_{13} = \frac{t}{m^2}$ being the remaining kinematic variables.
By computing the Euler characteristic with restricted variables, we expect a rank drop from $33$ to $7$.

To begin with, we choose the simplex $\sigma = \lrsbrk{1,2,3,4,5,10}$.
The remaining variables to restrict are thus $z_i = 1$ for $i \in \{6, 7, 8, 11, 12, 14, 15, 16, 17\}$, for which we apply \algref{alg:rest_to_pt}.
Taking parameters $\gamma = 3, k = 4$ and randomly chosen values $z_9 = 1/17, z_{13} = 1/19$, we are able to guess the basis
\begin{equation}
    \RStd = [\pd{8},\pd{15},\pd{16},\pd{17},\pd{9},\pd{13},1 ]\tr
\end{equation}
by running the algorithm\footnote{It takes 7.32s on an m-PC(Intel(R) Core(TM) i7-10700K CPU @ 3.80GHz, 8G memory)}
over the finite field of characteristic $p=100000007$.
Note that the basis is $7$-dimensional, as expected.

Next we use \algref{alg:alg3} to obtain a Macaulay matrix, which succeeds\footnote{It takes 3.28s on the m-PC.} at degree $D=1$.
Since the Macaulay matrix is valid, it means that the probabilistic basis from \algref{alg:rest_to_pt} is indeed a basis.
Using \soft{FiniteFlow} \cite{Peraro:2019svx},
in a matter of seconds we obtain a solution to equation \eqref{eq:MM_equation_for_restriction} for $i=9, 13$ via rational reconstruction over finite fields.
The result is a set of $7 \times 7$ Pfaffian matrices $P_9$ and $P_{13}$.
We have checked that they satisfy integrability.

\subsubsection{Two-loop diagonal-box}\label{sec:gkz_2L_0m_diagonal_box}
Here we consider the massless two-loop diagonal-box diagram, whose Lee-Pomeransky polynomial takes the form
\eq{
    \nonumber
    g(z;x) = &
    z_{1} x_1 x_4 +
    z_{2} x_1 x_5 +
    z_{3} x_1 x_6 +
    z_{4} x_2 x_4 +
    z_{5} x_2 x_5 +
    z_{6} x_2 x_6 + \\ \nonumber &
    z_{7} x_3 x_4 +
    z_{8} x_3 x_5 +
    z_{9} x_3 x_6 +
    z_{10} x_4 x_5 +
    z_{11} x_4 x_6 +
    z_{12} x_4 x_7 + \\ &
    z_{13} x_5 x_7 +
    z_{14} x_6 x_7 +
    z_{15} x_1 x_3 x_4 +
    z_{16} x_1 x_3 x_5 +
    z_{17} x_1 x_3 x_6 + \\ \nonumber &
    z_{18} x_1 x_4 x_5 +
    z_{19} x_2 x_4 x_6 +
    z_{20} x_2 x_4 x_7 +
    z_{21} x_2 x_5 x_7 +
    z_{22} x_2 x_6 x_7 \, .
}
We consider the following restriction of the GKZ system:
\begin{equation}
    A=\lrsbrk{
        \begin{array}{cccccccccccccccccccccc}
         1&  1&  1&  1&  1&  1&  1&  1&  1&  1&  1&  1&  1&  1&  1&  1&  1&  1&  1&  1&  1&  1 \\
         1&  1&  1& \mzero& \mzero& \mzero& \mzero& \mzero& \mzero& \mzero& \mzero& \mzero& \mzero& \mzero&  1&  1&  1&  1& \mzero& \mzero& \mzero& \mzero \\
         \mzero& \mzero& \mzero&  1&  1&  1& \mzero& \mzero& \mzero& \mzero& \mzero& \mzero& \mzero& \mzero& \mzero& \mzero& \mzero& \mzero&  1&  1&  1&  1 \\
         \mzero& \mzero& \mzero& \mzero& \mzero& \mzero&  1&  1&  1& \mzero& \mzero& \mzero& \mzero& \mzero&  1&  1&  1& \mzero& \mzero& \mzero& \mzero& \mzero \\
         1& \mzero& \mzero&  1& \mzero& \mzero&  1& \mzero& \mzero&  1&  1&  1& \mzero& \mzero&  1& \mzero& \mzero&  1&  1&  1& \mzero& \mzero \\
         \mzero&  1& \mzero& \mzero&  1& \mzero& \mzero&  1& \mzero&  1& \mzero& \mzero&  1& \mzero& \mzero&  1& \mzero&  1& \mzero& \mzero&  1& \mzero \\
         \mzero& \mzero&  1& \mzero& \mzero&  1& \mzero& \mzero&  1& \mzero&  1& \mzero& \mzero&  1& \mzero& \mzero&  1& \mzero&  1& \mzero& \mzero&  1 \\
\mzero& \mzero& \mzero& \mzero& \mzero& \mzero& \mzero& \mzero& \mzero& \mzero& \mzero&  1&  1&  1& \mzero& \mzero& \mzero& \mzero& \mzero&  1&  1&  1 \\
\end{array}
    }
    \quad
    \xrightarrow[i \, \leq \, 18, \> 19 \, \leq \, j \, \leq \, 22]{z_i = 1, \> z_j = t / s}
    \quad
    \includegraphicsbox{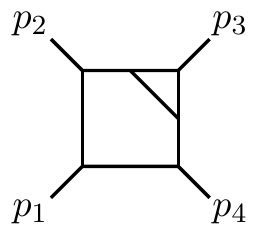}
\end{equation}
The GKZ system has rank $115$ for generic $\ep, \de$.
Upon restricting to the proper Feynman integral case,
\eq{
    z_1 = \ldots = z_{18} = 1
    \quad , \quad
    z_{19} = \ldots = z_{22} = z := \frac{t}{s}
}
we except a rank drop to $7$, according to the Euler characteristic.

We follow an analogous strategy as in the previous section.
We guess {\tt RStd}
by making the change of coordinates
$z_i=y_i+1$ for $i=1, \ldots, 18$,
$z_i=y_i+z$ for $i=19, 20, 21$,
$z_{22}=z$
and taking the restriction to $y_i=0$ ($i=1, \ldots, 21$).
For $\sigma = \sbrk{1,2,5,7,9,11,13,18}, \, \gamma = 3, \, k = 4, p=100000007$ and the randomly chosen value $z = 1/7$, we find a $7$-dimensional basis%
\footnote{It takes 40.56s on an m-PC(Intel(R) Core(TM) i7-10700K CPU @ 3.80GHz, 8G memory)}
\eq{
    \RStd =
    \sbrk{ \pd{21}\pd{z},\pd{z}^2,\pd{6},\pd{17},\pd{21},\pd{z},1 }\tr \, ,
}
and a valid corresponding Macaulay matrix at degree $D=1$%
\footnote{It takes 65.76s on the m-PC.}.
Note that $\pd{i}$ means $\partial/\partial y_i$.
With this Macaulay matrix, the $7 \times 7$ Pfaffian $P_z$ is solved for in less than $3$ minutes on a laptop via \soft{FiniteFlow}.

\subsubsection{Two-loop double-box}\label{sec:gkz_0m_dbox}

Our last Feynman integral example is the massless two-loop double-box.
We start from the $26$-term Lee-Pomeransky polynomial
\eq{
    \nonumber
    g(z;x) &=
    z_1 x_1 x_2 +
    z_2 x_1 x_3 +
    z_3 x_1 x_6 +
    z_4 x_1 x_7 +
    z_5 x_2 x_3 +
    z_6 x_2 x_4 +
    z_7 x_2 x_5 + \\&
    z_8 x_3 x_4 +
    z_9 x_3 x_5 +
    z_{10} x_3 x_6 +
    z_{11} x_3 x_7 +
    z_{12} x_4 x_6 +
    z_{13} x_4 x_7 +
    z_{14} x_5 x_6 + \\&\nonumber
    z_{15} x_5 x_7 +
    z_{16} x_1 x_2 x_5 +
    z_{17} x_1 x_2 x_6 +
    z_{18} x_1 x_3 x_5 +
    z_{19} x_1 x_3 x_6 +
    z_{20} x_1 x_5 x_6 + \\&\nonumber
    z_{21} x_1 x_5 x_7 +
    z_{22} x_2 x_3 x_5 +
    z_{23} x_2 x_3 x_6 +
    z_{24} x_2 x_4 x_6 +
    z_{25} x_2 x_5 x_6 +
    z_{26} x_3 x_4 x_7 \, ,
}
and compute the restriction
\begin{equation}
A=
    \lrsbrk{
        \begin{array}{cccccccccccccccccccccccccc}
            1&  1&  1&  1&  1&  1&  1&  1&  1&  1&  1&  1&  1&  1&  1&  1&  1&  1&  1&  1&  1&  1&  1&  1&  1&  1 \\
            1&  1&  1&  1& \mzero& \mzero& \mzero& \mzero& \mzero& \mzero& \mzero& \mzero& \mzero& \mzero& \mzero&  1&  1&  1&  1&  1&  1& \mzero& \mzero& \mzero& \mzero& \mzero \\
            1& \mzero& \mzero& \mzero&  1&  1&  1& \mzero& \mzero& \mzero& \mzero& \mzero& \mzero& \mzero& \mzero&  1&  1& \mzero& \mzero& \mzero& \mzero&  1&  1&  1&  1& \mzero \\
            \mzero&  1& \mzero& \mzero&  1& \mzero& \mzero&  1&  1&  1&  1& \mzero& \mzero& \mzero& \mzero& \mzero& \mzero&  1&  1& \mzero& \mzero&  1&  1& \mzero& \mzero&  1 \\
            \mzero& \mzero& \mzero& \mzero& \mzero&  1& \mzero&  1& \mzero& \mzero& \mzero&  1&  1& \mzero& \mzero& \mzero& \mzero& \mzero& \mzero& \mzero& \mzero& \mzero& \mzero&  1& \mzero&  1 \\
            \mzero& \mzero& \mzero& \mzero& \mzero& \mzero&  1& \mzero&  1& \mzero& \mzero& \mzero& \mzero&  1&  1&  1& \mzero&  1& \mzero&  1&  1&  1& \mzero& \mzero&  1& \mzero \\
            \mzero& \mzero&  1& \mzero& \mzero& \mzero& \mzero& \mzero& \mzero&  1& \mzero&  1& \mzero&  1& \mzero& \mzero&  1& \mzero&  1&  1& \mzero& \mzero&  1&  1&  1& \mzero \\
            \mzero& \mzero& \mzero&  1& \mzero& \mzero& \mzero& \mzero& \mzero& \mzero&  1& \mzero&  1& \mzero&  1& \mzero& \mzero& \mzero& \mzero& \mzero&  1& \mzero& \mzero& \mzero& \mzero&  1 \\
        \end{array}
    }
    \quad
    \xrightarrow[i \, \leq \, 25]{z_i = 1}
    \quad
    \includegraphicsbox{figures/dbox}
\end{equation}
For generic $\ep$ and $\de$, the rank of the GKZ system is a whopping $238$.
While the GKZ system has generic $z_i$, the proper massless double-box Feynman integral has monomial coefficients
\eq{
    z_i = 1 \text{ for all } i=1,\ldots,25
    \quad , \quad
    z_{26} =-1-s/t =: z \, .
}
The Euler characteristic suggests a rank $12$ system for the restricted variables above.
However, this Feynman integral is well-known to satisfy a rank $8$ Pfaffian system - we comment on this below.

We fix the simplex $\sigma = \sbrk{1,2,3,4,5,6,7,25}$ and then run \algref{alg:rest_to_pt} with parameters $\gamma = 3, \, k = 5, \, p=100000007$ and the randomly chosen value $z = 1/7$.
The result is a $12$-dimensional basis%
\footnote{It takes 23,815s CPU time on a t-PC(AMD EPYC 7552 four 48-Core Processors @ 1.5GHz, 1T memory)}
\begin{equation}
    \RStd =
    \sbrk{ \partial_{{z}}  \partial_{{15}},  \partial_{{23}}  \partial_{{z}},  \partial_{{24}}  \partial_{{z}},  \partial_{{z}}^{ 2} , \partial_{{13}}, \partial_{{15}}, \partial_{{21}}, \partial_{{22}}, \partial_{{23}}, \partial_{{24}}, \partial_{{z}}, 1}\tr \, ,
\end{equation}
whose size agrees with the Euler characteristic.

In step 4 of \algref{alg:rest_to_pt}, we need to compute the echelon form of a matrix of size $132,145 \times 33,649$.
\algref{alg:alg3} succeeds at finding a Macaulay matrix\footnote{It takes 501s CPU time on the t-PC.} at degree $D=1$.
The size of $M_{\Ext}$ is $2926 \times 10775$.
Running \soft{FiniteFlow}%
\footnote{It took 89,021s CPU-hours on the t-PC. However, \soft{FiniteFlow} ran in parallel over 96 threads, so the real time was about 20 min.}
to solve the matrix equation (\ref{eq:MM_equation_for_restriction}) associated to $M_{\Ext}$, we finally obtain a $12 \times 12$ Pfaffian matrix $P_z$,
which can be downloaded from \cite{dataAndProgramsOfThisPaper}.

We suspect that symmetry relations (namely, reparametrizations of the Euler
\emph{integrand}) would bridge the gap between $\chi = 12$ and the expected
rank $8$ for the Pfaffian system associated to this Feynman integral.
By using the secondary-like equation~\eqref{eqn:secondary_like_equation}, it should be
possible to block-diagonalize our $12$-dimensional ODE dictated by $P_z$ to
reach the minimal size of $8$.
Indeed, similar to~\secref{subsec: Unequal to equal mass sunrise}, with
the help of \soft{IntegrableConnections} \cite{integrableconnection} library we
found a $4$-dimensional invariant subspace of symmetry relations that reduce
the rank from $12$ down to $8$, as expected for this family of Feynman
integrals\footnote{
    It took about 70 CPU-hours and 4Gb of RAM on the t-PC for
    \soft{IntegrableConnections} to find these rational solutions.
}.
It is a future problem for us to give a more efficient implementation combining the secondary-like equation with the methods of \cite{Barkatou-1993, singer-1996}.

\subsection{Restriction of Appell's \texorpdfstring{$F_4$}{} and Horn's \texorpdfstring{$H_i$}{} functions to hypersurfaces}
\label{subsec:ex:F4}

The Appell $F_4$ function has the following series representation:
\begin{equation}
F_4(a,b,c_1,c_2;x,y)=\sum_{i,j=0}^\infty
  \frac{(a)_{i+j} (b)_{i+j}} {(1)_i (1)_j (c_1)_i (c_2)_j} x^i y^j.
\end{equation}
It appears in several contexts related to Feynman integrals \cite{Anastasiou:1999ui,FCCZ-2020}.
$F_4$ is annihilated by the following operators:
\begin{eqnarray}
\label{eq:apell_f1_f2}
f_1&=&\theta_x (\theta_x + c_1-1) - x (\theta_x+\theta_y+a)(\theta_x+\theta_y+b)\\
f_2&=&\theta_y (\theta_y + c_2-1) - y (\theta_x+\theta_y+a)(\theta_x+\theta_y+b) \, ,
\end{eqnarray}
where the Euler operators are $\theta_x = x \partial_x$ and $\theta_y = y \partial_y$.
Let us assume that the parameters $a, b, c_1, c_2$ take on generic values.
The operators $f_1$ and $f_2$ generate a holonomic ideal $\cI$
in the $\cD$-module $\CD_2$ whose coordinates are $z_1=x$, $z_2=y$.
The left $\CD_2$-module $\CD_2/\cI$ is regular holonomic with
holonomic rank $4$, and the singular locus is
\begin{equation}
  x y ((x-y)^2-2(x+y)+1) = 0 \, .
\end{equation}
Put $L=(x-y)^2-2(x+y)+1$.
The point $(x,y)=(1/4,1/4)$ is on the singular locus,
and we will construct a Pfaffian system restricted
to a codimension-$1$ stratum $W=V(L) \setminus \{(1,0),(0,1)\}$ containing that point.
Note that $V(L)$ and $xy=0$ intersect at $(1,0)$ and $(0,1)$.

Computing the ODE on $x=y$ by restricting the Pfaffian system,
we see that the set of exponents at the point $V(x-y) \cap W$ is
\begin{equation} \label{eq:exponents_along_VL}
    \{a, \, b, \, c_1, \, c_2\}
    \ \to \
    \{0, \, 0, \, 0, \, 2(a+b-c_1-c_2)+5\} \, .
\end{equation}
Moreover, the residue matrix at this point is diagonal,
meaning that there is no logarithmic solution.
By an analytic coordinate transformation, we can locally transform $W$
and $V(x-y)$
into hyperplanes that intersect transversally.
Moreover, it can be expressed as a logarithmic connection
with the residue matrix along $W$ being the diagonal matrix
with eigenvalues (\ref{eq:exponents_along_VL}).
It follows from \propref{prop:logarithmic_connection_as_D_module} that
the holonomic rank $r$ on $W$ is $3$
for generic parameter values.
Applying the restriction algorithm from \thmref{th:th2} with the shift
$(x,y) \to (x+1/4,y+1/4)$,
we obtain the standard basis $S = \RStd = [1,\pd{x},\pd{y}]\tr$
(see \remref{rem:choice_of_Std}).

Let us write a modifed version of \eqref{eq:restriction-of-MMM} modulo $L$:
\eq{
    \label{eq:restriction-of-MMM-frac}
    M_D \defas \lrsbrk{
        \begin{array}{c|c}
            M_{\Ext} & M_{\RStd}
        \end{array}
    } \text{ mod $L$}
    \>.
}
The left block matrix above has columns labeled by
$
    \Ext =
    \sbrk{ \pd{x}^3, \pd{x}^2 \pd{y}, \pd{x}\pd{y}^2,\pd{y}^3, \pd{x}^2, \pd{x}\pd{y}, \pd{y}^2 }
$.
We proceed to solve \eqref{eq:4.12m} for the unknown matrix $C$, which will produce the Pfaffian matrices by \eqref{eq:Pi_qi}.
Equation \eqref{eq:4.12m} requires the following two matrices as input:
\eq{
    \nonumber
    M_{\Ext} &=
    {\tiny
        \lrsbrk{
            \begin{array}{ccccccc}
                \mzero& \mzero& \mzero& \mzero&   -  {x}^{ 2} + {x}&  -   2  {y}  {x}&  -  {y}^{ 2}  \\
                \mzero& \mzero& \mzero& \mzero&  -  {x}^{ 2} &  -   2  {y}  {x}&   -  {y}^{ 2} + {y} \\
                \mzero&   -  {x}^{ 2} + {x}&  -   2  {y}  {x}&  -  {y}^{ 2} & \mzero&    (  - {a}  - {b}- 3)  {x}+ {c}_{1}&   (  - {a}  - {b}- 3)  {y} \\
                \mzero&  -  {x}^{ 2} &  -   2  {y}  {x}&   -  {y}^{ 2} + {y}& \mzero&   (  - {a}  - {b}- 3)  {x}&    (  - {a}  - {b}- 3)  {y}+  {c}_{2}+ 1 \\
                  -  {x}^{ 2} + {x}&  -   2  {y}  {x}&  -  {y}^{ 2} & \mzero&    (  - {a}  - {b}- 3)  {x}+  {c}_{1}+ 1&   (  - {a}  - {b}- 3)  {y}& \mzero \\
                 -  {x}^{ 2} &  -   2  {y}  {x}&   -  {y}^{ 2} + {y}& \mzero&   (  - {a}  - {b}- 3)  {x}&    (  - {a}  - {b}- 3)  {y}+ {c}_{2}& \mzero.
            \end{array}
        }
    } \\
    C_{\Ext} &=
    \lrsbrk{
        \begin{array}{ccccccc}
            \mzero& \mzero& \mzero& \mzero& \mzero& \mzero& \mzero \\
            \mzero& \mzero& \mzero& \mzero&  1& \mzero& \mzero \\
            \mzero& \mzero& \mzero& \mzero& \mzero&  1& \mzero \\
            \mzero& \mzero& \mzero& \mzero& \mzero& \mzero&  1 \\
        \end{array}
    }.
}
Although we explained a general method to find $C$ by a syzygy computation in \secref{sec:MMMR}, we employ a more efficient method for this case.
As \eqref{eq:4.12m} is written modulo $L=\underline{x^2} + \ldots$, we can eliminate all instances of $x^p$ for $p \geq 2$ via division by $L$.
Consequently, it is sufficient to write $C$ as a rational ansatz being linear in $x$:
\eq{
    C =
    \frac
    {p_{ij}(a,b,c_1,c_2) x + q_{ij}(a,b,c_1,c_2,y)}
    {r(a,b,c_1,c_2) x + s(a,b,c_1,c_2,y)} \, .
}
After this preprocessing, we solve the system of linear equations
\eq{
    \fun{Numerator}{C_{\Ext} - C \cdot M_{\Ext}} = 0
}
for $p_{ij}, q_{ij}, r$ and $s$, which, by \eqref{eq:Pi_qi}, yields the Pfaffian matrices $P_x$ and $P_y$ in the fraction field $\CC[x,y]$ (the matrices can be downloaded at \cite{dataAndProgramsOfThisPaper}).
As per \cite{oaku-1996}, $P_x$ and $P_y$ induce a ${\cal D}$-module on an affine variety.

Let $[g_1, \, g_2, \, g_3]$ denote the linearly independent holomorphic solutions at $(x,y) = (1/4, \, 1/4) \in W$.
The standard monomials on the singular locus $W$
are given by
$[s_1,\, s_2, \,s_3] = [1, \, \pd{x}, \, \pd{y}]$.
We then have a $3 \times 3$ matrix $F_{ij} = (s_i \bullet g_j)$.
On $W$, it holds that $(\pd{x} \bullet F) \cdot F^{-1} = P_x$ and $(\pd{y} \bullet F) \cdot F^{-1} = P_y$.
The matrices $P_x$ and $P_y$ are Pfaffian matrices on $W$.
Note that these equalities hold on $W$ but not outside of $V(L)$,
because the PDE system \eqref{eq:apell_f1_f2} is irreducible \cite{hattori-takayama}.

Let us end our discussion on $F_4$  with a numerical analysis of the restricted Pfaffian system.
To this end, we parameterize $L=0$ by
\eq{
    x = q_x(t),
    \,
    y = q_y(t)
    \quad \text{with} \quad
    q_x(t) = t,
    \,
    q_y(t) = t + 1 \pm 2 \sqrt{t} \, .
}
Fixing the exponent parameters $(a, \, b, \, c_1, \, c_2) = (-2/3, \, 1/3, \, 1/3, \, 1/3)$, we thus obtain an ODE of rank $3$ on $V(L)$:
\eq{
    \frac{\dd}{\dd t} g(t) =
    \lrsbrk{
        P_x(q_x(t), q_y(t))
        \frac{\dd q_x}{\dd t} +
        P_y(q_x(t), q_y(t))
        \frac{\dd q_y}{\dd t}
    }
    \cdot g(t) \, .
}
A solution for the $3$-dimensional vector $g(t)$ with initial condition $g = [1, \, 0, \, 0]\tr$
is given in a \soft{python} notebook at \cite{dataAndProgramsOfThisPaper}.

We note that the restriction of $F_4$ to its singular locus was also studied in \cite{mitsuo1995connection} from a different point of view, utilizing a double cover map.
However, it is not known that this method can be applied to Horn's functions $H_1, \ldots, H_7$ (see \cite{erdelyi} for definitions).
With our general algorithm, we are able to obtain rational restrictions of every system for $H_i$ onto irreducible components of their singular loci.
The relevant hypersurfaces are given in the following table:
\begin{center}
\begin{tabular}{c|c}
      & Hypersurface  \\ \hline
$H_1$ &   $y^2-4 x y+2 y+1=0$ \\
$H_2$ &   $xy-y-1=0$\\
$H_3$ &   $y^2-y+x=0$\\
$H_4$ &   $y^2-2y-4x+1=0$\\
$H_4$ &   $27y^2x-36xy-y+16x^2+8x+1=0$\\
$H_6$ &   $y^2x-y-1=0$\\
$H_7$ &   $4xy^2-y^2-2y-1=0$\\
\end{tabular}
\end{center}
The holonomic $\cD$-modules for each $H_i$ are of rank $4$, and in all cases the restricted standard basis is ${\RStd} = [1,\pd{x},\pd{y}]\tr$.
In other words, the rational restrictions of
each system is of rank $3$.
See \cite{dataAndProgramsOfThisPaper} for programs for these restrictions.

\section{Conclusion}\label{sec:conclusion}

In this work we studied the process of taking limits on Feynman integrals and,
more generally, Euler integrals. Our philosophy has been to study such
integrals from the perspective of the differential equations that they obey; in
particular their Pfaffian PDE systems. The structure of Pfaffian systems, in
turn, is richly encoded in the algebraic language of holonomic $\cD$-modules.
It is thereby possible to formalize the notion of a limit as a
\emph{restriction} of a $\cD$-module.
In our work, the motivating example of a $\cD$-module has been the GKZ
hypergeometric system, but our results and methods apply directly
to Feynman integrals as well.

A restriction is generally accompanied by a drop in rank. Regarding the
Pfaffian system for some Euler integral, this means that we expect a smaller
and simpler Pfaffian system to survive in the limit. The main result of this
publication is the development of efficient algorithms for computing these
smaller Pfaffian systems. We presented two such restriction algorithms, each
one having unique benefits. The first algorithm, working at the level of
Pfaffian matrices, only requires sequential applications of certain gauge
transformations. This algorithm is able to extract logarithmic corrections in
the limiting variable, a feature that we expect to be especially useful for
particle physics phenomenology. The second algorithm, working at the level of
$\cD$-modules, computes the restricted Pfaffian system by passing through the
\emph{Macaulay matrix}. This algorithm benefits from the use of linear algebra
over finite fields, making it quite efficient, and it does not require the unrestricted Pfaffian system, which can be hard to obtain, as input. While traditional algorithms
have been limited to restriction onto hyperplanes, we have shown that the Macaulay matrix based
method is able to restrict onto hypersurfaces.

In future work, it would be interesting to test our methods on state-of-the-art
examples relevant for current phenomenological applications. While
this paper only attempted to compute logarithmic corrections to the one-loop
case of Bhabha scattering, there is in principle no computational barrier preventing the push to higher
loops and more points. For the sake of increased accuracy, it is then alluring to
also include the mixed power-and-log corrections of the form $z_1^a \log^b(z_1)$,
$a,b \in \ZZ_{>0}\>$, as discussed at the end of
\secref{sec:relation_to_Feynman_integrals}.
On the mathematical side, we
believe that it would be interesting to investigate restriction \emph{ideals}
using the Macaulay matrix method. This would also aid in the computation of
symbol alphabets at two loops, as they can be obtained from restrictions of GKZ
systems \cite{Dlapa:2023cvx}.

\section*{Acknowledgements}
We deeply appreciate discussions had with our mentor Prof. Pierpaolo Mastrolia
and colleagues Manoj K. Mandal and Federico Gaparotto.
This work gives answers to questions raised in the previous paper \cite{Chestnov:2022alh}, written in collaboration with them.

S.M. and N.T. are supported in part by the JST CREST Grant Number JP19209317.
S.M. is supported by JSPS KAKENHI Grant Number 22K13930.
N.T. is supported in part by JSPS KAKENHI Grant Number JP21K03270.
V.C. is supported by the Diagrammalgebra Stars Wild-Card Grant UNIPD,
and in part by the European Research Council (ERC) under the
European Union's research and innovation programme grant agreement 101040760
(ERC Starting Grant FFHiggsTop).

\appendix

\section{Modules and tensors in a nutshell}  \label{sec:nut}
\def\DDD{\CD}
\def\RRR{R}

This exposition is independent of the main text and contains a pedagogical
review of tensor products in the context of $\cD$-modules, along with several
instructive examples of tensor products also appearing in the text.

Let $R$ be a ring with $1$ and $\CM$ be an abelian group
with the group operation denoted additively by ``+''.
Consider a map from $R \times \CM$ to $\CM$, whose image $a \, m$ for given $a \in R$
and $m \in \CM$ is denoted multiplicatively%
\footnote{
    Explicitly, we consider a map $\psi \, : \, R \times \CM \to \CM$ and write
    its image as multiplication: $\psi\brk{a, m} \equiv a \, m\>$.
    We use the same notation for multiplication inside of the ring $R$.
}
and is referred to as the result of the \textit{action} of $a$ on $m$.
Then the abelian group $\CM$ forms a \textit{left $R$-module}
when it satisfies the following conditions:
\begin{enumerate}
    \item
        For any $a \in R$ and $m_i \in \CM$, we have
        $ a \, (m_1 + m_2) = a \, m_1 + a \, m_2$.
    \item
        For any $a_i \in R$ and $m \in \CM$, we have
        $ (a_1 \, a_2) \, m = a_1 \, (a_2 \, m)\>$, and
        $(a_1 + a_2) \, m = a_1 \, m + a_2 \, m\>$.%
        \footnote{
            In other words,
            $\psi(a_1 \, a_2, m) = \psi\bigbrk{a_1, \psi\brk{a_2, m}}$ and
            $\psi(a_1 + a_2, m) = \psi(a_1, m) + \psi(a_2, m)\>$.
        }
    \item
        For any $m \in \CM$, we have $1 \, m = m$.
\end{enumerate}
When $R$ is a (commutative) field, a left $R$-module is nothing but an
$R$-vector space.

A \textit{right $R$-module} is defined analogously
by the following conditions:
$(m_1 + m_2) \, a = m_1 \, a +  m_2 \, a\>$,
$m \, (a_1 \, a_2) = (m \, a_1) \, a_2\>$,
$m \, (a_1 + a_2) = m \, a_1 + m \, a_2\>$, and
$m = m \, 1$.
\begin{example}
    \label{ex:D1-right}
    \rm
    Consider the ring of univariate differential operators with polynomial
    coefficients $\DDD=\CC\langle z, \pd{z} \rangle$.
    The quotient set $\CM = \DDD / (z \, \DDD)$ forms a right $\DDD$-module
    by defining the right action $[m \, a]$ of some $a \in \DDD$
    on $[m] \in \CM$ as the product \brk{i.e. composition} of differential
    operators in $\DDD$ and a subsequent reduction to the equivalence class%
    \footnote{
        Here and in the following the square brackets $[\bigcdot]$ denote the
        equivalence class of some element $\bigcdot$ under some relation.
        Sometimes we will omit these brackets for the sake of brevity if it
        does not lead to confusion.
    }
    by $z \, \DDD\>$.

    \hfill$\blacksquare$
\end{example}

\begin{example}
    \rm
    Let $\DDD_n$ be the Weyl algebra:
    \begin{equation}
        \DDD_n = \CC\langle z_1, \ldots, z_n, \pd{1}, \ldots, \pd{n} \rangle
        \>,
    \nonumber
    \end{equation}
    that is, the ring of differential operators in $n$ variables $z_1, \ldots,
    z_n\>$.
    The set of length $k$-vectors $\DDD_n^k$ is an example of a left
    $\DDD_n$-module: The addition is defined by
    $(f_1, \ldots, f_k) + (g_1, \ldots, g_k) = (f_1+g_1, \ldots, f_k+g_k)$
    and the action of $h \in \DDD_n$ on $(f_1, \ldots, f_k)$ is given by
    $(h f_1, \ldots, h f_k)\>$.
    $\DDD_n^k$ is called the \textit{free $\DDD_n$-module} of rank $k$.
    \hfill$\blacksquare$
\end{example}

Let us now briefly review the construction of a Gr\"obner basis in the free module
$\DDD_n^k$ that we have just introduced above.
We fix a term order $\prec$ in $\DDD_n$.
For a $k$-vector $f=(0,\ldots, 0, f_i, \ldots, f_k)$ with $f_i \not= 0\>$,
the index $i$ is called the \textit{top index} of $f$.
Consider two elements $f$ and $g$ in $\DDD_n^k\>$.  We say that $f$ is larger
than $g$ in the \textit{position over term} (POT) order of $\DDD_n^k\>$, when
the top index of $f$ is smaller than the top index of $g\>$, or when the top
indices of $f$ and $g$ are both $i$ and the corresponding components
$f_i \succ g_i\>$.
The reduction \brk{or the division} algorithm follows from this definition of the
total order. Namely, the $S$-pair $S(f,g)$ of $f$ and $g$ is a $(0, \ldots, 0)$ vector
when the top indices of $f$ and $g$ are not the same.
When the top indices are the same and equal to $i$, then the $S$-pair of $f$
and $g$ is defined by $h_1 f - h_2 \, g\>$, where $h_j \in \DDD_n$  are chosen
so that $h_1 f_i - h_2 \, g_i$ is the $S$-pair of $f_i$ and $g_i$ in $\DDD_n$.
\begin{example}
    \rm
    Consider the free module $\DDD_2^2$ in two variables $z_1$ and $z_2\>$.
    When $f=(\pd{1}, z_1)$ and $g=(0, z_1)\>$,
    then $S(f,g)=(0,0)$ as the top indices of $f$ and $g$ are not the same.
    If $f = (z_1 \, \pd{1}, z_1)$ and $g = (z_2 \, \pd{1}, z_1)$,
    then the $S$-pair becomes $S(f, g) = z_2 \, f - z_1 \, g = (0,z_1 z_2-z_1^2)\>$.
    For further examples and details we refer the interested reader to
    \cite{adams} and \cite[Sections 6.6-7]{dojo}.

    \hfill$\blacksquare$
\end{example}

\begin{definition}
    \label{def:tensor}
    For a left $\DDD_n$-module $\CM$ and a right $\DDD_n$-module $P$, let us define
    the \textit{tensor product} $P \otimes_{\DDD_n} \CM$ as follows.  We denote by $p
    \otimes m$ the element of the direct product $(p,m)$, where $p \in P$ and $m
    \in \CM$. Then $P \otimes_{\DDD_n} \CM$ is the set of a finite formal sum of elements of this
    form modulo the equivalence relation $\sim$ given by:
    \begin{enumerate}
        \item
            $p_1 \otimes m + p_2 \otimes m \sim (p_1 + p_2) \otimes m$,
            $m \in \CM$, $p_i \in P$ \brk{linearity in the left slot}.
        \item
            $p \otimes m_1 + p \otimes m_2 \sim p \otimes (m_1 + m_2)$,
            $p \in P$, $m_i \in \CM$ \brk{linearity in the right slot}.
        \item
            $p f \otimes m \sim p \otimes f m$
            for $f \in \DDD_n$, $m \in \CM$, and $p \in P\>$.%
            \footnote{
                Hence the subscript $\DDD_n$ in $\otimes_{\DDD_n}$  means that any element $f \in \DDD_n$
                can be carried across the $\otimes$ sign.
            }
    \end{enumerate}
    In other words, if we denote by $F(P \otimes \CM)$ the set of formal finite sums of $p \otimes m$,
    then the tensor product is just the factor
    $P \otimes_{\DDD_n} \CM \defas F(P \otimes \CM)/\sim\>$.
\end{definition}
\begin{remark}
    \rm
    \label{rem:Dn-1-action}
    When the left entry $P$ is also a left $\DDD_{n-1}$-module,
    the tensor product $P \otimes_{\DDD_n} \CM$ can be turned into a left
    $\DDD_{n-1}$-module by specifying the action of $g \in \DDD_{n-1}$ on $p \otimes
    m$ as $g \, (p \otimes m) = (g \, p) \otimes m\>$.
\end{remark}

\begin{example}
    \rm
    Let $\CM$ be a left $\DDD_n$-module.
    In order to illustrate the various concepts introduced so far,
    here we show a proof of the following well-known isomorphism between two left $\DDD_{n-1}$-modules written as quotients:
    \begin{equation}
        \DDD_n \big/ \brk{z_n \, \DDD_n} \otimes_{\DDD_n} \CM
        \simeq \CM \big/ \brk{z_n \, \CM}
        \label{eq:Dn-tensor-M}
    \end{equation}
    We start with showing that the right $\DDD_n$-module $\DDD_n \big/ (z_n \, \DDD_n)$
    \brk{see~\exref{ex:D1-right}}, while being also a left $\DDD_{n-1}$-module,
    is in fact \textit{not} a left $\DDD_n$-module.
    Indeed, suppose that it is a left $\DDD_n$-module.
    We denote by $[\bigcdot]$ the equivalence class of some element $\bigcdot$
    in $\DDD_n \big/ \brk{z_n \, \DDD_n}$.
    On one hand, $[z_n] = [0]$ in $\DDD_n \big/ \brk{z_n \, \DDD_n}\>$, so that the action
    of the derivative $\pd{n}$ should also be%
    \footnote{
        Note that the condition $a ([0]+[0])=a([0])=a[0]+a[0]$ implies
        $a[0]=[0]$ for any $a \in \DDD_n$.
    }
    $\pd{n}\sbrk{z_n} = [0]$ in $\DDD_n \big/ \brk{z_n \, \DDD_n}$.
    On the other hand, $\pd{n} [z_n] = [z_n \pd{n} + 1]$ and
    $\sbrk{z_n \pd{n}} =\sbrk{0}$ in $\DDD_n \big/ \brk{z_n \, \DDD_n}\>$,
    leading to $\pd{n} [z_n] = [1] \neq \sbrk{0}$ in $\DDD_n \big/ \brk{z_n \, \DDD_n}\>$,
    which is a contradiction.
    The following technical observation will become useful shortly:
    any non-zero element of $\DDD_n \big/ \brk{z_n \, \DDD_n} \otimes_{\DDD_n} \CM$
    can be expressed as $\sbrk{1} \otimes m$ for some $m \not\in z_n \, \CM\>$.
    Indeed, in a single product any factor $f \in \DDD_n$ can be moved from the
    left to the right slot:
    $\sbrk{f} \otimes m = \sbrk{1} \otimes f m\>$.
    After performing this operation termwise,
    any finite sum can be united into a single product:
    $\sbrk{1} \otimes m_1 + \sbrk{1} \otimes m_2 = \sbrk{1} \otimes(m_1 + m_2)\>$
    according to~\defref{def:tensor}.
    For two non-zero elements $\sbrk{1} \otimes m_1$ and  $\sbrk{1} \otimes
    m_2\>$, where
    $m_i \not\in z_n \, \CM\>$,
    it can be shown that
    the condition $\sbrk{1} \otimes m_1 \not= \sbrk{1} \otimes m_2$ is equivalent to
    $m_1 \not= m_2$ in $\CM$.
    We also note that $\sbrk{0} \otimes m$ and $\sbrk{f} \otimes 0$ are
    $0$ by definition.

    We define a $\DDD_{n-1}$-morphism $\varphi$
    from $\DDD_n \big/ \brk{z_n \, \DDD_n} \otimes_{\DDD_n} \CM$
    to $\CM/(z_n \CM)$ by
    $\varphi\bigbrk{\sbrk{f} \otimes m} \defas [f m]$
    and extend this definition on finite sums
    $\varphi\bigbrk{
        \sbrk{f_1} \otimes m_1
        + \sbrk{f_2} \otimes m_2
    } = \sbrk{f_1 m_1} + \sbrk{f_2 m_2}\>$ by linearity.
    Let us now prove that $\varphi$ is indeed an isomorphism.

    \paragraph{$\varphi$ is well-defined:}
        We prove that if $t_1 \sim t_2$,
        $t_i \in \DDD_n \big/ \brk{z_n \, \DDD_n} \otimes_{\DDD_n} \CM$,
        then $\varphi\brk{t_1} = \varphi\sbrk{t_2}$.
        From the definition of $\varphi$ it follows that if $t \sim 0 \otimes 0$, then $\varphi\brk{t} = \sbrk{0}$.
        Any $0$ in $\DDD_n \big/ \brk{z_n \, \DDD_n} \otimes_{\DDD_n} \CM$
        can be expressed as $\sbrk{z_n} \otimes m$, $m \in \CM$
        after a simplification by ``+''.
        Note that $\varphi$ satisfies $\varphi(t_1+t_2)=\varphi(t_1)+\varphi(t_2)$
        for any $t_i \in \DDD_n \big/ \brk{z_n \, D_n} \otimes_{\DDD_n} \CM$.
        Since
        $\varphi\bigbrk{\sbrk{z_n} \otimes m} = \sbrk{z_n m}$, $z_n m \in z_n \CM$,
        then $\sbrk{0} \otimes m$ is sent to $[0]$ by $\varphi$.

    \paragraph{$\varphi$ is $\DDD_{n-1}$-morphism:}
        Let $c_i \in \DDD_{n-1}$, $f_i \in \DDD_n$, and $m_i \in \CM$.
        \begin{alignat*}{3}
            &\varphi\bigbrk{
                  c_1 \, \brk{\sbrk{f_1} \otimes m_1}
                + c_2 \, \brk{\sbrk{f_2} \otimes m_2}
            } &&
            \\
            &=\varphi\bigbrk{
                  \sbrk{c_1 \, f_1} \otimes m_1
                + \sbrk{c_2 \, f_2} \otimes m_2
            }
            &&\mbox{(\remref{rem:Dn-1-action})}
            \\
            &= \sbrk{c_1 \, f_1 \, m_1}
             + \sbrk{c_2 \, f_2 \, m_2}
            &&\mbox{(definition of $\varphi$)}
            \\
            &= c_1 \, \sbrk{f_1 \, m_1}
             + c_2 \, \sbrk{f_2 \, m_2}
            &&\mbox{(action of $c_i$ on $\CM \big/ \brk{z_n \, \CM}$)}
            \\
            &= c_1 \, \varphi\bigbrk{\sbrk{f_1} \otimes m_1}
             + c_2 \, \varphi\bigbrk{\sbrk{f_2} \otimes m_2}
            \qquad
            &&\mbox{(definition of $\varphi$)}
        \end{alignat*}

    \paragraph{$\varphi$ is surjective:}
        Clearly the image $\varphi\bigbrk{\sbrk{1} \otimes \CM}$ covers the whole
        $\CM \big/ z_n \, \CM$.

    \paragraph{$\varphi$ is injective:}
        Due to the linearity of $\varphi$, it is enough to prove that if
        $\varphi\bigbrk{t} = 0$
        for some $t \in \DDD_n \big/ \brk{z_n \, \DDD_n} \otimes_{\DDD_n} \CM$
        ,
        then necessarily $t = 0$.
        Suppose that $t \not= 0$.
        Then there exists $m \in \CM$ such that $t \sim \sbrk{1} \otimes m$.
        Since $\varphi\bigbrk{\sbrk{1} \otimes m} = \sbrk{m} = 0$, we have
        $m \in z_n \, \CM$.
        Writing $m = z_n \, m'$, $m' \in \CM$, we have
        $\sbrk{1} \otimes \brk{z_n \, m'} \sim \sbrk{z_n} \otimes m' \sim 0 \otimes m'$,
        which means that $t = 0$.

        Thus we constructed an explicit isomorphism $\varphi$ between
        $\DDD_n \big/ \brk{z_n \, \DDD_n} \otimes_{\DDD_n} \CM$ and $\CM \big/ \brk{z_n
        \, \CM}$ and proved
        the identity~\eqref{eq:Dn-tensor-M}.
    \hfill$\blacksquare$
\end{example}

\begin{example}
    \rm
    Another example of tensor product is:
    \begin{equation}
        \CN=\CC[z_1,\ldots,z_n] \big/ \langle z_n \rangle
        \otimes_{\CC[z_1, \ldots,z_n]} \CM
        \>.
        \nonumber
    \end{equation}
    For $1 \leq i \leq n-1$
    we define the action of $\pd{i}$ on the product
    $\sbrk{f} \otimes m\>$ as
    $\pd{i} (\sbrk{f} \otimes m)
    \defas \sbrk{\pd{i} f} \otimes m
    + \sbrk{f} \otimes (\pd{i} m)\>$,
    where $\sbrk{f} \in \CC[z_1, \ldots, z_n] \big/ \langle z_n \rangle$ and
    $m \in \CM\>$.
    Multiplication by $h \in \CC[z_1, \ldots, z_{n-1}]$ is defined by
    $h \, (\sbrk{f} \otimes m) \defas \sbrk{h \, f} \otimes m\>$ \brk{similar to~\remref{rem:Dn-1-action}}.
    Then we can prove that $\CN$ is a left $\DDD_{n-1}$-module
    \begin{equation} \label{nteq:pullback}
     \CN \simeq \CM \big/ \brk{z_n \, \CM}
      \>.
    \end{equation}
    Indeed, the isomorphism $\varphi$ from $\CN$ to $\CM \big/ \brk{z_n \, \CM}$ is given by
    $\varphi\bigbrk{\sbrk{f} \otimes m} \defas \sbrk{f \, m}\>$.
    Note that for $\sbrk{f} = \sbrk{z_n \, h} \sim \sbrk{0}$ we have
    $\varphi\bigbrk{\sbrk{z_n \, h} \otimes m} = \sbrk{z_n \, h \, m} \sim
    \sbrk{0}\>$, which makes the map $\varphi$ is well-defined.
    It is left to the reader to prove that $\varphi$ is an isomorphism.
    \hfill$\blacksquare$
\end{example}

\begin{example}
    \rm
    \label{ex:tensors-3}
    Our final example of tensor product is
    $\CC[z][1/f] \otimes_{\CC[z]} \DDD_n$\ ,
    where $\CC[z]=\CC[z_1, \ldots, z_n]$ and $f$ is a polynomial in $z\>$.
    Since the left entry $\CC[z][1/f]=\CC[z_1, \ldots, z_n,1/f]$ can be viewed as
    a left $\DDD_n$-module, we can regard the tensor product
    $\CC[z][1/f] \otimes_{\CC[z]} \DDD_n$ as a left $\DDD_n$-module by
    $\pd{i} ( g \otimes \ell ) = (\pd{i}g) \otimes \ell + g \otimes (\pd{i}\ell)\>$,
    for $g \in \CC[z][1/f]$ and $\ell \in \DDD_n\>$.
    It is an exercise for the reader to describe the action of a general element
    of $\DDD_n$ explicitly.
    Furthermore, we can introduce the ring structure on
    $\CC[z][1/f] \otimes_{\CC[z]} \DDD_n$ by defining the product ``$\cdot$''
    between two elements as
    $(g_1 \otimes \ell_1) \cdot (g_2 \otimes \ell_2)
    = g_1 \otimes \ell_1 \, (g_2 \otimes \ell_2)\>$, where
    $g_i \in \CC[z][1/f]$ and $\ell_i \in \DDD_n\>$.
    Roughly speaking, $\CC[z][1/f] \otimes_{\CC[z]} \DDD_n$ is
    $\CC\langle z_1, \ldots, z_n, 1/f, \pd{1}, \ldots, \pd{n} \rangle\>$.
    \hfill$\blacksquare$
\end{example}

In \propref{prop:from_holonomic_to_zero_dim_ideal},
we construct a zero-dimensional $\cR$-module from a holonomic
$\cD_Y$-module by a tensor product.
The zero-dimensional one can be expressed as a quotient by a left ideal
as the right hand side of (\ref{eqn:A3}) below.
Let us sketch a proof of the isomorphism.
Take a left ideal $\cI$ in the Weyl algebra $\CD_n\>$, and let $\CM \defas \DDD_n
\big/ \cI\>$ be a left $\DDD_n$-module.
Let $\CR_n=\CC(z) \langle \pd{1}, \ldots, \pd{n}\rangle$
be the ring of differential operators with rational function coefficients,
where $\CC(z) \defas \CC(z_1, \ldots, z_n)$ is the field of rational functions.
The left $\CR_n$-modules $\CR_n \otimes_{\CD_n} \CM$ and
$\CR_n \big/ \brk{\CR_n \, \cI}$ are isomorphic as left $\CR_n$-modules.
Here the left action of $h \ni \CR_n$ on $f \otimes g \in \CR_n \otimes_{\CD_n} \CM$
is defined by $h \, f \otimes g\>$.
The isomorphism is given by
\begin{equation}
    \CR_n \otimes_{\CD_n} \CM \ni f \otimes g
    \mapsto f g \in \CR_n \big/ \brk{\CR_n \, \cI}
    \>,
    \label{eqn:A3}
\end{equation}
and we can easily check that it is an isomorphism of left $\CR_n$-modules
by noting that $\CR_n$ is an associative algebra.
Note that
$f_1 f_2 \otimes g = f_1 \otimes f_2 \, g$  in $\CR_n \otimes_{\CD_n} \CM$
when $f_2 \in \CD_n\>$.
We can show analogously that $\CR_n \otimes_{\CD_n} \CD_n^m \big/ \cJ$ is
isomorphic to $\CR_n^m \big/ \brk{\CR_n \, \cJ}$ as left $\CR_n$-modules, where
$\cJ$ is some left submodule of $\CD_n\>$.

\begin{remark}\rm \label{rmk:cyclic}
We referred to the algorithm of constructing a \emph{cyclic vector} in \secref{sec:preliminaries}.
It is a way to construct an ideal given a Pfaffian system.

In order to apply the cyclic vector to $\CR_n^m\big/ \brk{\CR_n \, \cJ}$, we first compute the Weyl closure of $\cJ$, and then apply the algorithm of \cite{leykin-cyclic} to the closure.
This procedure is exact, but it needs huge computational resources.
We present an easy but inexact method. For simplicity, we illustrate the case of $m=2$, $n=1$ ($\pd{}=\pd{1}$, $z=z_1$), and $\CR_1 \cJ$ is generated by
$$ (\pd{},0)-(a,b), \quad (0,\pd{})-(c,d), \quad a,b,c,d \in \CC(z). $$
These generators induce a system of differential equations for two unknown functions $u^1$ and $u^2$:
\begin{equation} \label{eqnut:sys}
 \pd{}\bullet u^1-au^1-bu^2=0, \quad \pd{}\bullet u^2-cu^1-du^2=0.
\end{equation}
Differentiating both equations w.r.t.~$z$, we obtain two more identities:
\begin{equation} \label{eqnut:sys2}
 u^1_{zz}-a_z u^1-au^1_z-b_zu^2-bu^2_z=0,
 u^2_{zz}-c_zu^1-cu^1_z-d_zu^2-du^2_z=0
\end{equation}
where ${u}_z$ and ${u}_{zz}$ denote the first and second derivatives w.r.t.~$z$.
We regard $u^1,u^2, u^1_z,u^2_z, u^1_{zz}$ and $ u^2_{zz}$ as independent variables in the polynomial ring $\CC(z)[u^1,u^2, u^1_z,u^2_z, u^1_{zz},u^2_{zz}]$
and eliminate $u^2, u^2_z, u^2_{zz}$ from \eqref{eqnut:sys} and \eqref{eqnut:sys2}.
Then we obtain an ODE $u^1_{zz} + c_1(z) u^1_z + c_2(z) u^1=0$ for the function $u^1$.
In some cases, the left ideal $\cI=\langle \pd{}^2+c_1\pd{}+c_2 \rangle$
in $\CR_1$ gives an isomorphism of
$\CR_1^2\big/ \brk{\CR_n \, \cJ}$ and $\CR_1/\cI$.
\end{remark}

\begin{remark}\rm\label{rmk:history}
As a historical concluding remark, let us rephrase the notion of a \emph{solution} to a system of differential equations in the language of homological algebra, which is a starting point of the theory of $\CD$-modules.
Let us denote by $\OO^{an}(U)$ the set of holomorphic functions%
\footnote{
    The superscript ``an'' stands for ``analytic''. On the other hand, $\OO^{alg}(U)$ is the set
    of rational functions that do not have poles on $U$.
    They are called regular functions on $U$.
    We omit the superscript ``alg'', that stands for ``algebraic'', in this paper.
}
on an open set $U\>$, and by
\begin{equation}
    \label{nteq:homSol}
    \Hom_{\CD_n}(\CM,\OO^{an}(U))\>
\end{equation}
the set of left $\CD_n$-morphisms from the left $\CD_n$-module $\CM$ to
$\OO^{an}(U)\>$.
The correspondence between a homomorphism
$\varphi \in \Hom_{\CD_n}(\CM,\OO^{an}(U))$ and a solution $f$ is given by
$f=\varphi\bigbrk{\sbrk{1}}\>$.
In fact, take $\CM \defas \DDD_n \big/ \cI$, and let $\ell$ be
an element of the left ideal $\cI\>$, so that $\sbrk{\ell} \sim
\sbrk{0}$ in $\CM\>$.
Since $\varphi$ is a left $\CD_{n}$-morphism, we have:
\begin{equation}
    0
    = \varphi\bigbrk{\sbrk{0}}
    = \varphi\bigbrk{\sbrk{\ell}}
    = \varphi\bigbrk{\ell \, \sbrk{1}}
    = \ell \bigcdot \varphi\bigbrk{\sbrk{1}}
    \>,
\end{equation}
which means that $\varphi\bigbrk{\sbrk{1}} \in \OO^{an}(U)$ is a solution to
the differential equation $\ell \bigcdot f = 0\>$.
Conversely, let $f$ be a solution of $\cI\>$, meaning that for any $\ell \in
\cI$ it satisfies $\ell \bigcdot f = 0\>$.
Then for and other element $m \in \CD_{n}\>$, the mapping
$m \mapsto m \bullet f \in \OO^{an}(U)$ defines a left $\CD_{n}$-morphism,
i.e. an element of $\Hom_{\CD_n}(\CM,\OO^{an}(U))\>$.
The derived functor of $\Hom$ gives cohomological solution spaces
in the theory of $\CD$-modules \cite{Hotta-Tanisaki-Takeuchi-2008}.
\end{remark}

\section{Normal form by Moser reduction}
\label{subsec:Moser_Reduction}

In \secref{subsec:Deligne_extension}, it was emphasised that a Pfaffian system must be in \emph{normal form} before the restriction protocol can proceed.
A given set of Pfaffian matrices $P_1, \ldots, P_m$ is in normal form w.r.t.~$z_1$ when $P_1$ satisfies two conditions:
\begin{enumerate}
    \item \emph{Logarithmic}: $P_1$ has a simple pole at $z_1 = 0$.
    \item \emph{Non-resonant}: The eigenvalues of the residue matrix $P_{1,-1}(z')$ have no integral difference, and $0$ is the unique integer eigenvalue.
\end{enumerate}
As explained later in this section, the non-resonance condition additionally implies that $P_2, \ldots, P_m$ are finite at $z_1 = 0$.
This is important, because in the restriction protocol of \secref{subsec:Deligne_extension} we use the matrices $P_2(0,z'), \ldots, P_m(0,z')$, where $z' = (z_2, \ldots, z_m)$.

Normal form can be obtained by following the Moser reduction scheme \cite{Moser}, which was later made algorithmic by Barkatou \cite{Barkatou1995rational} (see also \cite{Gerard-Ramis-1983}).
For CAS implementations, see e.g.~\texttt{ISOLDE} \cite{barkatou2013isolde} and \texttt{Fuchsia} \cite{Gituliar:2017vzm}.
For the reader's convenience, in this section we review Moser reduction, largely following \cite{Barkatou-Jaroschek-Maddah-2017}.

\subsection{Logarithmic \texorpdfstring{$P_1$}{}}

Let $\mId_r$ denote the $r$-dimensional identity matrix. The following theorem is a result of \cite{Moser, Barkatou1995rational}.

\begin{theorem}\label{thm:Moser}
    Let $P(z)$ be an $r\times r$ matrix, depending on a single variable $z$, with entries in the field of rational functions $\CC(z)$.
    Assume that a system of ODEs
    \begin{equation}\label{eq:regular_ode}
        \frac{d \MI}{dz}=P(z) \cdot \MI
    \end{equation}
    is regular singular at $z=0$.
    After a finite sequence of gauge transforms $P \to G[P] = G^{-1}(P \cdot G - \pd{z} G)$ with $G$ of the form
    \eq{
        G= \text{a constant matrix,}
        \quad
        \text{or}
        \quad
        G =
        \lrsbrk{
            \begin{array}{cc}
                z\mId_{r^\prime}&\mZero\\
                \mZero&\mId_{r-r^\prime}
            \end{array}
        } \, , \
        \text{for} \ r' \in \NN
        \, ,
    }
    the ODEs \eqref{eq:regular_ode} can be transformed into a system with only simple poles at $z=0$.
    Such a sequence of gauge transforms can be found algorithmically.
\end{theorem}

\begin{remark}\label{rem:Moser} \rm
    The gauge transformation matrices can be constructed without algebraic extensions.
    In other words, if the entries of $P(z)$ belong to $K(z)$ for some subfield $K\subset\C$, then $G$ can be taken to have coefficients in $K$.
    This is useful when the system \eqref{eq:regular_ode} contains additional parameters independent of $z$.
\end{remark}

The theorem above pertains to ODEs, but Pfaffian systems are PDEs in general.
Fortunately, we can still apply the theorem to the PDE case by treating one variable at a time.
Consider a Pfaffian system in $m$ variables: $\pd{i} \MI = P_i(z) \cdot \MI\>$
, for $i = 1, \ldots, m\>$.
Assume that the $P_i$ have higher-order poles along $z_1=0\>$.
Treat
\eq{
    \frac{\dd \MI}{\dd z_1} = P_1(z_1, z^\prime) \cdot \MI
}
as an ODE in $z_1$ with parameters $z^\prime = \brk{z_2, \ldots, z_m}\>$.
By \thmref{thm:Moser} and \remref{rem:Moser}, we can find a gauge
transformation matrix $G\>$ such that $G[P_i]$ only has
a simple pole along $z_1=0\>$.  We may thus assume that the Pfaffian matrices
take the form:
\eq{
    P_{1}(z) = \sum_{n=-1}^\infty P_{1,n}(z') \, z_1^n \, .
    \label{eq:expand_P_i(z)}
}
This solves the problem of finding a logarithmic set of Pfaffian matrices.

\subsection{Non-resonant \texorpdfstring{$P_{1,-1}$}{}}

Next we show how to cure the resonant spectrum of $P_{1,-1}(z')$ .
The first thing to note is that the eigenvalues of $P_{1,-1}(z^\prime)$ are independent of $z^\prime$ \cite[Theorem 12.1]{haraoka2020linear}.
Therefore, if $P_{1,-1}(z^\prime)$ is non-resonant at a point, then it is non-resonant everywhere.
Fix a generic value $z^\prime=\zvalue$, and set $P:=P_{1,-1}(\zvalue)$.
Suppose that $P$ is in Jordan normal form, namely in block-diagonal form with Jordan cells $J_1, \ldots, J_k$.
Let $r_i$ denote the size of each $J_i$.
Subdivide the matrix $P_1(z)$ into cells as
\begin{equation}
    P_1=
    \lrsbrk{
        \begin{array}{c|c|c|c}
        (P_1)_{11}&(P_1)_{12}&\dots&(P_1)_{1k}\\
        \hline
        (P_1)_{21}&(P_1)_{22}& & \vdots\\
        \hline
        \vdots&&\ddots& \\
        \hline
        (P_1)_{k1}&\dots& &(P_1)_{kk}
        \end{array}
    }
\end{equation}
and construct the gauge transformation matrix
\eq{
    G=
    \lrsbrk{
        \begin{array}{c|c|c|c}
        z_1 \mId_{r_1} & \mZero &\dots& \mZero \\
        \hline
        \mZero & \mId_{r_2}& & \vdots\\
        \hline
        \vdots&&\ddots& \\
        \hline
        \mZero &\dots& &\mId_{r_k}
    \end{array}
    } \, .
}
By assumption, $P_1$ has a simple pole along $z_1 = 0$, and the same will hold true for $G[P_1]$.
A short calculation shows that the residue matrix of $G[P_1]$ evaluated at $z' = \zvalue$ is of the form
\begin{equation}
    \label{gauge_transformation_non_resonance}
    G[P] =
    \lrsbrk{
        \begin{array}{c|c|c|c|c}
        J_1-\mId_{r_1}&(P_{1,0})_{12}&(P_{1,0})_{13}&\cdots&(P_{1,0})_{1k}\\
        \hline
        \mZero &J_2& \mZero &\dots & \vdots\\
        \hline
        \vdots&&\ddots& &\vdots\\
        \hline
        \mZero &\dots& &&J_{k}
        \end{array}
    } \, .
\end{equation}
We have thereby shifted the eigenvalue corresponding to the Jordan block $J_1$ by $1$.
Repeating this procedure sufficiently many times for each $J_i$, we will render $P$ non-resonant at $z' = \zvalue$, and therefore non-resonant everywhere.

Let us now show that non-resonance implies finiteness of $P_2, \ldots, P_m$.
The following lemma is well-known.
\begin{lemma}
    \label{lem:eigenvalue}
    Suppose that $A,B \in M_r(\CC)$ are $r\times r$ matrices with their entries in $\CC$.
    Let $\lambda(A), \lambda(B)$ denote the sets of eigenvalues for $A,B$.
    Consider the linear map
    \begin{equation}
        T_{A,B}:M_r(\CC)\ni X\mapsto A \cdot X-X\cdot B\in M_r(\CC) \, .
    \end{equation}
    Then the set of eigenvalues of $T_{A,B}$ equals the set of eigenvalue differences $\lambda(A) - \lambda(B)$.
\end{lemma}

\noindent
It follows from the integrability condition \eqref{eqn:integrability_condition} and \eqref{eq:expand_P_i(z)} that
\eq{
    \label{r_eqn:IC}
    \lrsbrk{
       P_{1,-1}(z'), \, P_{i,n}(z')
    }
    =
    n \, P_{i,n}(z')
    \quad , \quad
    i = 2, \ldots m
    \quad , \quad
    n < 0 \, .
}
We can rewrite this in terms of the map $T_{A,B}$ as
\eq{
    T_{P_{1, -1}, \, P_{1, -1}}(P_{i, n}) = n \, P_{i, n}(z') \, ,
}
meaning that $P_{i,n}$ is an eigenvector
of $T_{P_{1,-1}, \, P_{1,-1}}$ with eigenvalue $n$.
On the other hand, \lemref{lem:eigenvalue} shows that any eigenvalue of $T_{P_{1,-1}, \, P_{1,-1}}$ must be of the form $\lambda_a - \lambda_b$, where $\lambda_{a,b}$ are eigenvalues of $P_{1,-1}$.
If $P_{1,-1}$ is non-resonant, we cannot have that $\lambda_a - \lambda_b = n$.
Consequently, $P_{2,n}(z') = 0$ for $n < 0$.

Summarizing the previous two sections, we have a
\begin{theorem}\label{thm:2}
The algorithm for transforming a Pfaffian system into normal form w.r.t.~$z_1$ is as follows:
\begin{enumerate}
    \item Apply the Moser reduction scheme from \thmref{thm:Moser} over the field $\CC(z')$ to ensure that $P_1$ has a simple pole at $z_1 = 0$.
    \item Apply gauge transformations of the form \eqref{gauge_transformation_non_resonance} until $P_{1,-1}$ becomes non-resonant.
\end{enumerate}
\end{theorem}

\noindent This theorem can be sequentially applied to several variables, in case one seeks to restrict more than one variable.

\subsection{Jordan normal form without field extensions}\label{subsec:jordan_normal_form}
We have already noted that Moser reduction can be performed using only linear algebra without field extensions.
Step $1$ of \thmref{thm:2} is therefore efficient to implement.
Step $2$, on the other hand, apparently requires one to solve polynomials equations, as this is the standard way of computing Jordan normal forms.
Here we argue that algebraic extensions can \emph{also} be avoided in step $2$, wherefore \thmref{thm:2} does not require them at all.
Fix a subfield $K\subset\C$.
The purpose of this subsection is to prove the following lemma.

\begin{lemma}
There exists an invertible matrix $Q\in K^{r\times r}$ such that $Q^{-1}PQ$ is block diagonalized as
\eq{
Q^{-1}PQ
=
\lrsbrk{
        \begin{array}{c|c|c|c}
        B_1 & \mZero & \dots & \mZero \\
        \hline
        \mZero & B_2 & & \vdots\\
        \hline
        \vdots & & \ddots & \\
        \hline
        \mZero & \dots & & B_k
        \end{array}
    }
}
    where each $B_i$ is non-resonant.
    We write \emph{$\spec{B_i}$} for the set of eigenvalues of the block $B_i$.
    Then, for any pair of indices $i,j$, either
    \emph{$(\spec{B_i}+\ZZ) \cap \spec{B_j} = \varnothing$}
    or there exists a unique integer $m$ such that
    \emph{$\spec{B_i} + m = \spec{B_j}$}.
\end{lemma}
Once this lemma is proved, we can conclude that the gauge transformations involving Jordan normal forms from step $2$ of \thmref{thm:2} can be algorithmically chosen from $GL_r(K[z_1^\pm])$, which is free of field extensions.
Indeed, each Jordan block $J_i$ of $P$ is replaced by the block $B_i$.
Let us begin with a simple lemma.
\begin{lemma}
Let $f(t)\in K[t]$ be an irreducible polynomial.
If $\alpha$ is a root of $f(t)=0$, no integral shift of $\alpha$ is a root.
\end{lemma}

\begin{proof}
Set $f(t)=a_0+\cdots+a_{N-1}t^{N-1}+t^N$ and fix an integer $m\neq 0$ such that $f(\alpha+m)=0$.
We prove that the existence of $m$ leads to a contradiction.
A swift computation shows that
\eq{
    0 =
    f(\alpha + m) =
    f(\alpha) + (\text{terms of order } \leq N-1 \text{ in } \alpha) :=
    f(\alpha) + g(\alpha) =
    g(\alpha) \, ,
}
where in the last step we used that $\alpha$ is a root of $f$.
The coefficient of $\alpha^{N-1}$ in $g(\alpha)$ is $Nm$, which is clearly non-zero.
This means that $g(\alpha)=0$, i.e.~$\alpha$ is also a root of $g(t)$.
But this is a contradiction, because $f(t)$ is the minimal polynomial\footnote{
    The minimal polynomial of $\alpha$ over $K$ is the polynomial with the
    smallest degree, whose roots contain $\alpha$.
}.
\end{proof}

Suppose we are given a pair of irreducible polynomials $f_1,f_2\in K[t]$.
If they share a root $\alpha$ modulo $\ZZ$, that means that $f_1(t+s) = f_2(t)$ has an integral solution as a system of equations in $s$.
This is seen by expanding both sides and collecting powers of $t$, leading to a polynomial system for $s$ of the form $g_0(s) = \ldots = g_{N-1}(s) = 0$.
The GCD $g(s)$ is used whether we have an integral solution or not.

Let us recall a well-known fact from linear algebra (see e.g. \cite[Chapter 12]{Dummit-Foote-PID-Jordan}).
Given an $r \times r$ matrix $P$ with entries in $K$, we can equip $K^r$ with the structure of a $K[t]$-module by letting $t$ act on $K^r$ via left multiplication by $P$.
Consider an exact sequence
\begin{equation}
    K[t]^r
    \quad \overset{t\mId_r-P}{\longrightarrow} \quad
    K[t]^r
    \quad \longrightarrow \quad
    K^r
    \quad \longrightarrow \quad
    0 \, .
\end{equation}
The last morphism sends $f(t)\boldsymbol{e}_i$ to $f(P) \cdot \boldsymbol{e}_i$, where $f(t) \in K[t]$ and $\boldsymbol{e}_i$ is the standard unit vector.
The fundamental theorem of principal ideal domains shows that $t\mId_r-P$ is transformed into a canonical form
\eq{
    L(t)(t\mId_r-P)R(t)
    =
    \diag{1, \, \ldots, \, 1, \, d_1(t), \, \ldots, \, d_p(t)}
}
for some invertible polynomial matrices $L(t),R(t)\in GL(r,K[t])$ with $d_1(t)|\cdots|d_p(t)$.
According to the Chinese remainder theorem, we thus have an isomorphism (which can be made algorithmically explicit)
\eq{
    K^r
    \overset{\varphi}{\leftarrow} \
    \bigoplus_{i=1}^p K[t] / d_i(t)
    \ \simeq \
    \bigoplus_{i=1}^p \bigoplus_{j=1}^{q_i} K[t] / d_{ij}(t)^{m_{ij}}
}
where $d_i(t)=\prod_{j=1}^{q_i}d_{ij}(t)^{m_{ij}}$ is the irreducible decomposition.
The morphism $\varphi$ sends any element $(f_i(t))_{i=1}^p\in \bigoplus_{i=1}^p K[t] / d_i(t)$ to $ \sum_{i=1}^pf(t)v_i(t)|_{t\mapsto P}$ where $v_i(t)$ is the $(r-p+i)$-th column vector of $L(t)^{-1}$ and the subscript $t\mapsto P$ stands for the operation of replacing an expression $t^av$ $(v\in K^r)$ by $P^av$.
Each subspace $K[t]/d_{ij}(t)^{m_{ij}}$ is invariant under multiplication by $t$, i.e.~left multiplication by $P$.
This means that we can find a basis of $K^r$ such that the left multiplication by $P$ is block-diagonalized into a direct sum of multiplication by $t$ on $K[t] / d_{ij}(t)^{m_{ij}}$.
It is not difficult to check that the $d_{ij}(t)$ share a common root modulo $\ZZ$, as explained in the previous paragraph.
Finally, any basis of each block $K[t] / d_{ij}(t)^{m_{ij}}$ is sent to a set of column vectors $B_{ij}$.
The matrix $Q$ is obtained by aligning these  $B_{ij}$ to form a square matrix.

\begin{remark}\rm
    In the examples of Pfaffian systems for Feynman integrals discussed in this paper, any eigenvalue of the residue matrix is an integral linear combination of $\varepsilon$ and $\varepsilon\delta$.
    Since the Euler integral \eqref{f_Gamma(z)} admits a meromorphic continuation as in  \cite[Theorem 2.5]{berkesch2014euler}, we expect this to be true in general.
\end{remark}

\section{Improvements on restriction algorithms}
\label{appendix:improve_all}

Here we remark on subtle improvements to steps in the restriction algorithms presented in the main body of this text.

\subsection{\texorpdfstring{$\cD$}{}-module-level restriction} \label{appendix:improve}

The $b$-function played a central role in the $\cD$-module restriction protocol of \secref{sec:NTsec1}.
However, it is well-known that the $b$-function is computationally expensive to obtain.
When we can find a logarithmic connection for $\CD'/\cI$ via another method
(recall the definition of $\cD'$ below \remref{rem:monomial-basis}), such as the Macaulay matrix method, we can utilize the following proposition to compute the $b$-function more easily:
\begin{proposition}
Let $S=\sbrk{s_1=1, s_2, \ldots, s_r}\tr$ be a set of standard monomials for $\cI$
such that $\ord s_i = 0$.
We denote by ${\overline \cI}=\CR'\cI \cap \CD'$ the Weyl closure of $\cI$.
Assume that the Pfaffian system with respect to $S$ has the form
\begin{eqnarray*}
 \pd{x} S &=& \left(\frac{P_{-1}}{x} + P_0 + x P_1 + \ldots \right) \cdot S \quad \mod \ \CR'\cI \\
 \pd{y} S &=& \left( Q_0 + x Q_1 + \ldots \right) \cdot S \quad \mod \ \CR'\cI \\
\end{eqnarray*}
Then the eigenvalues of $P_{-1}$ contain the roots of the $b$-function
of ${\overline \cI}$
for the weight $w=(1,0)$.
\end{proposition}

\noindent
{\it Proof}\/.
We note that $x \pd{x} S = P_{-1} \cdot S$ modulo $F_{-1}+{\overline \cI}$.
Here $F_{-1}$ is the filtration defined in (\ref{eq:KM-filtration}).
Let $\varphi(s)$ be the characteristic polynomial of $P_{-1}$.
By the Cayley-Hamilton theorem, we have $\varphi(P_{-1})=0$.
Hence, we have $\varphi(x\pd{x}) S = 0$ modulo $F_{-1}+{\overline \cI}$.
In particular, we have $\varphi(x\pd{x})=0$ modulo $F_{-1}+{\overline \cI}$
because $s_1=1$.
Therefore, the $b$-function is a factor of $\varphi(s)$.
\qed

\bigbreak

If we have Pfaffian matrices for a logarithmic connection
with the singular locus $x=0$,
we can accelerate \algref{alg:alg1} with help from the residue matrix $P_{-1}$ as follows.
Let $S$ be a column vector consisting of standard monomials
of a Gr\"obner basis
of $\cI$ in $\CR'=\CC(x,y)\langle \pd{x}, \pd{y} \rangle$.
Suppose that
\begin{equation} \label{eq:pfx}
x \pd{x} S - (P_{-1} + P_0 x + \ldots) \cdot S \equiv 0 \quad \mathrm{mod} \ {\overline \cI}
\end{equation}
where $x P = P_{-1}+P_0 x + \ldots$ is the Pfaffian matrix in the $x$ direction multiplied by $x$.
Taking the limit $x=0$ in \eqref{eq:pfx}, we have
$P_{-1} \cdot S \equiv 0 \quad \mathrm{mod}\ {\overline \cI}+x\CD'$,
which gives a relation among standard monomials modulo ${\overline \cI}+x\CD'$.

We can obtain higher-order relations by applying $\pd{x}^j$
to \eqref{eq:pfx} and taking the limit $x=0$.
For example, applying $\pd{x}$ in conjunction with the relation $\pd{x}S \equiv P \cdot S$,
\begin{equation}
    \pd{x}S + x \pd{x}^2 S \equiv
     \big(P_{-1} + O(x)\big) \cdot S +
     \big(P_{-1}+P_0 x + O(x^2)\big) \cdot x^{-1} \big(P_{-1} + P_0 x + O(x^2)\big) \cdot S \, .
\end{equation}
When $P_{-1} \cdot S \equiv 0$, the limit $x=0$ yields
\begin{equation}
\pd{x}S \equiv
 P_{-1} \cdot S + (P_{-1} \cdot P_0 + P_{0} \cdot P_{-1}) \cdot S \equiv 0 \quad \mathrm{mod}\ {\overline \cI} + x \CD' \, .
\end{equation}

Let $\PP$ be the set of the relations coming from $P_{-1} \cdot S \equiv 0$ and higher-order obtained via the Pfaffian equation \eqref{eq:pfx}.
\begin{algorithm}[{\rm Rational restriction to $x=0$ using the residue matrix $P_{-1}$}]\rm \ \label{alg:alg2} \\
    \underline{Input}: Generators $\{f_1, \ldots, f_\mu\}$ of a holonomic ideal $\cI$ in $\CD'$.
    The relations $\PP$.
    Holonomic rank $r$ of the restriction to $x=0$ of $\cI$.
    A non-negative integer $\gamma$ such that $\gamma \geq \max(s_0+1,s_1+1)$. \\
    \underline{Output}: Generators of a left submodule $\cJ \cap \CD^{\gamma}$ of $\CD^{\gamma}$ such that
    $\CR^{\gamma}/\CR(\cJ\cap \CD^{\gamma}) = \CR \otimes_\CD \CD'/(\cI+x \CD')$.
    \begin{algorithmic}[1]
        \State $w=(1,0)$
        \State $k = \gamma-1$
        \Repeat
        \State  $\cJ
            = \CD \cdot \Bigbrc{
                v_k\, \bigbrk{\, :\pd{x}^j f_i:{\big|_{x=0}}}
                \Bigm|
                \ord_{(-w,w)} (\pd{x}^j f_i) \leq m, 1 \leq i \leq N, j \in \NN_0
        }$ \\
        \hspace{0.9cm} $\cup \ \bigbrc{ v_k(p) \bigm| p \in \PP, \ord_{(-w,w)} (p) \leq k }$
        \State $k:=k+1$
        \Until{$\mathrm{rank}\bigbrk{\CD^{\gamma}/\cJ \cap \CD^{\gamma}} = r$}
        \State \Return $\cJ$
    \end{algorithmic}
\end{algorithm}

\begin{theorem}
\algref{alg:alg2} stops and gives the correct answer.
\end{theorem}

\noindent
{\it Proof}\/.
Rows of the Pfaffian system belong to the Weyl closure ${\overline \cI}=(\CR' \cI)\cap \CD'$.
Then the elements of $\PP$ belong to ${\overline \cI} + x \CD'$,
and this gives an expression for $\CR \otimes_\CD (\CD'/(\cI+x\CD')$
by \propref{prop:prop1}.
Combining this fact with the proof of \thmref{th:th3},
we reach the desired conclusion.
\qed
\bigbreak

\begin{example}\rm \label{ex:triangle-1m}
To illustrate the difference between \algref{alg:alg1} and \algref{alg:alg2}, we consider the Feynman integral for the one-loop triangle with one internal mass (for comparison, see the example from \secref{sec:gkz_1L_1M_triangle}).
The Lee-Pomeransky polynomial takes the form
\eq{
    g(z;x) &=
    z_1 x_1 + z_2 x_2 + z_3 x_3 + z_4 x_1^2+z_5 x_1 x_2 + z_6 x_1 x_3 \\
    z_1&=z_2=z_3=z_4=z_5=1
    \quad \text{and} \quad
    z_6 \text{ generic} .
    \label{eq:triangle_LP_improv}
}
The associated GKZ system has the data
\eq{
    A=\lrsbrk{
        \begin{array}{cccccc}
            1&  1&  1&  1&  1&  1 \\
            1& \mzero& \mzero&  2&  1&  1 \\
            \mzero&  1& \mzero& \mzero&  1& \mzero \\
            \mzero& \mzero&  1& \mzero& \mzero&  1 \\
        \end{array}
    }
    \quad , \quad
    \beta = \sbrk{b_1, b_2, b_3, b_4} \, .
}
The GKZ case with generic $z_1,\ldots,z_6$ is of holonomic rank $3$.
The Feynman integral, subject to \eqref{eq:triangle_LP_improv}, has rank $2$.
We proceed to compute this drop in rank via a $\cD$-module restriction.

To begin, choose the simplex $\sigma = \sbrk{1,2,3,4}$.
The remaining variables are $z_5$ and $z_6$, and we seek the restriction $z_5 = 1$.
Shifting $z_5 \to z_5 + 1$, we can alternatively restrict $z_5 = 0$.
The generators of the GKZ system then read (see \cite[Appendix A]{Chestnov:2022alh} for how to rewrite generators given a simplex $\sigma$)
\begin{eqnarray*}
f_1&=&    (  {z}_{5}  - {z}_{6}+ 1)  \underline{{\partial}_5 \partial_6 }+  {b}_{4}  {\partial}_5 -  {b}_{3}  \partial_6 \\[5pt]
f_2&=&      \big[   (  - {z}_{6}+ 1)  {z}_{5}  - {z}_{6}+ 1\big]  \underline{{\partial}_5 \partial_6 }+  {z}_{6}( 1 - {z}_{6})  \partial_6^{ 2} +  (   {b}_{4}  {z}_{5}+ {b}_{4})  {\partial}_5 \\
 & & +  \big[   (  - {b}_{1}+  {b}_{2}+  {b}_{3}+   2  {b}_{4}- 1)  {z}_{6}+   2  {b}_{1}  - {b}_{2}  -  2  {b}_{3}  -  2  {b}_{4}+ 1\big]  \partial_6  \\
 & & -  {b}_{4}  (    {b}_{3}- {b}_{1}+ {b}_{2}+ {b}_{4}) \\[5pt]
f_3&=&      -z_5(  {z}_{5} + 1)  \underline{{\partial}_5^{ 2} }-   {z}_{6}  {z}_{5}  {\partial}_5 \partial_6 +  \big[   (  - {b}_{1}+  {b}_{2}+   2  {b}_{3}+  {b}_{4}- 1)  {z}_{5}+  {b}_{1}- {b}_{4}\big]  {\partial}_5 +   {b}_{3}  {z}_{6}  \partial_6 \\ & &
  -  {b}_{3}  (    {b}_{3}- {b}_{1}+ {b}_{2}+ {b}_{4})
\end{eqnarray*}
where we underlined the leading terms according to the weight $w = (1,0)$.
By the Macaulay matrix method, we find Pfaffian matrices $P_5$ and $P_6$ written in the basis
$S = \sbrk{\pd{5}, \pd{6}, 1}\tr$.
These matrices have denominators
$  {z}_{5}  (  {z}_{5}+ 1)  (   {z}_{6}- {z}_{5}- 1) $
and
$ {z}_{6}  (  {z}_{6}- 1)  (   {z}_{6}- {z}_{5}- 1)$
respectively.
Since the Pfaffian system is logarithmic along $z_5=0$,
we may apply \algref{alg:alg2} with
$x=z_5$ and $y=z_6$.
The residue matrix at $x=0$ is
\eq{
    P_{-1}=\lrsbrk{
        \begin{array}{ccc}
            {b}_{1}- {b}_{4}&    {b}_{3}   {y}&    {b}_{3}  (    {b}_{1}- {b}_{2}- {b}_{3}- {b}_{4}) \\
            \mzero& \mzero& \mzero \\
            \mzero& \mzero& \mzero \\
        \end{array}
    }.
}
In $(x,y)$ variables, the basis reads $S = \sbrk{\pd{y}, \pd{x}, 1}\tr$.
The only non-zero element of $P_{-1} \cdot S$ is thus
\begin{equation} \label{eq:pp0}
 (  {b}_{1}- {b}_{4})  {\partial}_x +   {b}_{3}  {y}  {\partial}_y +     {b}_{3} ( {b}_{1}  -   {b}_{2}  -  {b}_{3} -  {b}_{4}  ).
\end{equation}
Next we execute \algref{alg:alg2} with $k=1, \, \gamma=2$ and the target holonomic rank $r=2$.
The operator \eqref{eq:pp0} is included in $\cJ$.
The POT Gr\"obner basis for $\cJ$ in $\CD^2$ is the set
\begin{eqnarray*}
 &\{&  [0, \   y(1  -  {y})   \partial_{{y}}^{ 2} +  (   (  - {b}_{1}+  {b}_{2}+  {b}_{3}+   2  {b}_{4}- 1)  {y}+   2  {b}_{1}  - {b}_{2}  - {b}_{3}  -  2  {b}_{4}+ 1)  \partial_{{y}} \\
 & & \ \quad +   {b}_{4} ({b}_{1}  -  {b}_{2}  -  {b}_{3}-  {b}_{4}) ]\tr , \\
 & &  [{b}_{1}- {b}_{4}, \ {b}_{3}  {y}  \partial_{{y}}
                         +  {b}_{3}({b}_{1}  -  {b}_{2}  -  {b}_{3} -  {b}_{4} )]\tr \ \}.
\end{eqnarray*}
This is a rank $2$ system when $b_1-b_4 \not= 0$, wherefore the terminating condition for the algorithm is met.

Incidentally, the $b$-function (the indicial polynomial) along $x=0$ is
$  {s} (    {s}- {b}_{1} + {b}_{4} - 1) $\asir{
{\tt import("nk\_restriction.rr"); nk\_restriction.generic\_bfct(Ideal,[z5,z6],[dz5,dz6],[1,0] | params=[b1,b2,b3,b4]); }
gives the $b$-function. }.
When $s=1$, it is equal to $-(b_1-b_4)$, which appears in the second element
of the Gr\"obner basis.
    \hfill$\blacksquare$
\end{example}

\subsection{Pfaffian-level restriction}
Here we collect technical remarks on the computational side of the
restriction method based on equations~\eqref{eq:M-def, eq:gauge-transform}
of~\secref{subsec:Deligne_extension}.
\subsubsection{Details on the \texorpdfstring{$M$}{}-matrix}
\label{sssec:M-matrix}
We start with a matrix-based argument that motivates the bottom-left
$\mZero$-block in the gauge-transformed Pfaffian
matrix~\eqref{eq:gauge-transform}.  Following the row-decomposition
of~\eqref{eq:M-def}, let us focus on the first $r^\prime$ columns of $M^{-1}$:
\vspace{-.8cm}
\eq{
    M^{-1} = \lrsbrk{
        \begin{array}{c|c}
            \> \mtop{V} \>
            &
            \> \mbot{\star} \>
        \end{array}
    }
    = \vcenter{\hbox{
        \includegraphicsbox{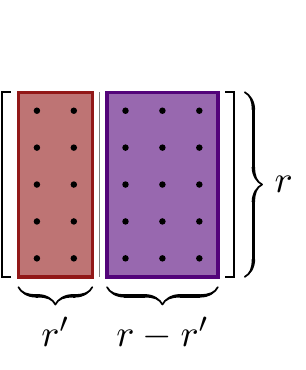}
    }}
    \, .
    \label{eq:V-def}
}
The column block $V$ satisfies the two identities
\eq{
    \mId\subsm{r \times r}
    =
    M \cdot M^{-1}
    =
    \lrsbrk{
        \begin{array}{c}
            \mtop{B} \\
            \chline
            \mbot{R}
        \end{array}
    }
    \cdot
    \lrsbrk{
        \begin{array}{c|c}
            \> \mtop{V} \>
            &
            \> \mbot{\star} \>
        \end{array}
    }
    =
            \lrsbrk{
                \begin{array}{c|c}
                    \> \mtop{B} \cdot \mtop{V} \>
                    &
                    \> \star \>
                    \\\hline
                    \> \mbot{R} \cdot \mtop{V} \>
                    &
                    \> \star \>
                \end{array}
            }
    \implies
    \begin{cases}
        B \cdot V = \mId\subsm{r^\prime \times r^\prime}
        \>,
        \\
        R \cdot V = \mZero\subsm{\brk{r - r^\prime} \times r^\prime}
        \>.
    \end{cases}
    \label{eq:V-system}
}

Next, note that since the matrix $R$ forms a basis for the row-span of
$P_{1}\brk{0, z^\prime}$, we can represent it as a product
$R = T \cdot P_{1}\brk{0, z^\prime}\>$, where $T$ is some $\brk{r - r^\prime} \times r$ matrix
\brk{of row operations}. If all non-zero rows of $P_{1}\brk{0, z^\prime}$
are already independent, we may choose $T$ as a binary matrix, for
example
\vspace{-.8cm}
\eq{
    R
    = T \cdot P_{1}\brk{0, z^\prime}
    =
    \lrsbrk{
        \begin{array}{ccccc}
            1 & \mzero & \mzero & \mzero & \mzero
            \\
            \mzero & 1 & \mzero & \mzero & \mzero
            \\
            \mzero & \mzero & \mzero & 1 & \mzero
        \end{array}
    }
    \cdot
    \vcenter{\hbox{
        \includegraphicsbox{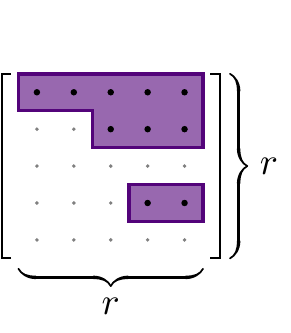}
    }}
    =
    \vcenter{\hbox{
        \includegraphicsbox{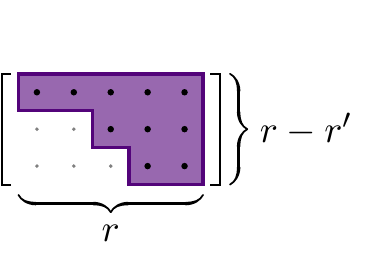}
    }}
    \>,
}
but in general $T$ might be more complicated and depend on $z^\prime$.
We plug this decomposition into the LHS of~\eqref{eq:gauge-transform} and focus
on its bottom block
\eq{
    &\bigbrk{
        \pd{i} R
        + R \cdot P_i\brk{0, z^\prime}
    } \cdot M^{-1}
    \nonumber\\
    &\quad=
    \bigbrk{
        \brk{\pd{i} T} \cdot P_{1}\brk{0, z^\prime}
            + T \cdot \sbrk{P_{i}\brk{0, z^\prime},\, P_{1}\brk{0, z^\prime}}
            + T \cdot P_{1}\brk{0, z^\prime} \cdot P_{i}\brk{0, z^\prime}
    } \cdot M^{-1}
    \nonumber\\
    &\quad=
    \bigbrk{\brk{\pd{i} T} + T \cdot P_{i}\brk{0, z^\prime}}
    \cdot
    \bigbrk{P_{1}\brk{0, z^\prime} \cdot M^{-1}}
    \>.
    \label{eq:bottom-block}
}
At this stage we recall that the second condition of~\eqref{eq:V-system}
implies
\eq{
    R \cdot V = \mZero
    \Longleftrightarrow
    T \cdot P_{1}\brk{0, z^\prime} \cdot V = \mZero
    \implies
    P_{1}\brk{0, z^\prime} \cdot V = \mZero
    \>,
}
so that the second factor in~\eqref{eq:bottom-block} indeed has a $\mZero$-block
\eq{
    P_{1}\brk{0, z^\prime} \cdot M^{-1}
    =
    P_{1}\brk{0, z^\prime} \cdot \lrsbrk{
        \begin{array}{c|c}
            \> V \>
            &
            \> \star \>
        \end{array}
    }
    =
    \lrsbrk{
        \begin{array}{c|c}
            \> \mZero \>
            &
            \> \star \>
        \end{array}
    }
    \>,
}
which leads to the $\mZero$-block in the bottom-left corner
of~\eqref{eq:gauge-transform}.

\subsubsection{An algorithmic choice of the \texorpdfstring{$B$}{}-matrix}
\label{sssec:choice-B-matrix}
Using the notation of~\secref{subsec:Deligne_extension}, it is easy to see that
one possible choice of the $r^\prime \times r$ basis matrix $B$ appearing in~\eqref{eq:M-def} is
simply the nullspace of the $r \times r$ residue matrix $P_1\brk{0, z^\prime}$
\begin{align}
    B
    \defas \nullSpace{P_1\brk{0, z^\prime}}
    = \nullSpace{R}
    \> ,
\end{align}
where $R \defas \rowReduce{P_1\brk{0, z^\prime}}$ is the matrix of independent
rows of $P_1\brk{0, z^\prime}$.
Now, recall that the restricted Pfaffian system is
actually an integrable connection on a space given as a quotient by the rows of $R$
\brk{see~\propref{prop:logarithmic_connection_as_D_module}}.
Hence, we are free to add any linear combination of rows from the $R$-matrix to simplify our choice of the
basis matrix $B$ without changing the restricted Pfaffian system.
In practice, we can do such a simplification by row reducing the auxiliary
block matrix
\begin{align}
    \RowReduce{\lrsbrk{
        \begin{array}{c|c}
            \mId_{r^\prime}   & B
            \\
            \chline
            \mZero_{\brk{r - r^\prime} \times r^\prime} & R
        \end{array}
    }}
\end{align}
and then extracting the top right $r^\prime \times r$ block.

\subsubsection{Inversion of the \texorpdfstring{$M$}{}-matrix}
\label{sssec:inv-M-matrix}
To find the restricted Pfaffian $Q_i$ in \eqref{eq:gauge-transform} it is
enough to only partially invert the $M$-matrix~\eqref{eq:M-def}.
Indeed, using the column decomposition~\eqref{eq:V-def}
to find $Q_i$ we only need to determine the $V$-block
from the system~\eqref{eq:V-system} in two steps.
First find \brk{any basis of} the right nullspace of $R$
\eq{
    \widetilde{V} \defas \nullSpace{R} \, ,
}
and then rotate it from the right to fix the normalization
\eq{
    V \defas \widetilde{V} \cdot \brk{B \cdot \widetilde{V}}^{-1}
    \ .
    \label{eq:V-def}
}
This second step will require inversion of a smaller $r^\prime \times r^\prime$
matrix $\brk{B \cdot \widetilde{V}}$, which is easier than inverting
the full $r \times r$ matrix $M$.

With the column block~\eqref{eq:V-def}, the Pfaffian $Q_i$ can be extracted from the full restriction gauge transformation as follows:
\eq{
    \bigbrk{
        \partial_i B
        + B \cdot P_i\brk{0, z^\prime}
    } \cdot V = Q_i
    \ .
}

\section{Proofs} \label{sec:proofs}

\subsection{\texorpdfstring{\thmref{thm:Gauss_Manin}}{}}
\label{app:Gauss_Manin}

Setting $\mathcal{O}(X)\defas \C[z_1,\dots,z_m,x_1^{\pm 1},\dots,x_n^{\pm 1},\frac{1}{g}]$, we define an action of $\cD_Y$ on
$f=f(z,x) \in \mathcal{O}(X)$ by
\begin{align}
        \frac{\partial}{\partial z_i}\bullet f &=
     \frac{\partial f}{\partial z_i} + \beta_0
     \left(
     {1 \over g(z;x)}
     {\partial g(z;x) \over \partial z_i}
     \right)
     f
    \ ,
     \\
    \frac{\partial}{\partial x_i}\bullet f &=
    \frac{\partial f}{\partial x_i} +
    \beta_0
    \left(
     {1 \over g(z;x)}
     {\partial g(z;x) \over \partial x_i}
    \right)f
    - \beta_i \frac{f}{x_i}
    \ .
\end{align}
The symbol
$\mathcal{O}(X)\, g^{\beta_0}x_1^{-\beta_1}\dots x_n^{-\beta_n}$ denotes the left $\cD_Y$-module $\mathcal{O}(X)$ endowed with this action.
By definition, $M_A(\beta;Y)$ is the $0$-th cohomology group of the direct image $\cD_Y$-module $\int_{\pi}\mathcal{O}(X)\, g^{\beta_0}x_1^{-\beta_1}\dots x_n^{-\beta_n}$.
The first claim is an immediate consequence of the preservation of regular holonomicity \cite[Theorem 6.1.5]{Hotta-Tanisaki-Takeuchi-2008}.

Let us prove the second claim.
Since $M_A(\beta;Y)$ is holonomic, it is an integrable connection on a generic open subset of $Y$.
We fix a generic $z\in Y$.
Then the restriction of $M_A(\beta;Y)$ to $z$ is the $n$-th twisted cohomology group on $\pi^{-1}(z)$ with its twist given by the integrand of \eqref{f_Gamma(z)}.
In view of \cite[Theorem 2.6]{Agostini:2022cgv}, its dimension is given by $|\chi(\pi^{-1}(z))|$.

\subsection{\texorpdfstring{\thmref{thm:Schur}}{}}
\label{app:Schur}

    Let $\varphi\in{\rm End}_{\cR}(\cM)$.
    Writing $\overline{\cK}$ for the algebraic closure of $\cK$, the action of $\cR$ on $\cK$ extends to that on $\overline{\cK}$.
    We can regard $\varphi$ as an element of ${\rm End}_{\overline{\cK}\otimes_{\cK}\cR}(\overline{\cM})$ where we set $\overline{\cM}:=\overline{\cK}\otimes_{\cK}\cM$.
    We first prove that any eigenvalue of $\varphi$ is in $\C$.
    Let us take an eigenvector $v\in \overline{\cM}$ of $\varphi$.
    Since $\overline{\cM}$ is finite-dimensional over $\overline{\cK}$, there exists a minimal $k$ such that ${\rm Span}_{\overline{\cK}}\{ v,\partial_1v,\ldots,\partial_1^{k}v\}={\rm Span}_{\overline{\cK}}\{ v,\partial_1v,\ldots\}$.
    It follows that there exist $a_0,a_1,\dots,a_k\in\overline{\cK}$ such that $\partial_1^{k}v=\sum_{j=0}^ka_j\partial_1^jv$.
    It is straightforward to see that
    \begin{equation}\label{eqn:compare1}
        \varphi(\partial_1^kv)=\sum_{j=0}^k\binom{k}{j}(\partial_1^{k-j}\alpha)\partial_1^jv=\sum_{j=0}^{k-1}\binom{k}{j}(\partial_1^{k-j}\alpha)\partial_1^jv+\sum_{j=0}^{k-1}a_j\alpha(\partial_1^jv)
    \end{equation}
    holds true.
    On the other hand, the fact that $\varphi$ commutes with the action of $\partial_1$ implies the following identity:
    \begin{equation}\label{eqn:compare2}
        \varphi(\partial_1^kv)=\sum_{j=0}^{k-1}a_j\varphi(\partial_1^jv)=\sum_{j=0}^{k-1}a_j\sum_{\ell=0}^j\binom{j}{\ell}(\partial_1^{j-\ell}\alpha)(\partial_1^\ell v).
    \end{equation}
    Comparing the coefficients of $\partial_1^{k-1}v$ in \eqref{eqn:compare1} and \eqref{eqn:compare2}, we obtain that $\partial_1\alpha=0$.
    Similarly, we can prove that $\partial_i\alpha=0$ for all $i=1,\dots,n$.
    Since such a function $\alpha$ must be a constant function, it must belong to $\C$.
    Now, suppose that $\dim_{\C}{\rm End}_{\cR}(\cM)\geq 2$ and take a morphism $\varphi\in {\rm End}_{\cR}(\cM)$ which is linearly independent from ${\rm id}_{\cM}$ over $\C$.
    For any eigenvalue $\alpha\in\C$ of $\varphi$, $\alpha \cdot {\rm id}_{\cM}-\varphi$ has a non-trivial kernel, which is a non-trivial $\cR$-submodule of $\cM$. This is a contradiction.

\subsection{\texorpdfstring{\thmref{thm:semi_simple}}{}}
\label{app:semi_simple}
We set $\cM:=\cR\otimes_{\cD_Y}M_A(\beta;Y)$.
Let $Y_0$ be a non-empty Zariski open subset of $Y$.
Any $\cD_{Y_0}$-submodule $N$ of $M_A(\beta;Y)|_{Y_0}$ gives rise to an $\cR$-submodule $\cM$.
Conversely, any $\cR$-submodule $\cN$ of $\cM$ is of the form $\cN=\cR\otimes_{\cD_{Y_0}}N$ for some non-empty Zariski open subset $Y_0$ of $Y$ and a $\cD_{Y_0}$-submodule $N$ of $M_A(\beta;Y)|_{Y_0}$.
Shrinking $Y_0$ appropriately if necessary, we may assume that both $M_A(\beta;Y)|_{Y_0}$ and $N$ are free as $\cO_{Y_0}$-modules, i.e., they are integrable connections.
Now, the sheaf $\cL$ of flat sections of $M_A(\beta;Y)|_{Y_0}$ admits an invariant Hermitian form.
This can be shown by adopting the argument of \cite[p. 559-560]{goto2020homology} since our exponent $\beta$ is real and generic.
Then, the sheaf $\cL_{N}$ of flat sections of $N$ admits an orthogonal complement $\cL_{N^\prime}$ which is the sheaf of flat sections of a $\cD_{Y_0}$-submodule $N^\prime$ of $M_A(\beta;Y)|_{Y_0}$ in view of \cite[Th\'eor\`eme 5.9]{Deligne}.
The direct sum decomposition $M_A(\beta;Y)|_{Y_0}=N\oplus N^\prime$ gives rise to the corresponding decomposition $\cM=\cN \oplus\cN^\prime$, which proves that $\cM$ is a semi-simple $\cR$-module.

\subsection{\texorpdfstring{\thmref{thm:Schur2}}{}}
\label{app:Schur2}
It is enough to show that $\dim_{\CC}{\rm End}_{\cR}(\cM)>1$ if $\cM$ is not simple.
By the assumption that $\cM$ is semi-simple, there exist a pair of $\cR$-submodules $\cN,\cN^\prime$ such that $\dim_{\cK}\cN,\dim_{\cK}\cN^\prime\geq 1$ and $\cM=\cN\oplus\cN^\prime$.
The identity morphism ${\rm id}_{\cN}:\cN\to\cN$ extends uniquely to an $\cR$-endomorphism $\varphi:\cM\to\cM$ by the requirement $\varphi(\cN^\prime)=0$.
In the same way, the identity morphism ${\rm id}_{\cN^\prime}:\cN^\prime\to\cN^\prime$ extends to an $\cR$-endomorphism $\varphi^\prime:\cM\to\cM$.
It is obvious from the construction that $\varphi$ and $\varphi^\prime$ are linearly independent, which shows the desired inequality $\dim_{\CC}{\rm End}_{\cR}(\cM)>1$.

\subsection{\texorpdfstring{\propref{prop:logarithmic_restriction}}{}}
\label{app:logarithmic_restriction}
    In order to simplify the discussion, we only prove the case when $m=1\>$,
    i.e. when $Y^\prime$ is a single point and $Y \simeq \Affine^1\>$. The
    case of $m \geq 2$ can be proved in the same way.

    We identify $\cD_{Y^\prime}$ with the field of complex numbers $\C$. It is
    easy to see that the following morphism is well-defined:
    \eq{
    \varphi:
        \cM^\prime=\C^r/\cN^\prime \ni \sbrk{v}
        \mapsto
        \sbrk{v} \in \cD_Y^r/(z_1\cD_Y^r+\cN)=\iota^*\cM
        \>,
    }
    where the submodule $\cN$ is defined
    in~\propref{prop:logarithmic_connection_as_D_module}.
    Let us take a constant column vector $v \in \C^r$ and consider an
    element of the form $v \, \partial_1^n\in \cD_Y^r\>$ for integer $n \ge 1$.
    A straightforward computation shows the identity
    $
        v \, \pd{1}^n \equiv
        \bigbrk{\mId_r \, n - P_1(0)\tr}^{-1}
        \cdot \sum_{k=1}^n \binom{n}{k} \bigbrk{\pd{1}^k P_1\tr|_{z_1 = 0}} \cdot v \,
        \pd{1}^{n-k}
    $
    in $\iota^*\cM$, which proves that $\varphi$ is surjective.
    Finally, we prove that $\varphi$ is injective. Let us assume that
    $\varphi([v])=\sbrk{0}$.
    Since the constant vector $v$ does not contain $\pd{1}\>$, there exist an element
    $p \in \cD_Y^r$ and a column vector $w \in \C^r$ such that
    $v = z_1 p + (\mId_r \, z_1 \pd{1} - P_1(z_1)\tr) \cdot w\>$.
    Substituting $z_1 = 0$, we obtain $v=-P_1(0)\tr \cdot w\>$, which proves that
    $\sbrk{v} = \sbrk{0}$.
\subsection{\texorpdfstring{\thmref{thm:restriction}}{}}
\label{app:restriction}

We prove the case $m=2$.
The proof can easily be generalized to the case of $m>2$.

Let $\MI(z)=\sum_{k=0}^\infty \MI_k(z_2) \> z_1^k$ be a vector of holomorphic
functions and let $P_1\>$, $P_2$ be the $r\times r$ Pfaffian matrices given by
\eqref{eqn:connection_matrix_log}.  We expand $P_1$ as shown in \eqref{eqn:expansion_of_P_1}
and $P_2$ as
\begin{equation}
    P_2(z)=\sum_{k=0}^\infty P_{2,k}(z_2) \> z_1^k
    \>.
\end{equation}
Then the integrability condition~\eqref{eqn:integrability_condition} is equivalent to
\eq{
    \pd{2}\, P_{1,-1}
    &= \bigsbrk{P_{2,0}, P_{1,-1}}
    \>,
    \nonumber\\
    \pd{2}\, P_{1,k}-(k+1)P_{2,k+1}
    &= \sum_{\substack{i+j=k \\ i\geq 0,\> j\geq -1}}
    \bigsbrk{P_{2,i}, P_{1,j}}
    \>.
\label{eqn:Int}
}
The fact that $\MI(z_1,z_2)$ is a solution of \eqref{eqn:pfaffian_system} is
equivalent to $\{ \MI_k(z_2)\}_k$ being subject to the following relations:
\begin{alignat}{3}
    \bigbrk{
        (k+1)
        - P_{1, -1}
    } \cdot \MI_{k+1}
    &= \sum_{\substack{i+j=k \\ i,\> j \geq 0}}
    P_{1,i} \cdot \MI_j
    \>,
    \quad&&\text{for $k\geq -1\>$,}
    \label{eqn:C1}
    \\
    \pd{2} \MI_k
    &= \sum_{\substack{i+j=k \\ i,\> j \geq 0}}
    P_{2,i} \cdot \MI_j
    \>,
    &&\text{for $k\geq 0\>$.}
    \label{eqn:C2}
\end{alignat}
If $\MI(z_1,z_2)$ is a solution of \eqref{eqn:pfaffian_system}, it is clear that $\MI_0(z_2)$ is a solution to the restricted Pfaffian equation
\eqref{eqn:restriction_equation}.
Conversely, let $\MI_0(z_2)$ be a solution of the restricted Pfaffian equation
\eqref{eqn:restriction_equation}.
We can uniquely extend $\MI_0(z_2)$ to the solution
$\sum_{k=0}^\infty \MI_k(z_2) \, z_1^k$ by requiring that the coefficients
$\MI_{k}\brk{z_2}$ solve the linear system~\eqref{eqn:C1}, that is:
\eq{
    \bigbrk{(k+1)-P_{1,-1}(z_2)} \cdot \MI_{k+1}\brk{z_2}
    =
    \sum_{\substack{i+j=k \\ i,\> j \geq 0}}
    P_{1,i}\brk{z_2} \cdot \MI_j\brk{z_2}
    \>,
    \label{eqn:define_I_k}
}
with $z_2$ kept as a parameter. By assumption, the residue matrix $P_{1,
-1}\brk{z_2}$ is non-resonant, so the only integer eigenvalue it has is
$0\>$, which means that the matrix on the LHS of \eqref{eqn:define_I_k} is
invertible and there is a unique $\MI_{k + 1}\brk{z_2}$ that solves it for each $k \ge -1\>$.
So by construction of \eqref{eqn:define_I_k}, the condition \eqref{eqn:C1}
is satisfied and it is left to show that the other condition~\eqref{eqn:C2} is
satisfied as well. To this end, consider the identity
\begin{equation}
    \bigbrk{\brk{k+1} - P_{1,-1}(z_2)} \cdot \partial_2 \MI_{k+1}(z_2)
    =
    \bigbrk{\brk{k+1} - P_{1,-1}(z_2)} \ \cdot
    \sum_{\substack{i+j=k+1 \\ i,\> j\geq 0}}
    P_{2,i} \cdot \MI_j(z_2)
    \>,
\end{equation}
which holds by induction in $k$ from acting with $\pd{2}$ on \eqref{eqn:define_I_k}
and the integrability condition \eqref{eqn:Int}. Due to
invertibility of the common matrix factor here, this equation implies \eqref{eqn:C2}.

Finally, we prove the convergence of $\MI(z_1,z_2)=\sum_{k=0}^\infty \MI_k(z_2)z_1^n$.
Let $f(t;z_2)=\fun{Det}{t-P_{1,-1}(z_2)}$ be the characteristic polynomial of $P_{1,-1}$.
The leading coefficient of $f(t;z_2)$ in $t$ is $1$, and thus there exists a
constant $C>0$ such that $f(t;z_2)\geq C \, t^r$ for any $t>0$ sufficiently large.
To be more precise, $z_2$ belongs to a small region in which we have a
uniform bound on the coefficients of $f(t;z_2)$.  Therefore, the matrix norm of
the cofactor matrix of $\bigbrk{t-P_{1,-1}(z_2)}$ can be estimated by $C^\prime \, t^{r-1}$
for some constant $C^\prime>0$ and for any $t>0$ sufficiently large.  Taking into
account Clamer's formula for matrix inverse, we obtain \begin{equation}
    \lrnorm{\bigbrk{t-P_{1,-1}\brk{z_2}}^{-1}}
    \leq
    C_1 \, t^{-1}
\end{equation}
for any $t>0$ sufficiently large.
Here $\lrnorm{\bigcdot}$ stands for a matrix norm.
On the other hand, convergence of the expansion~\eqref{eqn:expansion_of_P_1} and Cauchy's
inequality implies that there are positive constants $C_2$ and $s$ such that
\begin{equation}
    \lrnorm{P_{1,k}(z_2)} \leq C_2 \, s^k
    \>,
\end{equation}
for any $k\geq 0$.
We claim that the inequality
\begin{equation}\label{eqn:estimate}
    \lrnorm{\MI_k(z_2)}
    \leq
    R^n \> \lrnorm{\MI_0(z_2)}
\end{equation}
is true with $R=\max(C_1C_2,s)\>$.
Indeed, if the estimate~\eqref{eqn:estimate} is true up to $k$, then
\eq{
    \lrnorm{\MI_{k+1}(z_2)}
    &\leq
    \lrnorm{\bigbrk{\brk{k+1} - P_{1,-1}(z_2)}^{-1}} \>
    \sum_{i+j=k} \lrnorm{P_i(z_2)} \> \lrnorm{\MI_j(z_2)}
    \\
    &\leq
    C_1 \, C_2 \, \frac{1}{k+1}\sum_{i+j=k}s^i \, R^j \> \lrnorm{\MI_0(z_2)}
    \\
    &\leq R^{k+1} \> \lrnorm{\MI_0(z_2)}
    \>.
}
This proves the convergence of $\MI(z_1,z_2)\>$.
Uniqueness of $\MI$ is clear from (\ref{eqn:C1}) and the assumption of the theorem.

\subsection{\texorpdfstring{\thmref{thm:log_restriction}}{}}
\label{app:log_restriction}
We use the same notation as in \appref{app:restriction}.
Suppose we construct a matrix $G(z_1,z_2)\in GL(r;\C\{\!\{ z_1,z_2\}\!\})$ such that the gauge transformation $\MI(z_1,z_2)=G(z_1,z_2) \cdot \Mi(z_1,z_2)$ turns the system \eqref{eqn:pfaffian_system} into
\begin{align}
    \pd{1} \Mi(z_1,z_2)&=\frac{\Gamma}{z_1}\Mi(z_1,z_2)
    \>,
    \nonumber\\
    \pd{2} \Mi(z_1,z_2)&=0
    \>,
    \label{eqn:Pfaffian2}
\end{align}
for some constant matrix $\Gamma\>$. This system can be solved by the
matrix exponential: $\Mi(z_1,z_2)=z_1^\Gamma\>$.

The transformation matrix $G$ is also subject to its own Pfaffian system:
\begin{align}\label{eqn:Pfaffian3}
    \pd{1} G &= T_{P_1,\> \Gamma / z_1}(G)
    \>,
    \nonumber\\
    \pd{2} G &= P_2 \cdot G
    \>.
\end{align}
Here $T_{A,B}:M(r;\C)\to M(r;\C)$ is a linear map defined by $T_{A,B}(G) \defas A \cdot G - G \cdot B\>$.
Note that this is a Pfaffian system of the rank $r^2\>$.

\begin{lemma}
  The system \eqref{eqn:Pfaffian3} is integrable.
\end{lemma}

\begin{proof}
    For a matrix $A\in M(r;\C)$, let us denote by $T_A:M(r;\C)\to M(r;\C)$ the
    linear map given by the left multiplication with $A\>$, explicitly
    $T_A\brk{G} \defas A \cdot G\>$, so that $T_A\brk{G} = T_{A,
    \mZero}\brk{G}\>$.
    We want to prove the following integrability condition:
    \begin{equation}
        \pd{1} T_{P_2} - \pd{2} T_{P_1,\> \Gamma / z_1}
        =
        \bigsbrk{T_{P_1,\> \Gamma / z_1},\> T_{P_2}}
        \>.
    \end{equation}
    To show this, we use identities for the three building
    blocks appearing above, namely for the two derivatives
    $\pd{1} T_{P_2} = T_{\pd{1} P_2}\>$ and
    $\pd{2} T_{P_1,\> \Gamma / z_1} = T_{\pd{2} P_1}\>$,
    and for the commutator $\bigsbrk{T_{P_1,\> \Gamma / z_1},\> T_{P_2}} = T_{\sbrk{P_1,\> P_2}}\>$.
    Combining them and taking into account the integrability
    condition \eqref{eqn:pfaffian_system} proves the integrability of the
    system \eqref{eqn:Pfaffian3}.
\end{proof}

The Pfaffian system \eqref{eqn:Pfaffian 3} is logarithmic along $\{z_1=0\}$.
\thmref{thm:restriction} shows that there exists a convergent solution $G$
to \eqref{eqn:Pfaffian 3} with $G(0,0)=\mId_r\>$, since
$\mId_r\in \nullSpace{{\rm Res}\bigbrk{T_{P_1, \Gamma / z_1}}|_{z_2=0}}$
is a valid boundary condition.
Therefore, any solution to the system \eqref{eqn:pfaffian_system} can be
written in the form
\begin{equation}
    \MI(z_1,z_2)=G(z_1,z_2) \cdot z_1^\Gamma \cdot y
    =
    \sum_{n=0}^{r-1} G(z_1,z_2) \cdot \Gamma^n \cdot y \> \log^nz_1
    \label{eq:MI-Gamma}
\end{equation}
for some $y\in\C^r$, which is fixed by the boundary condition.
This proves that the boundary value map $\MI(z_1,z_2)\mapsto \MI_0(0,z_2)$ of \thmref{thm:log_restriction} is injective.
Let us prove that the boundary value map is surjective.
We observe that
\begin{equation}
    P_{1,-1}(z_2) \cdot G(0,z_2)=G(0,z_2) \cdot \Gamma
\end{equation}
is true by \eqref{eqn:Pfaffian 3}.
This identity shows that any holomorphic solution $\MI(z_2)$ of
\eqref{eqn:restriction_equation} can be written as $G(0,z_2) \cdot y\>$, where
$y\in\C^r$ is an element of the generalized $0$-eigenspace of $\Gamma\>$.
So we conclude that $\MI(z_1,z_2)=G(z_1,z_2) \cdot z_1^\Gamma \cdot y$ is the desired solution of
\eqref{eqn:pfaffian_system} with the boundary value $G(0,z_2) \cdot y$.

\begin{remark}\rm
    Finally let us comment on how the Pfaffian system~\eqref{eqn:Pfaffian3}
    can be constructed in practice.
    If we choose to flatten the $G$ matrix into a vector $\vec{G}$ of length
    $r^2$ as
    \eq{
        \vec{G} \defas \fun{Flatten}{G}
    }
    in \soft{Mathematica} notation,
    then the matrix representation of the $T_{A, B}$ map \brk{which we denote by
    the same symbol} becomes\footnote{
        In tensorial notation: $T_{A, B} = A \otimes \mId - \mId \otimes B\tr\>$.
    }
    \eq{
        T_{A, B} =
        \fun{Flatten}{
            \fun{Outer}{\code{Times,} A\code{,} \mId} - \fun{Outer}{\code{Times,}\mId\code{,} B\tr}
            \code{,\brc{\brc{1, 3}, \brc{2, 4}}}
        }
        \>.
    }
    For example, let
    $A = \lrsbrk{\vcenter{\hbox{\includegraphicsbox[scale=.5]{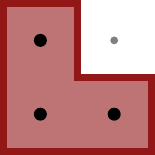}}}}$ and
    $B = \lrsbrk{\vcenter{\hbox{\includegraphicsbox[scale=.5]{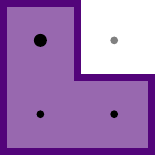}}}}$ be
    some $2 \times 2$ matrices,
    then the $4 \times 4$ matrix $T_{A, B}$ looks like the following difference:
    \eq{
        T_{
            \lrsbrk{\vcenter{\hbox{\includegraphicsbox[scale=.2]{figures/mat1}}}}
            ,
            \lrsbrk{\vcenter{\hbox{\includegraphicsbox[scale=.2]{figures/mat2}}}}
        }
        =
        \lrsbrk{\vcenter{\hbox{\includegraphicsbox[scale=.5]{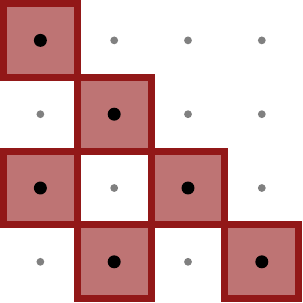}}}}
        -
        \lrsbrk{\vcenter{\hbox{\includegraphicsbox[scale=.5]{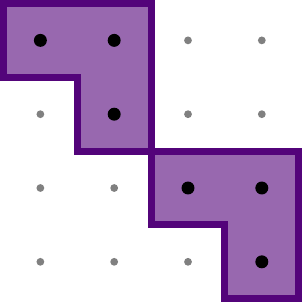}}}}
        \>.
        \label{eq:T-example}
    }
    So we conclude, that the rank $r^2$ Pfaffian system~\eqref{eqn:Pfaffian3} can be
    obtained from the original system~\eqref{eqn:pfaffian_system} by simple
    matrix operations, and, as can be seen from~\eqref{eq:T-example}, these two
    Pfaffian systems share many common properties, such sparsity and
    block-triangularity, which is important for practical computations.
\end{remark}

\subsection{\texorpdfstring{\propref{prop:prop1}}{}}
\label{appendix:proof:prop:prop1}

The map $\CD_Y/\cI \ni \ell \stackrel{\varphi}{\mapsto} \ell \in \CD_Y/(\CR_Y \cI)\cap \CD_Y $
is surjective.
Then the induced map
$\varphi \, : \, \CD_Y/(\cI+ \sum_{j=1}^{\ypcodim} z_j \CD_Y) \longrightarrow
 \CD_Y/((\CR_Y \cI)\cap \CD_Y + \sum_{j=1}^{\ypcodim} z_j \CD_Y)
$
is surjective, and the induced map for solution spaces
$$
\Hom_{\CD_Y}(\CD_Y/(\CR_Y \cI) \cap \CD_Y, \OO(U)) \stackrel{\psi}{\longrightarrow}
\Hom_{\CD_Y}(\CD_Y/\cI,\OO(U))
$$
is injective for a sufficiently small open ball $U$.
When $U$ is a simply connected domain outside of the singular locus of $\CD_Y/\cI$,
holomorphic solutions on $U$ of the two systems
$\CD_Y/\cI$ and $\CD_Y/(\CR_Y \cI) \cap \CD_Y$ agree.
Therefore, the injective map $\psi$ is an isomorphism for any sufficiently
small open ball $U$
by the identity theorem of complex function theory.
Let $U$ be a small open ball with a center on
$Y^\prime=\{z \,|\, z_1= \ldots = z_{\ypcodim}=0\}$
and suppose that the center is outside of the singular locus of
the restriction $\CD_Y/(\cI+\sum_{j=1}^{\ypcodim} z_j \CD_Y)$.
Then we have
$\Hom_{\CD_Y}(\CD_Y/\cI, \OO(U))|_V =
 \Hom_{\CD_{Y^\prime}} (\CD_Y/(\cI+\sum_{j=1}^{\ypcodim} z_j \CD_Y), \OO(U \cap V))
$.
It implies that the holonomic ranks
of
$\CD_Y/(\cI+\sum_{j=1}^{\ypcodim} z_j \CD_Y)$
and
$\CD_Y/((\CR_Y\cI)\cap \CD_Y+\sum_{j=1}^{\ypcodim} z_j \CD_Y)$
agree.
The left $\CR_{Y^\prime}$ morphism induced by $\varphi$ from
$ \CR_{Y^\prime} \otimes_{\CD_{Y^\prime}} \CD_Y/(\cI+ \sum_{j=1}^{\ypcodim} z_j \CD_Y)$
to
$ \CR_{Y^\prime} \otimes_{\CD_{Y^\prime}} \CD_Y/((\CR_Y \cI)\cap \CD_Y + \sum_{j=1}^{\ypcodim} z_j \CD_Y)$
is surjective, and the vector space dimensions of these two agree
over the field $\CC(z_{{\ypcodim}+1}, \ldots, z_{\ydim})$.
Thus, $\varphi$ is one-to-one.

\subsection{\texorpdfstring{\thmref{th:th3}}{}}
\label{appendix:proof:th:th3}
Let $r$ be the rank of the restriction.
Recall that the $(-1,0,1,0)$-order of the monomial
$\ell=x^{\alpha_1} y^{\alpha_2} \pd{x}^{\beta_1} \pd{y}^{\beta_2}$
is $-\alpha_1+\beta_1$.
We denote it by $\ord(\ell)=-\alpha_1+\beta_1$, thereby omitting the suffix $(-w,w)$.

Since the input is holonomic, there exists an operator of the form
\begin{equation}
B=b(\theta_x) + x c(x,y,\theta_x,\pd{y}),
\quad \theta_x=x \pd{x}, \
b(u) \in \CC[s], \ c \in \CD'
\end{equation}
in $\cI$ (see, e.g., \cite[Theorem 5.1.2]{SST}).
Note that $\ord(b(\theta_x))=0$ and
$\ord( x c(x,y,\theta_x,\pd{y})) < 0$.
We have
\begin{eqnarray}
\pd{x}^k B &=& b(\theta_x+k) \pd{x}^k + \pd{x}^k x c(x,y,\theta_x,\pd{y}) \nonumber \\
           &=& b(k) \pd{x}^k + (b(\theta_x+k)-b(k)) \pd{x}^k
               +\pd{x}^k x c(x,y,\theta_x,\pd{y}).  \label{eq:keyfact}
\end{eqnarray}
We have
$(b(\theta_x+k)-b(k))|_{x=0} = 0$,
and
$\ord(\pd{x}^k x c(y, \theta_x,\pd{y}) < k$.
Then the equation (\ref{eq:keyfact}) implies that
$b(k) \pd{x}^k = 0$  modulo $x\CD'+\cI+F_{k-1}$
where $F_{k-1}$ is defined in (\ref{eq:KM-filtration}).
This property, which was firstly utilized for the restriction algorithm in \cite{Oaku-1997}, is a key in proving the termination
of the algorithm
(see also \cite[Lemma 5.1.8]{SST}).

The operators $f_i$, $i=1, \ldots, N$ are generators of the left ideal $\cI$.
Suppose that the operator $B$ is expressed as
$$ B = \sum_{i=1}^N \sum_{j=0}^{m_i} d_{ij}(x,y,\pd{y})\pd{x}^j f_i \, .
$$
Let $u_0$ be the maximal non-negative integral root of $b(u)=0$.
We denote by $e_0, e_1, e_2, \ldots$ the standard basis of $\CD^{k+1}$.
In other words, $e_i$ stands for $\pd{x}^i$, and we use this identification in this proof.
Note that the last element of $v_k(\ell)$ is the component for $e_0$,
the second to last element of $v_k(\ell)$ is the component for $e_1$,
and so on.
When $k \geq \mathrm{max}(u_0+1,\mathrm{max}\, (m_i+\ord (f_i))) =:k_0$,
it follows from (\ref{eq:keyfact}) that the left module
generated by
\begin{equation} \label{eq:generators}
v_{k}(:\pd{x}^j f_i:{|_{x=0}}), \quad \ord (\pd{x}^j f_i) \leq k, \quad i=1, \ldots, N, \quad j \in \NN_0
\end{equation}
in $\CD^{k+1}$
contains an element of the form
\begin{equation}  \label{eq:bb}
 b(k') e_{k'} + \sum_{i=0}^{k'-1} d_i(y,\pd{y}) e_i, \quad k' > u_0, b(k')\not= 0 \, .
\end{equation}
Let $G_{k}$ be a Gr\"obner basis by the POT order
of
\begin{equation}  \label{eq:module-gb}
    \Bigbrc{
        v_{k}\, \bigbrk{\, :\pd{x}^j f_i:{\big|_{x=0}}}
        \Bigm|
        1\leq i \leq N, \ord (\pd{x}^j f_i) \leq k, j \in \NN_0
    }
\end{equation}
in $\CD^{k+1}$.
The element (\ref{eq:bb}) is contained in $G_{k}$
when $k \geq k_0$, which means that
there is no standard monomial of the form $e_{k'}$, $k' > u_0$
in the Gr\"obner basis.
We note that the submodules
$\CD G_{k} \cap \CD^\gamma$ represent an increasing set,
i.e.~we have
$\CD G_{k} \cap \CD^\gamma \subseteq \CD G_{k+1} \cap \CD^\gamma$.
For a sufficiently large $k$, the submodule
$\CD G_{k} \cap \CD^{\gamma}$ contains the denominator module of (\ref{eq:rest-by-gb2})
because the $v_{k+1}(:\partial{x}^\beta g_i:{|_{x=1}})$'s of $g_i$ in (\ref{eq:rest-by-gb2})
are contained in $\CD G_k$.
Note that $e_i$ for $ u_0 < i \leq \gamma$ do not appear in the standard monomials.

Recall the restriction algorithm of $\CD$-modules (see \secref{sec:NTsec1} and \cite[Section 5.2]{SST}).
The algorithm
constructs the restriction by the $(-w,w)$-Gr\"obner basis $G$.
Elements for which the $(-w,w)$-order is more than $u_0$ are not used
to construct the restriction.
Our procedure does not have this trashing procedure, so our module
$\CD G_k$ might contain these trashed elements.
This is why we need to take $\gamma$ larger than $u_1$
(see also \exref{ex:y2-x3-1-2}).

Consider two left $\CR$-submodules $\CR J$ and $\CR J'$ of $\CR^\gamma$
such that $\CR J \subset \CR J'$.
If $\dim_{\CC(y)} \CR^\gamma/(\CR J) = \dim_{\CC(y)} \CR^\gamma/(\CR J')$,
then these two $\CR$ modules are isomorphic.
Thus, we obtain the conclusion.
\qed

As final remark, note that the identity
(\ref{eq:keyfact}) is the key to proving the existence of $\gamma$.

\subsection{\texorpdfstring{\lemref{lemma:basis-lemma}}{}}
\label{appendix:basis-lemma}

Consider the holomorphic solution sheaf
$\text{Hom}_{\cD_{Y^\prime}}(\cN,\OO^{an})$ on $Y^{\prime}$.
The germs are finite-dimensional vector spaces over $\CC$;
let $r$ be the their maximal dimension.
Let $W'$ be the maximal subset of $Y^\prime$ such that
the dimension of the vector space is $r$.
Since the solution sheaf is a perverse sheaf \cite[Section 4.6]{Hotta-Tanisaki-Takeuchi-2008}, the codimension of $W'$
in $Y^\prime$ is $0$.
Let $h_1, \ldots, h_r$ be linearly independent holomorphic solutions on $W'$.
Consider the matrix $H=(s_i \bullet h_j)$
under the assumption $0 \in W'$.
Since $S=\{s_i\}$ is the basis of $\DD_Y/(\cI+\sum_{i=1}^{\ydim} z_i \DD_Y)$,
the determinant $|H|$ is not identically equal to $0$ \cite[Proposition 5.2.14]{SST}.
Let $W$ be the non-zero locus of $|H|$ in $W'$.
We may suppose that $0$ is in $W$ from the beginning.
Let $\pd{Y^\prime}$ be an element of $\cD_{Y^\prime}$.
Note that the restriction $\cN$ is regular holonomic.
Then $H^{-1} \pd{Y^\prime}\bullet H$ is a matrix consisting of regular function on $W$,
and it gives the action of $\pd{Y^\prime}$ to $S$.
\qed

\subsection{\texorpdfstring{\thmref{th:th6}}{}}
\label{appendix:proof-of-th6}

Let $\CC(Y^\prime)$ be the fraction field of
the quotient ring $\CC[Y^\prime]=\CC[z_1, \ldots, z_{\ydim}]/L$.
We note that the following identities hold,
\begin{equation} \label{eq:several-expressions-of-restriction}
 \CD_Y/(\cI + L \CD_Y)
\simeq
\CD_Y/L\CD_Y \otimes_{\CD_Y} \CD_Y/\cI
 \simeq \CC[Y^\prime] \otimes_{\CC[z]} \CD_Y/\cI, \
\end{equation}
and regular functions on $Y^\prime$ can be expressed by elements of $\CC[Y^\prime]$.
This is holonomic $\CD_{Y^\prime}$-module and is isomorphic to $\cN = \iota^* \cD_Y/\cI$.
Since the support of  $\cD_Y/(\cI+L \cD_Y)$
is $Y^\prime=V(L)$, it can also be regarded as the left $\cD_Y$-module on $Y$
supported on $Y^\prime$
by the Kashiwara equivalence (see, e.g., \cite[Section 1.6]{Hotta-Tanisaki-Takeuchi-2008})
and then $\pd{i}$, $i=1, \ldots, \ydim$ act on it.
Since $S$ is the basis of the connection $\cN(W)$
on the maximal dimensional stratum $W$,
$1 \times s_j$'s can be regaded as the basis of
$$
\CC(Y^\prime) \otimes_{\CC[z]} \CC[Y^\prime] \otimes_{\CC[z]} \CD_Y/\cI \simeq \CC(Y^\prime) \otimes_{\CC[z]} \CD_Y/\cI
$$
as a vector space over $\CC(Y^\prime)$.
Then there exists a matrix $P_i(z) \in \CC(Y^\prime)^{r \times r}$
such that
\begin{equation}
    \lrsbrk{
        \begin{array}{c}
            1 \times \pd{i}s_1 \\
            \vdots \\
            1 \times \pd{i}s_r \\
        \end{array}
    }
    = P_i(z) \cdot
    \lrsbrk{
        \begin{array}{c}
            1 \times s_1 \\
            \vdots \\
            1 \times s_r \\
        \end{array}
    }.
\end{equation}
Thus, we obtain the conclusion.
\qed

\subsection{Construction of a stratification by a comprehensive Gr\"obner basis}
\label{appendix:strata-and-comprehensive-gb}
Here we explain how the stratum $W'$, or a somewhat smaller one, can be obtained by
the comprehensive Gr\"obner basis
\cite{weispfenning-1992}.
Comprehensive Gr\"obner bases and an application to
$\DD$-module restrictions
are discussed in \cite{nakayama-takayama-2023}.
A different approach to construct strata is given in \cite{walther-2005}.

Let $\{f_1, \ldots, f_\mu\}$ be a set of generators for a holonomic ideal in $\CD_Y$.
Taking $n$ indeterminates $c_1, \ldots, c_{\ydim}$, we make a change of coordinates
\begin{equation} \label{eq:change-of-coordinates}
 z_i \rightarrow z_i+c_i, \quad i=1, \ldots, \ydim
\end{equation}
in every $f_j$.
We apply the \algref{alg:rest_to_pt} using the input
$ f_1(c,z,\pd{}), \ldots, f_\mu(c,z,\pd{})$.
We make the translation to row echelon form
or compute the Gr\"obner basis with the POT order over the rational function
field $\CC(c_1, \ldots, c_{\ydim})$.
Let $g(c)$ be the least common multiple of all the denominator polynomials needed
to obtain row echelon form.
Consider the stratum $V(L) \cap \{ z \,|\, g(z) \not= 0\}$.
The stratum is maximal dimensional.
Specializing $c$ to a number vector in the stratum
for the parametric echelon form,
we obtain the echelon form to evaluate the rank.
It follows from the assumption on $g$ that the rank is constant
on the stratum and maximal.
This is the first step in constructing the comprehensive
Gr\"obner basis.

\bibliographystyle{nb}
\bibliography{biblio}

\end{document}